\documentclass{article}
\usepackage{hyperref}
\usepackage{amsmath, amsfonts, amsthm, geometry, tikz, color, graphicx, hyperref, cleveref, thmtools, thm-restate}
\usepackage{mathrsfs}

\theoremstyle{plain}
\newtheorem{theorem}{Theorem}[section]
\newtheorem{lemma}[theorem]{Lemma}
\newtheorem{proposition}[theorem]{Proposition}

\newtheorem{claim}[theorem]{Claim}

\theoremstyle{definition}
\newtheorem{definition}[theorem]{Definition}

\theoremstyle{remark}

\newcommand{\poly}{\mathrm{poly}}
\newcommand{\negl}{\mathrm{negl}}
\newcommand{\F}{\mathbb F}
\newcommand{\R}{\mathbb R}
\newcommand{\N}{\mathbb N}

\newcommand{\E}{\mathbb E}

\newcommand{\cD}{\mathcal D}
\newcommand{\cC}{\mathcal C}
\newcommand{\xor}{\textsc{xor} }
\newcommand{\eps}{\varepsilon}
\renewcommand{\epsilon}{\varepsilon}
\newcommand{\dGV}{\delta^{(\mathrm{GV})}}
\renewcommand{\Pr}{\mathbb{P}}
\newcommand{\wt}{\mathrm{wt}}

\title{Linear time encodable binary code achieving GV bound with\\ linear time encodable dual achieving GV bound}
\author{Martijn Brehm\footnote{Informatics Institute, University of Amsterdam. \href{m.a.brehm@uva.nl}{m.a.brehm@uva.nl}.} \and Nicolas Resch\footnote{Informatics Institute, University of Amsterdam. \href{n.a.resch@uva.nl}{n.a.resch@uva.nl}. Research supported by an NWO (Dutch Research Council) grant with number C.2324.0590. This work was done in part while visiting the Simons Institute for the Theory of Computing, supported by DOE grant \#DE-SC0024124.}}
\date{\today}

\begin{document}
\maketitle
\tableofcontents
\newpage 

\begin{abstract}
    We initiate the study of what we term ``fast good codes'' with ``fast good duals.'' Specifically, we consider the task of constructing a rate 1/2 binary linear code such that both it and its \emph{dual} are asymptotically good (in fact, have rate-distance tradeoff approaching the GV bound), and are encodable in linear time. While we believe such codes should find applications more broadly, as motivation we describe how such codes can be used the secure computation task of \emph{encrypted matrix-vector product}.
    
    Our main contribution is a construction of such a fast good code with fast good dual. Our construction is inspired by the repeat multiple accumulate (RMA) code. To create the rate 1/2 code, after repeating each message coordinate, we perform \emph{accumulation steps} -- where first a uniform coordinate permutation is applied, and afterwards the prefix-sum mod 2 is applied -- which are alternated with \emph{discrete derivative steps} -- where again a uniform coordinate permutation is applied, and afterwards the previous two coordinates are summed mod 2. Importantly, these two operations are inverse of each other. In particular, the dual of the code is very similar, with the accumulation and discrete derivative steps reversed. 

    Our analysis is inspired by a prior analysis of RMA: we bound the expected number of codewords of weight below the GV bound. We face new challenges in controlling the behaviour of the discrete derivative operation (which can significantly drop the weight of a vector), which we overcome by careful case analysis.
\end{abstract}

\section{Introduction} \label{sec:intro}

The theory of error-correcting codes is largely concerned with constructing linear subspaces of finite vector spaces satisfying interesting combinatorial properties. In this work, we focus on the most popular setting of binary codes, i.e., subspaces $\cC \leq \F_2^n$. A first desirable property of a code $\cC$ is that the codewords (i.e., elements of $\cC$) are well-spread. To quantify this, one typically uses the \emph{minimum distance} $\delta(\cC)$, defined as the minimum fraction of coordinates one needs to change in order to transform one codeword into another one. In other words, $\delta(\cC):=\min\{d(c,c'):c,c' \in \cC, c \neq c'\}$, where $d(x,y):=\frac1n|\{i \in [n]:x_i\neq y_i\}|$ denotes the (relative) \emph{Hamming distance}. In fact, for linear codes, it suffices to look at the minimum distance of a nonzero codeword from $0$, i.e., $\delta(\cC) = \min\{\wt(c):c \in \cC\setminus\{0\}\}$, where $\wt(x)=\frac1n |\{i \in [n]:x_i\neq0\}|$ is the (relative) \emph{Hamming weight}. 

On the other hand, one would like the linear code to be fairly large. This is typically quantified by the code's \emph{rate}, defined as $R(\cC) = \frac kn$ where $k = \dim(\cC)$. The well-known \emph{Gilbert-Varshamov (GV)} bound states that such binary linear codes of rate $R$ and distance $\delta$ exist whenever $R \leq 1-h(\delta)$, where $h(x):=x \log\frac 1x + (1-x)\log\frac{1}{1-x}$ denotes the binary entropy function, which has a continuous inverse $h^{-1}:[0,1/2]\to[0,1]$. For binary codes, the GV bound remains the best known achievable tradeoff between rate and distance. We will be looking for codes that get close to the GV bound, as in the following definition.

\begin{definition}[GV Bound] \label{def:gv-bound}
    Fix $R \in (0,1)$ and $\eps>0$. We define $\dGV(R):=h^{-1}(1-R)$, and say a code $\cC$ of rate $R$ is \emph{$\eps$-close} to the GV bound if $\delta(\cC) \geq \dGV(R)-\eps$. 
\end{definition}

Beyond hoping for large rate and large distance, one can make other demands on a code. For example, one could hope for \emph{encoding} to be very fast. That is, recalling that for every linear code $\cC$ one can define a \emph{generator matrix} $G \in \F_2^{n \times k}$ for which $\cC = \{Gm:m \in \F_2^k\}$, one could hope that the time it takes to compute $Gm$ from $m$ is as small as possible, ideally $O(n)$.\footnote{We quantify encoding time in a circuit model, where all fan-in 2 gates are available. For additional details on this point, see \Cref{sec:prelims}.} 

Additionally, especially in the context of cryptography the \emph{dual} code of a binary linear code $\cC$ is also often crucial (being connected to, e.g., the  secrecy threshold in secret-sharing schemes). Recall that the dual of a code $\cC$ is defined as $\cC^\perp :=\{x \in \F_2^n:\langle x,c\rangle=0,~\forall c \in \cC\}$, where for $x,y \in \F_2^n$ we have $\langle x,y\rangle = \sum_{i=1}^n x_iy_i$, the standard inner-product modulo 2. Recall that if $\dim(\cC)=k$ then $\dim(\cC)=n-k$, so if $\cC$ has rate $R$ then $\cC^\perp$ has rate $1-R$. We thus find it most natural to consider codes of rate $1/2$, so that $\cC$ and $\cC^\perp$ are of the same rate (and could potentially achieve the same distance). 

In this work, we consider the task of constructing binary linear codes of rate $1/2$ meeting the following list of desiderata:
\begin{itemize}
    \item Firstly, we would like $\delta(\cC)$ to approach the GV bound $\dGV(1/2) \approx 0.1100$.
    \item We would like the dual distance $\delta(\cC^\perp)$ to also approach the GV bound $\dGV(1/2) \approx 0.1100$.
    \item Lastly, we would like both $\cC$ and $\cC^\perp$ to be encodable in $O(n)$ time (which, up to constants, is clearly optimal). 
\end{itemize}
To coin a term, we call such a code a \emph{fast good code} with \emph{fast good dual}. 

\subsection{Our results}

In this work, we prove that such a fast good code with fast good dual indeed exists. 

\begin{theorem}\label{thm:main_informal}
    For all large enough (even) $n \in \N$ and all $\eps>0$, there exists a binary code $\cC \leq \F_2^n$ of rate $1/2$ such that: 
    \begin{itemize}
        \item $\cC$ and $\cC^\perp$ both have distance at least $\dGV(1/2)-\eps$;
        \item $\cC$ and $\cC^\perp$ are both encodable in $O_\eps(n)$ time.
    \end{itemize}
\end{theorem}
As we discuss later, the encoding time is in fact practically fast. For example, with $\eps = 10^{-4}$ there is a boolean circuit with binary gates implementing the encoding map with size $8(n-1)$. Unfortunately we do not now how to bound more precisely the encoding time of $\cC$ and $\cC^\perp$ in terms of $\eps$. However, empirically the constant in front of $n$ is very small; see~\Cref{tab:deltas} and the surrounding discussion. We additionally note that Boolean circuits of size $O_\eps(n)$ and depth $O_\eps(\log n)$ can implement the encoding map, should low depth be desired.

While we will provide a much broader overview of our proof technique later, we now briefly outline the idea. Firstly, we emphasize that our construction is \emph{randomized}. Both $\cC$ and $\cC^\perp$ will be sampled in such a way that they are reminiscent of \emph{repeat multiple accumulate ($RMA$)} codes. Essentially, $RMA$ codes are defined by multiple rounds, where in each round a coordinate permutation is applied, followed by an \emph{accumulation} step, which is just the prefix-sum modulo 2. In particular, these rounds can be implemented very quickly (just $n-1$ $\xor$'s). The idea is to consider \emph{random} $RMA$ codes, where the coordinate permutations are chosen uniformly at random, and demonstrate they are likely to have very good distance.

Fortunately for us, a sequence of works~\cite{pfister1999,pfister2003capacity,bazzi,kliewer_new,ravazzi_spectra_2009,blaze} have shown these codes achieve good distance. The basic idea of these works is, given two weights $\alpha$ and $\beta$, to analyze the expected number of message vectors of weight $\alpha$ that are mapped to a codeword of weight $\beta$. Using some coding-theoretic jargon, this is called the \emph{input-output weight-enumerator function}, or \emph{IOWEF} for short. Concretely, for a fixed generator matrix $G \in \mathbb F_2^{k \times n}$ we define 
\[
    N_{G}(\alpha,\beta):=\{m \in \F_2^k:\wt(m)=\alpha \text{ and } \wt(Gm)=\beta\} \ .
\]
Then, if $\mathbf G$ denotes a random generator matrix (sampled according to some given distribution), the IOWEF will be $\mathbb E_{\mathbf G}\left[N_{\mathbf G}(\alpha,\beta)\right]$. The pleasant feature of this code ensemble is that its IOWEF has a fairly simple expression. To show a code achieves minimum distance $\delta$ with probability $1-p$ for some $p \in (0,1)$, by Markov's inequality it suffices to show that the sum of the IOWEF over all $\alpha,\beta$ with $\beta \leq \delta$ is $\leq p$. 

Now, it is not difficult to see that an $RMA$ code has constant weight dual codewords (assuming the number of rounds is constant, as we always do). Hence, they are unsuitable for our target. Nonetheless, we note that the inverse of the accumulation step is quite simple: namely, one just takes the $\xor$ of the previous two coordinates. We view this operation as a \emph{discrete derivative}. While this operation is not very helpful in terms of achieving good distance (it can at most double the weight of a vector, and in the worst case can drop it all the way down to $1/n$), as it corresponds to doing an accumulation step in the dual, it increases the dual distance! Thus, we can hope for both distance and dual distance to be large. A main challenge is to argue that these discrete derivative rounds are unlikely to \emph{harm} the minimum distance. 

Thus, our task is to bound the IOWEF of codes where accumulation rounds are interleaved with discrete derivative rounds. As in prior works~\cite{kliewer_new,ravazzi_spectra_2009,blaze}, roughly speaking we break the analysis into the case where $(\alpha,\beta)$ are small or large. The first case is handled largely by combinatorial reasoning concerning binomial coefficients; the latter case is handled by looking at the exponent of the IOWEF, i.e., looking at $\frac{1}{n}\log$ of the IOWEF and bounding it via analytic means: this function is often called the \emph{spectral shape function}. In some sense, we manage to bound the spectral shape functions of our code ensembles by that of an $RMA$ code, allowing us to obtain the same distance guarantees. This is the sense in which we manage to show the discrete derivative rounds don't harm minimum distance. In particular, we can get distance $\eps$-close to the GV bound, where the closeness to the GV bound is determined by the number of ``rounds'' of our encoding maps. 

\medskip
We believe that the existence of such a code is interesting, and likely to find further applications in coding theory, cryptography, or theoretical computer science more broadly. As a small proof of concept, we show their utility for the task of computing \emph{encrypted matrix-vector product (EMVP)}, as studied by Benhamouda et al~\cite{benhamouda2025encrypted}. The task is as follows. In an initial offline phase, a client sends an encryption $\hat M$ of a matrix $M \in \F^{m \times \ell}$ to a server, while keeping a (short) secret key. In the online phase, the client may send encryptions $\hat q$ of query vectors $q \in \F^\ell$; based on these encryptions, the server computes a value $M'$, which the client can use to determine the matrix-vector product $Mq$. Importantly, the server learns nothing about the matrix $M$ or the query vectors $q$. As Benhamouda et al~\cite{benhamouda2025encrypted} discuss, this is an important subroutine for efficient secure computation tasks in a number of domains, including encrypted fuzzy search and secure ML. Additionally, note that the special case of each $q$ being a unit vector (i.e., having a single $1$ entry) is equivalent to (single-server) private information retrieval (PIR). While we provide more details later of this application in~\Cref{sec:crypto-application}, the fast encoding time for the $\cC$ and $\cC^\perp$ we construct yield improved running times for the protocol of Benhamouda et al~\cite{benhamouda2025encrypted}.

Again, we emphasize that this is just a small proof-of-concept for our codes. We consider the main contribution of this work to be the code construction itself, along with the (fairly involved) analysis. Nonetheless, we flesh out this application further in \Cref{sec:crypto-application}.

\subsection{Related work} \label{subsec:related-work}
To the best of our knowledge, ours is the first work to construct ``fast good codes'' with ``fast good duals.'' However, our codes are heavily inspired by so-called \emph{repeat-multiple-accumulate (RMA)} codes. These codes were initially introduced as containing only a single round of permuting and accumulating by Divsalar, Jin and McEliece~\cite{divsalar1998coding}, and then extended to contain multiple rounds by Pfister and Siegel in an 1999 conference paper \cite{pfister1999}. A series of works~\cite{pfister2003capacity,kliewer_new,ravazzi_spectra_2009,blaze} (that we build upon) have managed to show that such codes can have minimum distance approaching the GV bound\footnote{There are some minor caveats concerning the speed at which one approaches the GV bound which we discuss later, but for practical purposes such codes achieve essentially this rate/distance tradeoff.} along with linear-time encoding. We will introduce these codes more formally later (\Cref{sec:RMA}) and discuss in detail how we build upon the analyses of these works. We remark that the codes in these works are, like ours, non-explicit.

Regarding the task of constructing \emph{explicit} linear-time encodable binary codes that are \emph{asymptotically good} (namely, rate and minimum distance are positive constants), Spielman~\cite{spielman1995linear} provided a construction based on expander codes; these codes can also be \emph{decoded} in linear time. Many works have built off this construction, allowing for more and more sophisticated decoding guarantees; however, this is somewhat tangential from our main focus. Note that the proven rate/distance tradeoff of Spielman's code is quite far from the GV bound. Guruswami and Indyk~\cite{guruswami2005linear} are the first to construct linear-time encodable codes achieving the GV bound, but only over sufficiently large fields (and in fact, they work in the regime where the distance can be $1-R-\eps$, i.e., roughly the Singleton bound). 

A work in the same vein as ours is by Druk and Ishai~\cite{druk2014linear}; therein, the authors give a (randomized, non-explicit) construction of a code meeting the GV bound which is linear-time encodable, and additionally outline some cryptographic applications of these codes. Finally, we mention a work by Rudra and Wootters~\cite{rudra2015ll} which, among other results, demonstrates the existence of linear-time encodable binary codes of rate $\Omega(\eps^2)$ that are list-decodable up to radius $1/2-\eps$, where we recall that a code $\cC \subseteq \F_2^n$ is called $(\rho,L)$-list-decodable if from every vector $z \in \F_2^n$ there are at most $L$ codewords $c \in \cC$ satisfying $d(c,z) \leq \rho$ (the list-size $L$ in this work is a constant depending only on $\eps$). In fact, this code is obtained by randomly folding Spielman's linear-time encodable code. 

We remark that such prior constructions of linear time encodable codes typically have sparse parity-checks; in particular, the duals cannot be asymptotically good. This is a challenge that we overcome, by carefully adding encoding steps specifically designed to improve the dual distance while also not harming the original code's distance. 

\subsection{Future directions} \label{subsec:future-directions}
Before moving onto the technical portion of our paper, we now take some time to highlight some open problems that we consider worthy of study. Firstly, our analysis is inherently limited to the binary field case. One could naturally define similar codes over larger fields: the main thing which should change is that instead of just randomly permuting, one should also randomly \emph{rescale} the entries, i.e., multiply by a uniformly random full-rank diagonal matrix.\footnote{Note that over $\F_2$, the only full-rank diagonal matrix is the identity; hence, this is a natural generalization.} However, even for the case of $RMA$ codes most analyses up till now have focused on the binary case. The simple reason for this is that, once $q>3$ the IOWEF is in fact \emph{exponentially large};\footnote{For $q=3$ as well it is unclear how to bound the IOWEF; already for this field size, new ideas appear to be required.} thus, an argument simply applying Markov's inequality cannot succeed. This is elucidated in the work of Block et al~\cite{block2024field}, where the authors managed to develop some nontrivial analysis of the minimum distance of RA codes over large fields that \emph{does not} directly use Markov's inequality. Nonetheless, the current best argument proving that RMA codes over large alphabets have nontrivial distance follows from simply analyzing the code over $\F_2$, and then using the same generator matrix to define a code over $\F_{2^t}$, i.e., some (large) extension of $\F_2$~\cite{blaze}. It would be very interesting to find an analysis over large fields that leads to improved rate/distance tradeoffs. 

Continuing with questions concerning $RMA$ codes, we now have a fairly good understanding of their minimum distance; in particular, Brehm et al~\cite{blaze} provide a fairly thorough investigation of the probability they fail to achieve good distance, demonstrating they achieve good distance for concrete values of $n$. A next task could be to determine what \emph{list-decodability} they might possess. A reasonable approach for this could be to show that they are \emph{locally similar to random linear codes}~\cite{guruswami2022punctured,mosheiff2024randomness}, which would (roughly) imply that they inherit the list-decodability of a typical linear code.\footnote{Technically, $RMA$ will not be fully locally similar to random linear codes, as they contain constant-weight codewords with an overly large probability $1/\poly(n)$. The same issue arises with LDPC codes, but can be resolved~\cite{mosheiff2020ldpc}.} It would additionally be very interesting to find \emph{explicit} $RMA$ codes achieving nontrivial distance; to the best of our knowledge the only result in this vein is due to Guruswami and Machmouchi~\cite{guruswami2008explicit}, which essentially partially derandomizes the construction of an asymptotically good $RAA$ code (namely, it requires only the second permutation to be sampled randomly; the first can be deterministically chosen). Finally, we remark that every question asked above could additionally be asked for the codes we define and study. 

Next, we list a couple interesting questions motivated by cryptographic applications. Specifically, we consider the code property of \emph{multiplicativity}. Given two vectors $x, y \in \F^n$, we define their \emph{(Schur) product} as $x*y:=(x_1\cdot y_1,\dots,x_n \cdot y_n)$ as the component-wise product, and then given two linear codes $\cC$ and $\cD$ their (Schur) product is defined as $\cC * \cD := \mathrm{span}\{c*d:c \in \cC,d \in \cD\}$. A code pair $(\cC,\cD)$ is said to be \emph{multiplicative} if $\cC*\cD \neq \F^n$; namely, it does not span the whole space. Firstly, it can be observed that a binary linear code $\cC \leq \F_2^n$ and its dual $\cC^\perp$ are multiplicative, as $\cC*\cC^\perp$ only contains even weight vectors. It is known~\cite{cramer2000general} that multiplicative codes yield multiplicative secret-sharing schemes,\footnote{Briefly, secret-sharing schemes allow a dealer holding to a secret $s$ to generate shares $(s_1,\dots,s_n)$ to be given to $n$ parties such that \emph{authorized} sets of parties can reconstruct the secret, but no \emph{unauthorized set} of parties can. As for multiplicativity, this can be viewed as an additional property allowing parties to locally convert shares of secrets $a,b$ into a new value $c_i$ such that the product $ab$ is a linear combination of the $c_i$'s.} and additionally that such a multiplicative secret-sharing scheme can be used for general-purpose secure multiparty computation (MPC)~\cite{ben2019completeness,yao1986generate}. The encoding efficiency of the codes $\cC$ and $\cC^\perp$ directly translate into the complexity of sharing a secret~\cite{druk2014linear}, which must be done by each party for each multiplication gate of the circuit representing the function one is securely computing. Thus, in the quest for ``constant-overhead cryptography~\cite{ishai2008cryptography}'' (where the hope is that, if the cost of computing a functionality \emph{without} security is $N$, the cost of computing it \emph{with} security is just $O(N)$) multiplicative secret-sharing scheme with linear-time encoding can play an important role. We observe that typical multiplicative secret-sharing schemes are ``algebraic'' (e.g., Shamir secret-sharing~\cite{shamir1979share}, which is based on polynomial evaluations); here, one can deal the shares in \emph{quasilinear} time (using fast Fourier-transform-type ideas), but getting linear-time dealing appears to require more ``combinatorial'' approaches. 

To summarize the above, our work provides a multiplicative secret sharing scheme over $\F_2$ with linear-time encoding. One could hope for more. Firstly, one can hope for higher-degree multiplicative secret-sharing~\cite{barkol2010d} allowing for multiplying more secrets together, again with linear-time algorithms for dealing the shares. Next, one could additionally hope for \emph{strong} multiplicativity of a secret-sharing scheme, which informally would follow from a pair of asymptotically good codes $\cC,\cD$ such that $\cC *\cD$ has \emph{constant} distance (such codes are said to satisfy the \emph{multiplication property}). Note that for a linear code $\cC$ and its dual $\cC^\perp$, as $\cC*\cC^\perp$ only contains even weight vectors, it has minimum distance $\geq 2/n$, which is nontrivial, but far from constant. Strongly multiplicative secret-sharing schemes lead to general MPC protocols with \emph{malicious} security~\cite{cramer2000general}, i.e., security even in the presence of parties that may actively deviate from the prescribed protocol. One could hope to have such multiplication codes with linear time encoding towards getting general-purpose MPC with constant-computational overhead and malicious security. We remark that prior work has already considered the possibility of constructing somewhat ``combinatorial'' codes (i.e., codes that could admit linear time encoding) that additionally have the multiplication property~\cite{dinur2023new}. 

Broadly speaking, we view our work as a stepping stone towards a broader suite of ``fast codes'' with interesting combinatorial properties. We emphasize that from a cryptographic perspective, decodability is often not required,\footnote{And in fact, in certain cases the security of the scheme rests on the assumed \emph{hardness} of decoding; this is (essentially) the case for the encrypted matrix-vector product protocol we sketch in \Cref{sec:crypto-application}.} but instead ask for nontrivial conditions concerning, e.g., the dual code or multiplicativity. 

\subsection{Organization} \label{subsec:organization}
In \Cref{sec:prelims} we set notation
The remainder of this paper is organized as follows. Firstly, we set notation, define RMA codes, review their existing analysis, and we define our pair of dual codes in \Cref{sec:prelims}. Next, we provide a broad overview of our proof strategy for \Cref{thm:main_informal} in \Cref{sec:proof-strategy}. The bulk of the proof is contained in \Cref{sec:middle-weight,sec:bounding-edge}. Next, we discuss the utility of these codes in the context of the encrypted matrix-vector product problem~\Cref{sec:crypto-application}. 

\subsection{Acknowledgements}
The authors would like to thank Yuval Ishai, for patient discussions concerning potential cryptographic applications of our codes; Geoffroy Couteau, for sharing the construction of the codes that we analyze; Divya Ravi, for kindly reading over the cryptographic sections and suggesting relevant citations; the anonymous reviewers for helpful comments; and Henri Pfister for suggesting overlooked citations on RMA codes.

\section{Preliminaries} \label{sec:prelims}
    Throughout this work we will use Latin letters $a,b,c,\dots$ to refer to absolute weights (integers from $1$ to $n$) and Greek letters $\alpha,\beta,\gamma$ to refer to relative weights (reals in $[0,1]$). Except for in \Cref{sec:bounding-edge}, weight will by default refer to \emph{relative} weight (i.e., the fraction of nonzero coordinates); in \Cref{sec:bounding-edge}, it is most convenient to have it default to \emph{absolute} weight (i.e., the number of nonzero coordinates). By default, $\log$ denotes the base-2 logarithm. 

    We use standard Landau notation, e.g., $O(\cdot)$, $\Omega(\cdot)$, $o(n)$, $\omega(n)$, etc. In all contexts the growing parameter will be $n$. We use $\poly(n)$ to refer to a function of the form $n^{O(1)}$ and $\negl(n)$ a function of the form $n^{-\omega(n)}$. The notation $\tilde{O}(\cdot)$ suppresses terms of the form $\log^{O(1)}n$. 
    
    Next, we recall Markov's inequality, and the specific way in which we use it in this work. 

    \begin{theorem} [Markov's inequality] \label{thm:markov-ineq}
        Let $X$ be a non-negative random variable. Then for any $\alpha>0$, 
        \[
            \Pr\big[X \geq \alpha\big] \leq \frac{\E[X]}{\alpha} \ .
        \]
        In particular, if $X$ takes values in $\N\cup\{0\}$ then 
        \[
            \Pr\big[X > 0\big]\leq \E[X] \ .
        \]
    \end{theorem}
    Now, recall that we defined the \emph{dual} of a linear code $\cC \leq \F_2^n$ of dimension $k$ as 
    \[
        \cC^\perp := \{x \in \F_2^n:\forall c \in \cC,~\langle x,c\rangle = 0\} \ .
    \]
    We recall the following standard facts. Firstly, $\dim(\cC^\perp) = n-k$ and $(\cC^\perp)^\perp=\cC$. Additionally, if linear codes $\cC$ and $\cD$ are generated by matrices $G \in \F_2^{n \times k}$ and $H\in \F_2^{n \times (n-k)}$, respectively, in the sense that $\cC = \{Gm:m \in \F_2^k\}$ and $\cD = \{Hm:m \in \F_2^{n-k}\}$, then $\cC^\perp = \cD$ if and only if $H^\top G = 0$.

    Lastly, as we wish to discuss encoding complexity of binary linear codes, we choose the standard model of Boolean circuits where gates of fan-in 2 are available.\footnote{Of course, any other constant fan-in only affects the encoding complexities up to constants; the choice of 2 is just for concreteness.} We remark that this circuit model is robust to the precise set of available gates, but for simplicity we assume all gates are available. This should be contrasted with more liberal models such as arithmetic circuits or RAM models, which are more sensitive to ring or word size, respectively. Additional discussion of this model choice can be found in~\cite{spielman1995linear}. 

    \subsection{Repeat-multiple-accumulate $RMA$ codes} \label{sec:RMA}
        \paragraph{Definition}
        The codes we study in this paper are directly inspired by the well-studied \textit{repeat-multiple-accumulate (RMA) codes} over the binary field $\F_2$. These are parametrised by three integers $m,n,r\in\mathbb{N}$, where $m$ is the number of \textit{rounds} of the code, $r$ is the repetition factor which fixes the rate $1/r$ of the code, and finally $n$ is the \textit{block-length}. We assume $n$ is divisible by $r$. 

        The encoding function of an $RMA$ code is simple to describe. First, take your message vector of length $n/r$ and repeat it $r$ times. Then apply a permutation to the resulting length $n$ vector. Next, apply the accumulator operation to this vector, which is simply a prefix sum: bit $i$ of the output will be equal to the \xor of the first $i$ bits in the input vector. We repeat this permute-accumulate step $m$ times. More formally, the three operations in the encoding can be described in terms of the following linear operations.
        \begin{itemize}
            \item For a constant $r \in \N$ dividing $n$, let $F_r \in \F_2^{n \times n/r}$ denote the \emph{repetition matrix} corresponding to the linear operator that repeats each entry in the vector $r$ times. That is, $F_r[i,j] = 1$ if and only if $\lfloor i/r\rfloor+1=j$. 
            \item Let $A \in \F_2^{n \times n}$ be the \emph{accumulator matrix}, $A_{ij}=1$ if and only if $i \geq j$. 
            \item For a permutation $\pi : [n] \to [n]$, let $M_\pi \in \F_2^{n \times n}$ be the \emph{permutation matrix} corresponding to $\pi$. That is, $(M_\pi)_{ij}=1$ if and only if $\pi(i)=j$.
        \end{itemize}
        
        Say we fix $m=2$ and have $r\in\mathbb{N}$. We then obtain a rate $1/r$ $RMA$ code with two permute-accumulate rounds, which we refer to as an $RAA$ code. Its generator matrix $RAA_{\pi_1,\pi_2} : \F_2^{n/r} \to \F_2^n$ code can be defined as
        $RAA_{\pi_1,\pi_2}(x) = A  M_{\pi_2}  A M_{\pi_1}  F_r \, x$. More generally, we will refer to an $RMA$ code of $m$ rounds as an $RA^m$ code.
        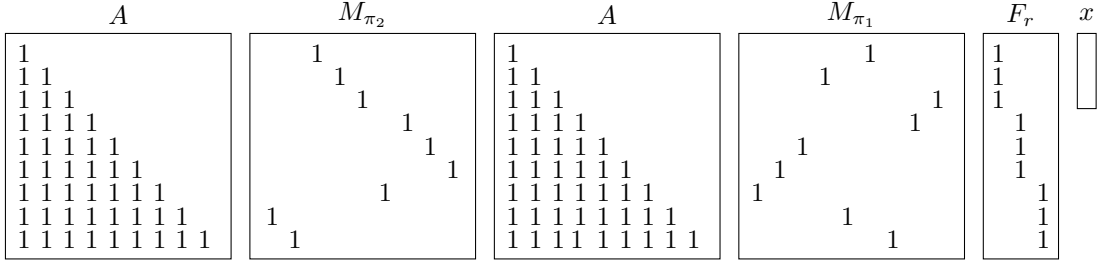
\begin{figure}[ht]
            \centering
            \begin{tikzpicture}
                \draw (7.25,6) rectangle (7.5,5); 
                \node at (7.375,6.25) {$x$};
        
                \draw (6,6) rectangle (7,3);
                \node at (6.5,6.25) {$F_r$};
                
                \node at (6.2,5.74) {$1$};
                \node at (6.2,5.43) {$1$};
                \node at (6.2,5.12) {$1$};
                \node at (6.5,4.81) {$1$};
                \node at (6.5,4.5) {$1$};
                \node at (6.5,4.19) {$1$};
                \node at (6.8,3.88) {$1$};
                \node at (6.8,3.57) {$1$};
                \node at (6.8,3.26) {$1$};
                
                \draw (2.75,6) rectangle (5.75,3);
                \node at (4.25,6.25) {$M_{\pi_1}$};
        
                \node at (3,3.88) {$1$};
        
                \node at (3.3,4.19) {$1$};
               
                \node at (3.6,4.5) {$1$};
        
                \node at (3.9,5.43) {$1$};
                
                \node at (4.2,3.57) {$1$};
        
                \node at (4.5,5.74) {$1$};
                
                \node at (4.8,3.26) {$1$};
                
                \node at (5.1,4.81) {$1$};
        
                \node at (5.4,5.12) {$1$};
                
                \draw (-0.5,6) rectangle (2.5,3);
                \node at (1,6.25) {$A$};
        
                \node at (-0.25,5.74) {$1$};
                \node at (-0.25,5.43) {$1$};
                \node at (-0.25,5.12) {$1$};
                \node at (-0.25,4.81) {$1$};
                \node at (-0.25,4.5) {$1$};
                \node at (-0.25,4.19) {$1$};
                \node at (-0.25,3.88) {$1$};
                \node at (-0.25,3.57) {$1$};
                \node at (-0.25,3.26) {$1$};
        
                \node at (0.05,5.43) {$1$};
                \node at (0.05,5.12) {$1$};
                \node at (0.05,4.81) {$1$};
                \node at (0.05,4.5) {$1$};
                \node at (0.05,4.19) {$1$};
                \node at (0.05,3.88) {$1$};
                \node at (0.05,3.57) {$1$};
                \node at (0.05,3.26) {$1$};
        
                \node at (0.35,5.12) {$1$};
                \node at (0.35,4.81) {$1$};
                \node at (0.35,4.5) {$1$};
                \node at (0.35,4.19) {$1$};
                \node at (0.35,3.88) {$1$};
                \node at (0.35,3.57) {$1$};
                \node at (0.35,3.26) {$1$};
        
                \node at (0.65,4.81) {$1$};
                \node at (0.65,4.5) {$1$};
                \node at (0.65,4.19) {$1$};
                \node at (0.65,3.88) {$1$};
                \node at (0.65,3.57) {$1$};
                \node at (0.6,3.26) {$1$};
        
                \node at (0.95,4.5) {$1$};
                \node at (0.95,4.19) {$1$};
                \node at (0.95,3.88) {$1$};
                \node at (0.95,3.57) {$1$};
                \node at (0.95,3.26) {$1$};
        
                \node at (1.25,4.19) {$1$};
                \node at (1.25,3.88) {$1$};
                \node at (1.25,3.57) {$1$};
                \node at (1.25,3.26) {$1$};
        
                \node at (1.55,3.88) {$1$};
                \node at (1.55,3.57) {$1$};
                \node at (1.55,3.26) {$1$};
                
                \node at (1.85,3.57) {$1$};
                \node at (1.85,3.26) {$1$};
        
                \node at (2.15,3.26) {$1$};
                
                \draw (-3.75,6) rectangle (-0.75,3);
                \node at (-2.25,6.25) {$M_{\pi_2}$};
        
                \node at (-3.45,3.57) {$1$};
        
                \node at (-3.15,3.26) {$1$};
               
                \node at (-2.85,5.74) {$1$};
        
                \node at (-2.55,5.43) {$1$};
                
                \node at (-2.25,5.12) {$1$};
        
                \node at (-1.95,3.88) {$1$};
                
                \node at (-1.65,4.81) {$1$};
                
                \node at (-1.35,4.5) {$1$};
        
                \node at (-1.05,4.19) {$1$};
                
                \draw (-7,6) rectangle (-4,3);
                \node at (-5.5,6.25) {$A$};
        
                \node at (-6.75,5.74) {$1$};
                \node at (-6.75,5.43) {$1$};
                \node at (-6.75,5.12) {$1$};
                \node at (-6.75,4.81) {$1$};
                \node at (-6.75,4.5) {$1$};
                \node at (-6.75,4.19) {$1$};
                \node at (-6.75,3.88) {$1$};
                \node at (-6.75,3.57) {$1$};
                \node at (-6.75,3.26) {$1$};
        
                \node at (-6.45,5.43) {$1$};
                \node at (-6.45,5.12) {$1$};
                \node at (-6.45,4.81) {$1$};
                \node at (-6.45,4.5) {$1$};
                \node at (-6.45,4.19) {$1$};
                \node at (-6.45,3.88) {$1$};
                \node at (-6.45,3.57) {$1$};
                \node at (-6.45,3.26) {$1$};
        
                \node at (-6.15,5.12) {$1$};
                \node at (-6.15,4.81) {$1$};
                \node at (-6.15,4.5) {$1$};
                \node at (-6.15,4.19) {$1$};
                \node at (-6.15,3.88) {$1$};
                \node at (-6.15,3.57) {$1$};
                \node at (-6.15,3.26) {$1$};
        
                \node at (-5.85,4.81) {$1$};
                \node at (-5.85,4.5) {$1$};
                \node at (-5.85,4.19) {$1$};
                \node at (-5.85,3.88) {$1$};
                \node at (-5.85,3.57) {$1$};
                \node at (-5.85,3.26) {$1$};
        
                \node at (-5.55,4.5) {$1$};
                \node at (-5.55,4.19) {$1$};
                \node at (-5.55,3.88) {$1$};
                \node at (-5.55,3.57) {$1$};
                \node at (-5.55,3.26) {$1$};
        
                \node at (-5.25,4.19) {$1$};
                \node at (-5.25,3.88) {$1$};
                \node at (-5.25,3.57) {$1$};
                \node at (-5.25,3.26) {$1$};
        
                \node at (-4.95,3.88) {$1$};
                \node at (-4.95,3.57) {$1$};
                \node at (-4.95,3.26) {$1$};
                
                \node at (-4.65,3.57) {$1$};
                \node at (-4.65,3.26) {$1$};
        
                \node at (-4.35,3.26) {$1$};
            \end{tikzpicture}
            \caption{A pictorial representation of (a generator matrix for) an RAA code with rate $R=1/3$ and block-length $n=9$.}
            \label{fig:RAA-code}
        \end{figure}

        Note that computing the $RA^m$ encoding can be done by a Boolean circuit consisting of just $m(n-1)$ $\xor$ gates of fan-in 2; additionally, classical work of Ladner and Fischer~\cite{ladner1980parallel} gives an implementation with $O(n)$ gates and depth $O(\log n)$ (and the constants are very reasonable). 

        Rather than considering a fixed choice of $RA^m$ code, we will consider a \emph{random} $RA^m$ codes defined by sampling uniformly at random the $m$ permutations making up the code. Thus, in the analysis, for a fixed vector weight one can imagine that going into each accumulator round we in fact have a \emph{uniformly random} vector of that weight. To establish that a random $RA^m$ code achieves good distance, the important fact is that a uniformly random vector of \emph{any} weight has expected weight $1/2$ after accumulating. Thus, a random $RA^m$ code typically ``pushes the weights towards the middle,'' which is exactly what we want when we recall that, for a linear code, achieving minimum distance $\delta$ is equivalent to having no codewords of weight below weight $\delta$.

        \paragraph{Distance analysis}
        As the analysis of the minimum distance of our codes is inspired by prior analyses of $RA^m$ codes, we now sketch these approaches. As mentioned in the introduction, $RA$ codes (with only a single round of permute and accumulate) were introduced by Divsalar, Jin and McEliece~\cite{divsalar1998coding}. Beyond the efficiency of the encoding, part of the appeal of these codes is the fairly simple expression of their IOWEF: the expression telling you the expected number of message vectors of a given weight $a$ that are mapped to a codeword of weight $b$. This expression is very useful because, as noted in the introduction, to show that a code achieves relative minimum distance $\delta$ with probability $1-p$ for some $p\in(0,1)$, by Markov's inequality it suffices to show that the expected number of codewords of relative weight at most $\delta$ (which is simply the sum of the IOWEF for absolute codeword weights $b \leq \delta n$). When introducing $RA$ codes, the authors indeed show that we can express the probability that a weight $a$ vector is mapped to weight $b$ after permuting and accumulation as follows (where we rewrite the binomials as this will be useful in the later analysis):
        $$p_A(a,b) = \frac{\binom{n-b}{\lfloor a/2\rfloor}\binom{b-1}{\lceil a/2 \rceil-1}}{\binom{n}{a}} =  \frac{\binom{n-b}{\lfloor a/2\rfloor}\binom{b}{\lceil a/2 \rceil}}{\binom{n}{a}}\frac{\lceil a/2 \rceil}{b} = \frac{\binom{n-a}{b-\lceil a/2\rceil}\binom{a}{\lceil a/2 \rceil}}{\binom{n}{b}}\frac{\lceil a/2 \rceil}{b}\ .$$
        % where the latter equalities follow by elementary manipulations of the binomial coefficients.
         With $p_A(a,b)$ in hand, one can easily express the expectation that we are interested in.\footnote{Note also that a glance at $p_A(a,b)$ reveals the previous claim that the expected outcome is roughly $b=n/2$.} For example, consider a random $RA$ code of rate $1/r$. After repeating a message of absolute weight $a$, there are $\binom{n/r}{a}$ vectors of absolute weight $ra$. We then expect $p_A(ra,b)$ of these vectors to be mapped to absolute weight $b$ after applying the first permute-accumulate operation. Now sum over all permitted values of $a$ and $b$, but restrict $b \leq \delta n$. 
        
        Using this, Kahale and Urbanke proved that these single-round codes are not asymptotically good \cite{kahale_minimum_1998}. This led Pfister and Siegel, in a 1999 conference paper, to first study $RA^m$ codes with multiple rounds of permute and accumulate \cite{pfister1999}. They numerically estimate the expected number of low weight codewords, which suggests that with at least two rounds, $RA^m$ codes achieve asymptotic goodness, and in fact distance relatively close to the GV-bound. This work was extended into chapter 3 of Pfister's dissertation \cite{pfister2003capacity}, which contains the first formal proof of the asymptotic goodness of $RA^m$ codes with at least two rounds, as well as precise numerical estimates of the minimum distance attained by $RA^m$ codes of varying rates and number of rounds (similar to our \cref{tab:deltas}). Importantly, the provably attained minimum distance is not equal to these numerical estimates, and is in fact far below the GV-bound. Thus, this does not prove distance close to the GV-bound yet.
        
        It takes a few more years before another group of authors formally prove that $RA^m$ codes obtain minimum distance according to these numerical estimates \cite{kliewer_new}. They achieve this by taking a more analytic approach. Instead of analyzing the IOWEF directly (a complicated expression with many binomial coefficients), they use Stirling's approximation,\footnote{There are some subtleties regarding floors and ceiling, which can be addressed, as we do in \Cref{lem:A_D_claims}.} to say that $p_A(a,b) \leq O(1) \cdot 2^{n \cdot f_A(a/n,b/n)}$, where
        \[
            f_A(\alpha,\beta) = \alpha-h(\beta) + (1-\alpha)h\left(\frac{\beta-\alpha/2}{1-\alpha}\right) \ .
        \] 
        This, along with the standard bound $\binom{n/r}{a} \leq 2^{n \cdot h(a/n)/r}$, gives a fairly natural representation for the exponent of the terms appearing in the IOWEF, which is called the \emph{(asymptotic) spectral shape function} (which we formally define below, see \Cref{def:spec-shape-RMA}). This spectral shape function admits an analysis of its critical points, allowing us to maximize over the non-final weights. This reveals a tradeoff between the repetition factor $r$ (which recall determines the rate as $R=1/r$), and the target minimum distance $\delta$, reminiscent of the tradeoff appearing in the proof of the GV bound (although the expression involved is more complicated). Using this method, they reproduce Pfister's numerical estimates of the minimum distance of $RA^m$ codes, but now proving this distance can be achieved.

        This work, however, only proves that the expectation (and hence probability of attaining this minimum distance quite close to the GV-bound) goes to 1 with the block length $n$; it does not explicate how fast this happens. Building on this work, Ravazzi and Fagnani first prove that this probability has the form $1-1/\poly(n)$ and specify what this polynomial is. Moreover, they prove that the minimum distance converges to the GV-bound as the number of rounds goes to infinity (this was only conjectured before based on numerical estimates) \cite{ravazzi_spectra_2009}. Many years later, another group of authors characterise the expectation more precisely, showing for the first time that we can expect $RA^m$ codes to have good distance for concrete block lengths $n$ \cite{blaze}.

        Let us now formally define the spectral shape function, as we will require it in the analysis of our codes. 
        
        \begin{definition} \label{def:spec-shape-RMA}
            Let $r,m\in\mathbb{N}$. The \textit{asymptotic spectral shape function} of a rate $1/r$ $RA^m$ code is defined recursively as
            $$\hat{r}^{(m)}_{A,r}(\gamma) := \max\limits_{0\leq \alpha \leq 1} \{\hat{r}_{A,r}^{(m-1)}(\alpha) + f_A(\alpha,\gamma)\} \ \ \text{ and }\ \ \hat{r}^{(1)}_{A,r}(\gamma)=h(\gamma)/r \ ,$$
            and we use it to define 
            $$\delta^{(m)}_{A,r}:=\max\{\delta \in [0,1/2) : \forall \epsilon \leq \delta, \hat{r}^{(m)}_A(\epsilon) = 0 \} \ ,$$
            where we will use the shorthands $\hat{r}^{(m)}_{A}(\gamma):=\hat{r}^{(m)}_{A,2}(\gamma) $ and $\delta^{(m)}:=\delta^{(m)}_{A,2}$, as we will only consider codes of rate $1/2$ throughout.
        \end{definition}
        
        The spectral shape function is thus flat at 0 for $\gamma \leq \delta^{(m)}_A(r)$, where it was proven that $\delta^{(m)}_A(r)$ converges to the GV-bound \cite{ravazzi_spectra_2009}. After this point, the spectral shape function grows and becomes strictly positive. This implies that we expect exponentially many codewords of weight $\gamma > \delta^{(m)}$ in a rate $1/r$ $RA^m$ code, showing that one cannot hope to achieve distance beyond $\delta^{(m)}$ in an $RA^m$ code.\footnote{At least, not when using Markov's inequality; a stronger inequality may be able to prove better distance, but this seems unlikely to us.} When the spectral shape function is equal to 0, for $\gamma \leq \delta^{(m)}$, we do not expect exponentially many vectors. But observe we can't conclude there are $0$ such vectors: just subexponentially many. Despite this, previous authors established that we expect few \cite{kliewer_new}, and eventually more precisely only $1/\poly(n)$ \cite{ravazzi_spectra_2009,blaze}, such vectors. How did they prove this?
        
        To answer this question, we need to ask why the spectral shape function is flat at 0 for $\gamma \leq \delta^{(m)}_A(r)$ in the first place. This happens because for such small $\gamma$, the expression over $\alpha$ (recall that we maximize over $\alpha$) is in fact decreasing with $\alpha$. To maximize, we are then forced to fix $\alpha=0$, which yields the function $\hat{r}^{(m-1)}_A(0)+f_A(0,\gamma)=0$. Suppose now that one has a non-zero lower bound on $\alpha$, say $h/n\leq \alpha$. In that case, one would be forced to pick $\alpha=h/n$ to maximize and this would yield a strictly negative value for the spectral shape function. One can use this negative value of the spectral shape function to argue that there will only be about $\exp(-h)$ codewords of weight at most $\delta^{(m)}_A(r)$ in expectation.
        
        Why can we assume we have a lower bound on $\alpha$? Well, $\alpha$ described the weight of the codeword preceding the final round of permute-accumulate. Hence, if we can argue that only very few (say, some inverse polynomial number) vectors have relative weight below $h/n$ before entering this last round of the encoding, then we can consider that case dealt with, and assume for the rest of the argument that $h/n \leq \alpha$. 
        
        We now have two conflicting forces. We need to pick $h$ large enough so that $\exp(-h)$ is decreasing in $n$, but also pick $h$ small enough so that we can realistically argue that there are only few vectors of weight at most $h$ after $m-1$ rounds of permute-accumulate. A sensible assignment is $h=\log^2n$, so that $\exp(-h)=n^{-\log n}$ is negligible in $n$, but turns out to be small enough to argue that the expected number of vectors of weight at most $\log^2 n$ after $m-1$ rounds is $1/\poly(n)$. We will use this same general proof strategy to analyse our codes: split the computation of the expectation based on whether the weight of the codeword before going into the final round of the encoding is at least $h/n=\log^2n/n$ or not.

        \paragraph{Convergence of distance to GV-bound}
        Let us end this section by making a note about the convergence of $\delta^{(m)}_A(r)$ to the GV-bound with $m$, which as we said was proven by \cite{ravazzi_spectra_2009}. The closer we wish to approximate $\dGV$, the higher $m$ needs to be. Unfortunately, we lack a simple analytical expression for $\delta^{(m)}$, and thus have no theoretical understanding of how fast $\delta^{(m)}_A(r)$ converges to the GV-bound with $m$. Therefore, if one wishes to get $\epsilon=1/10000$-close to the GV-bound, we know this is possible for large enough $m$, but whether this is already true for $m=3$ or requires $m=10$, or even $m=100$, is not something we can theoretically argue. 

        What we are left with is numerically estimating $\delta^{(m)}$. To do so, we need to estimate the spectral shape function. Suppose you have a rate $1/r$ $RA^m$ code where we describe the weights in between the rounds using $\alpha_1, \alpha_2, ..., \alpha_{m+1}$, and the non-maximized asymptotic spectral shape function is described by $f_{A^m}(\alpha_1,...,\alpha_{m+1})$. To determine the value of the spectral shape function, we need to maximize over the $\alpha_i$'s. Pfister first did this by a grid search over the domains of these variables \cite{pfister2003capacity}. A more complicated approach was taken by \cite{kliewer_new}. They gave the constraints
        $$\frac{\partial f_{A^m}(\alpha_1,...,\alpha_{m+1})}{\partial \alpha_1} = 0 \iff \alpha_2 = \frac{1}{2}\left(1 \pm (1-\alpha_1)\sqrt{1-\left(\frac{\alpha_1}{1-\alpha_1}\right)^{2/r}}\right) \ ,$$
        and 
        $$\frac{\partial f_{A^m}(\alpha_1,...,\alpha_{m+1})}{\partial \alpha_i} = 0 \iff \alpha_{i+1} = \frac{1}{2} \pm \frac{1-\alpha_i}{2}\sqrt{1 - \left(\frac{\alpha_i-\frac{\alpha_{i-1}}{2}}{1-\alpha_i-\frac{\alpha_{i-1}}{2}} \frac{1-\alpha_i}{\alpha_i} \right)^2} \ .$$
        Say we hope to achieve some relative minimum distance $\delta$. Then we can use the above expressions to fix $\delta=\alpha_{m+1}=f(\alpha_m,\alpha_{m-1})$, and so on, to eventually obtain an expression of the form $\delta=f(\alpha_1)$, which one can solve for $\alpha_1$. This yields assignments to $\alpha_1,...,\alpha_m$ which guarantee the partial derivatives to all these variables are 0. In short, this lets one find critical points of the asymptotic spectral shape function, for some target minimum distance. Find the local maxima, and if these are negative for your choice of $\delta,r$ and $m$, then this means that $\delta \leq \delta^{(m)}$. We used the latter method to fill up \Cref{tab:deltas} which contains lower bounds on various $\delta_m^{(m)}$. For example, to get $10^{-4}$-close to the GV-bound you need $m\geq4$.

    \subsection{Our codes}\label{sec:RAD}
        The codes introduced in this paper are directly inspired by $RA^m$ codes and their analysis as introduced by the previously cited works. As hinted in the introduction, our idea is to ``fix'' the distance of the dual code, and to do that we introduce rounds where we perform the inverse of the accumulation operation. This inverse operation is the following:
        \begin{itemize}
            \item Let $D \in \F_2^{n \times n}$ denote the \emph{discrete derivative matrix}, i.e., $D_{ij} = 1$ iff $i-1=j$ or $i=j$. 
        \end{itemize}
        To see this is the inverse, note that if $y = Ax$ we have $y_i = \sum_{j=1}^ix_j$, so $(Dy)_i = (\sum_{j=1}^ix_j)+(\sum_{j=1}^{i-1}x_j)=x_i$. Again, this $D$ matrix can be implemented with $n-1$ $\xor$ gates (and in fact here the depth is just $1$!). 
    
        With this $D$ operation in hand, we are able to construct pairs of dual codes of rate $1/2$. For instance the $RAD$ and $RDA$ codes are dual to one another when $r=2$, assuming the accumulation orders are ``reversed'' (i.e., in one code, $A$ and $D$ are transposed). More precisely, we have the following. 

        \begin{proposition} \label{prop:dual}
            Let $\pi_1,\dots,\pi_{2m}:[n]\to[n]$ be permutations. Consider the generator matrices 
            \[
                G=AM_{\pi_{2m}}DM_{\pi_{2m-1}}\cdots AM_{\pi_2}DM_{\pi_1}F_2
            \]
            and 
            \[
                H=D^\top M_{\pi_{2m}}A^\top M_{\pi_{2m-1}}\cdots D^\top M_{\pi_2}A^\top M_{\pi_1}F_2 \ .
            \]
            Then $H^\top G=0$. In particular, they generate dual codes. 
        \end{proposition}

        \begin{proof}
            Since for matrices $A,B$ we have $(AB)^\top = B^\top A^\top$ and the transpose of a permutation matrix is its inverse, it follows that 
            \begin{align*}
                H^\top = F_2^\top M_{\pi_1}^\top  A M_{\pi_2}^\top D \cdots M_{\pi_{2m-1}}^\top A M_{\pi_{2m}}^\top D = F_2^\top M_{\pi_1}^{-1} D^{-1} M_{\pi_2}^{-1}A^{-1}\cdots M_{\pi_{2m-1}}^{-1}D^{-1}M_{\pi_{2m}}^{-1}A^{-1} \ .
            \end{align*}
            Thus $ H^\top G = F_2^\top F_2 = 0$, this last equality being immediate from the definition. 
        \end{proof}
        
        Our goal will be to study the pair of codes $R(AD)^m$ and $R(DA)^m$ for rate $1/2$ over a uniformly random choice of $\pi_1,\dots,\pi_{2m}$ and prove that both achieve distance converging to the GV-bound as $m$ grows. 

        \begin{figure}[ht]
            \centering
            \begin{tikzpicture}
                \draw (7.05,6) rectangle (7.35,4.35); 
                \node at (7.175,6.25) {$x$};
        
                \draw (6,6) rectangle (6.8,2.7);
                \node at (6.5,6.25) {$F_r$};
                
                \node at (6.2,5.74) {$1$};
                \node at (6.2,5.43) {$1$};
                \node at (6.2,5.12) {$1$};
                \node at (6.2,4.81) {$1$};
                \node at (6.2,4.5) {$1$};
                \node at (6.6,4.19) {$1$};
                \node at (6.6,3.88) {$1$};
                \node at (6.6,3.57) {$1$};
                \node at (6.6,3.26) {$1$};
                \node at (6.6,2.95) {$1$};
                
                \draw (2.45,6) rectangle (5.75,2.7);
                \node at (4.1,6.25) {$M_{\pi_1}$};
        
                \node at (2.7,3.88) {$1$};
        
                \node at (3,4.19) {$1$};
               
                \node at (3.3,4.5) {$1$};
        
                \node at (3.6,5.43) {$1$};
                
                \node at (3.9,3.57) {$1$};
        
                \node at (4.5,5.74) {$1$};
                
                \node at (4.8,3.26) {$1$};
                
                \node at (5.1,4.81) {$1$};
        
                \node at (5.4,5.12) {$1$};

                \node at (4.2,2.9) {$1$};
                
                \draw (-1.1,6) rectangle (2.2,2.7);
                \node at (0.55,6.25) {$A$};

                \node at (-0.85,5.74) {$1$};
                \node at (-0.85,5.43) {$1$};
                \node at (-0.85,5.12) {$1$};
                \node at (-0.85,4.81) {$1$};
                \node at (-0.85,4.5) {$1$};
                \node at (-0.85,4.19) {$1$};
                \node at (-0.85,3.88) {$1$};
                \node at (-0.85,3.57) {$1$};
                \node at (-0.85,3.26) {$1$};
                \node at (-0.85,2.95) {$1$};
        
                \node at (-0.55,5.43) {$1$};
                \node at (-0.55,5.12) {$1$};
                \node at (-0.55,4.81) {$1$};
                \node at (-0.55,4.5) {$1$};
                \node at (-0.55,4.19) {$1$};
                \node at (-0.55,3.88) {$1$};
                \node at (-0.55,3.57) {$1$};
                \node at (-0.55,3.26) {$1$};
                \node at (-0.55,2.95) {$1$};
        
                \node at (-0.25,5.12) {$1$};
                \node at (-0.25,4.81) {$1$};
                \node at (-0.25,4.5) {$1$};
                \node at (-0.25,4.19) {$1$};
                \node at (-0.25,3.88) {$1$};
                \node at (-0.25,3.57) {$1$};
                \node at (-0.25,3.26) {$1$};
                \node at (-0.25,2.95) {$1$};
        
                \node at (0.05,4.81) {$1$};
                \node at (0.05,4.5) {$1$};
                \node at (0.05,4.19) {$1$};
                \node at (0.05,3.88) {$1$};
                \node at (0.05,3.57) {$1$};
                \node at (0.05,3.26) {$1$};
                \node at (0.05,2.95) {$1$};
        
                \node at (0.35,4.5) {$1$};
                \node at (0.35,4.19) {$1$};
                \node at (0.35,3.88) {$1$};
                \node at (0.35,3.57) {$1$};
                \node at (0.35,3.26) {$1$};
                \node at (0.35,2.95) {$1$};
        
                \node at (0.65,4.19) {$1$};
                \node at (0.65,3.88) {$1$};
                \node at (0.65,3.57) {$1$};
                \node at (0.65,3.26) {$1$};
                \node at (0.65,2.95) {$1$};
        
                \node at (0.95,3.88) {$1$};
                \node at (0.95,3.57) {$1$};
                \node at (0.95,3.26) {$1$};
                \node at (0.95,2.95) {$1$};
        
                \node at (1.25,3.57) {$1$};
                \node at (1.25,3.26) {$1$};
                \node at (1.25,2.95) {$1$};
                
                \node at (1.55,3.26) {$1$};
                \node at (1.55,2.95) {$1$};
        
                \node at (1.85,2.95) {$1$};
                
                \draw (-4.7,6) rectangle (-1.4,2.7);
                \node at (-3.05,6.25) {$M_{\pi_2}$};
                
                \node at (-3.6,3.88) {$1$};
        
                \node at (-4.2,4.19) {$1$};
               
                \node at (-2.4,4.5) {$1$};
        
                \node at (-1.5,5.43) {$1$};
                
                \node at (-1.8,3.57) {$1$};
        
                \node at (-4.5,5.74) {$1$};
                
                \node at (-2.1,3.26) {$1$};
                
                \node at (-3,4.81) {$1$};
        
                \node at (-3.9,5.12) {$1$};

                \node at (-3.3,2.9) {$1$};
                
                \draw (-8.3,6) rectangle (-5,2.7);
                \node at (-6.65,6.25) {$D$};
        
                \node at (-8,5.74) {$1$};
                \node at (-8,5.43) {$1$};
        
                \node at (-7.7,5.43) {$1$};
                \node at (-7.7,5.12) {$1$};
        
                \node at (-7.4,5.12) {$1$};
                \node at (-7.4,4.81) {$1$};
        
                \node at (-7.1,4.81) {$1$};
                \node at (-7.1,4.5) {$1$};
        
                \node at (-6.8,4.5) {$1$};
                \node at (-6.8,4.19) {$1$};
        
                \node at (-6.5,4.19) {$1$};
                \node at (-6.5,3.88) {$1$};

                \node at (-6.2,3.88) {$1$};
                \node at (-6.2,3.57) {$1$};                
                \node at (-5.9,3.57) {$1$};
                \node at (-5.9,3.26) {$1$};
        
                \node at (-5.6,3.26) {$1$};
                \node at (-5.6,2.95) {$1$};

                \node at (-5.3,2.95) {$1$};
            \end{tikzpicture}
            \caption{A pictorial representation of (a generator matrix for) an RAD code with rate $R=1/2$ and block-length $n=10$.}
            \label{fig:RAD-code}
        \end{figure}
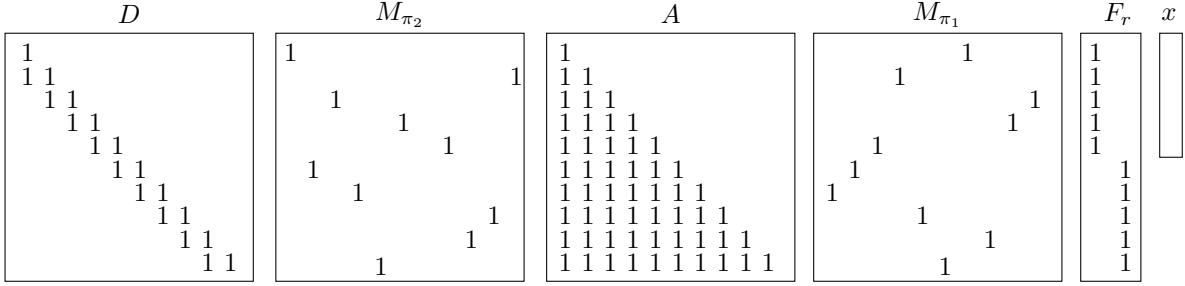
    
        Since $D$ is the inverse of $A$, it is not hard to describe $p_D(a,b)$, the probability that a uniformly random vector of absolute weight $a$ is mapped to a vector of absolute weight $b$ under $D$, in terms of $p_A(a,b)$: we can simply swap $a$ and $b$ in the numerator: 
        $$p_D(a,b) = \frac{\binom{n-a}{\lfloor b/2\rfloor}\binom{a-1}{\lceil b/2 \rceil-1}}{\binom{n}{a}} = \frac{\binom{n-a}{\lfloor b/2\rfloor}\binom{a}{\lceil b/2 \rceil}}{\binom{n}{a}}\frac{\lceil b/2 \rceil}{a}\ .$$
        We can then also define
        $$
        f_{D}(\alpha,\beta) = (1-\alpha )h\left(\frac{ \beta}{2(1-\alpha)}\right)+ \alpha h\left(\frac{ \beta}{2\alpha }\right) - h(\alpha) \ .
        $$
        such that $p_D(a,b) \leq O(1) \cdot 2^{n \cdot f_D(a/n,b/n)}$. We will often use the shorthands $f_{DA}(\alpha,\beta,\gamma):=f_D(\alpha,\beta)+f_A(\beta,\gamma)$ and $f_{AD}(\alpha,\beta,\gamma):=f_A(\alpha,\beta)+f_D(\beta,\gamma)$. With this, we can define the asymptotic spectral shape function for our codes of interest.

        \begin{definition}
            Let $m\in\mathbb{N}$. The \emph{asymptotic spectral shape function} of the rate $1/2$ $R(AD)^m$ and $R(DA)^m$ codes are defined, respectively, as
            $$\hat{r}^{(m)}_{AD}(\gamma) \leq \max\limits_{\alpha,\beta} \{\hat{r}^{(m-1)}_{AD}(\alpha) + f_{AD}(\alpha,\beta,\gamma)\} \ \ \text{ and }\ \ \hat{r}^{(1)}_{AD}(\gamma)=h(\gamma)/2 \ ,$$
            $$\hat{r}^{(m)}_{DA}(\gamma) \leq \max\limits_{\alpha,\beta} \{\hat{r}^{(m-1)}_{DA}(\alpha) + f_{DA}(\alpha,\beta,\gamma)\} \ \ \text{ and }\ \ \hat{r}^{(1)}_{DA}(\gamma)=h(\gamma)/2 \ .$$
        \end{definition}

    Lastly, we note that the probability distribution of $\wt(A^\top x)$ for a uniformly random $x$ of absolute weight $b$ is the same as the probability distribution of $\wt(A x)$, and similarly for $D$ and $D^\top$. This follows from the property that $A^\top$ is obtained by flipping each column of $A$, which means operationally that instead of accumulating ``down,'' $A^\top$ now accumulates ``up'' (and similarly for $D$ and $D^\top$). Thus, while technically one of the generator matrices needs to use the transpose of these matrices, the functions appearing in the analysis ($p_A,p_D,f_{AD}$, etc.) are unchanged. That is, our analysis applies regardless of whether or not the $A$ and $D$ matrices are transposed. 
        
\section{Proof strategy} \label{sec:proof-strategy}
    Our goal is to prove \Cref{thm:main_informal}: to find a pair of binary linear codes that are encodable with a linear sized circuit (in the block-length of the code), that are dual to each other, and that can attain minimum distance arbitrarily close to the GV-bound. We will do so for the pair of $R(DA)^m$ and $R(AD)^m$ codes of rate $1/2$, as introduced in \Cref{sec:RAD}. By \Cref{prop:dual}, these can be chosen dual to one another and are encodable with a linear sized circuit: in fact a circuit of size $2m(n-1)$. The task that remains is thus to prove that both of these codes can attain minimum distance arbitrarily close to the GV-bound. We will show this is the case for a randomly sampled such code, except with $1/\poly(n)$ probability. Getting closer to the GV-bound will require larger $m$, and thus a longer encoding time. In other words, we aim to prove the following theorem, which implies the main result stated as \Cref{thm:main_informal}.

    \begin{theorem}\label{thm:main}
        For all $\epsilon>0$ and all (even) $n \in \mathbb N$ large enough, there exists some $m=m(\eps) \in \mathbb N$ such that $R(AD)^m$ and $R(DA)^m$ codes of rate $1/2$ and block-length $n$ achieve relative minimum distance $\dGV-\epsilon$, except with probability $\leq n^{1/2-m/4}=1/\poly(n)$. The circuit size of the encoding function for either code is $2m(n-1)$, or $O(mn)$ with depth $O(m\log(n))$.
    \end{theorem}

    The natural question raised by this theorem is what the function $m(\eps)$ looks like. If $m(\eps)$ is too big this could make the codes practically irrelevant (despite linear circuit size asymptotically). The way we will prove the above theorem is to prove instead that the $R(DA)^m$ and $R(AD)^m$ codes can attain relative distance $\epsilon$-close to $\delta^{(m)}$, where we recall from \Cref{sec:RMA} that this was the distance achievable by an $RA^m$ code of rate $1/2$. We recall from that same section that $\delta^{(m)}$ converges to $\dGV$ with $m$~\cite{ravazzi_spectra_2009}. This, in turn, would establish the above result, as it implies that our codes can achieve minimum distance arbitrarily close to $\dGV$. 
    
    However, as the reader may recall from that same section, we do not have a theoretical understanding of how $m$ is related to $\epsilon$. Instead, we have a program which takes $m$ and computes $\epsilon$. From its behaviour we can see that the convergence is quite fast: for example, to get $10^{-4}$-close to the GV-bound you need $m\geq4$ (see \Cref{tab:deltas}). Getting a better quantitative understanding of this convergence remains an interesting open question for future work.

    \begin{table}[ht]
        \centering
        \begin{tabular}{c|ccccc}
            $m$ & 2 & 3 & 4 & 5 & $\infty$ ($\dGV$) \\\hline
            $\delta^{(m)}$  & $0.02859547585$ & $0.1033989603$  & $0.1099391081$ & $0.1100278348$ & $ 0.1100278644$ \\ 
            $\epsilon$ to $\dGV$ & $<10^{-1}$ & $<10^{-2}$ & $<10^{-4}$ & $<10^{-7}$ & 0
        \end{tabular}
        \caption{Lower bounds on the relative minimum distances $\delta^{(m)}$ achievable by $RA^m$ codes of rate $1/2$, as well as of the GV-bound (best known rate/distance trade-off for binary codes). We will prove that these relative minimum distances are also achieved by $R(AD)^m$ and $R(DA)^m$ codes of rate $1/2$. The values of $\delta^{(m)}$ were computed using the method outlined in \Cref{sec:RMA}.}
        \label{tab:deltas}
    \end{table}

    \begin{proof}[Proof of \Cref{thm:main}]
        As explained above, to establish this theorem, it suffices to prove that rate $1/2$ $R(AD)^m$ and $R(DA)^m$ codes achieve relative minimum distance $\epsilon$-close to $\delta^{(m)}$, except with probability $\tilde O(n^{1/2-m/4})=1/\poly(n)$. By a union bound over the failure probabilities for both codes, we then obtain the desired result. 
        
        Our strategy, which is directly inspired by the method used to prove that $RA^m$ codes achieve this same minimum distance, is as follows. Recall $p_A(a,b)$ and $p_D(a,b)$, which describe the probability that a uniformly random vector of absolute weight $a$ is mapped by $A$, $D$, respectively to a vector of absolute weight $b$. These let us express the expected number of codewords that have weight at most $\delta^{(m)}-\eps$, a quantity which, by Markov's inequality, translates into an upper bound on the probability that a randomly sampled code fails to achieve the given distance $\delta^{(m)}-\eps$.
        
        We will prove this expectation asymptotically behaves like $1/\poly(n)$, proving the theorem. Just like with $RA^m$ codes, we prove this by splitting the computation of the expectation based on the weight of a vector when entering the final round of the encoding (where we consider $AD$ to be one round and $DA$ to be one round). Specifically, we split based on whether the vector has a weight that is within $h:=\log^2 n$ of the endpoints $1$ and $n$ of the weight range. To work this out formally, let us write $d=(\delta_n-\epsilon)n$ and $A_w(code)$ for the expected number of weight $w$ codewords in the given code, e.g. $A_w(R(AD)^2)$. We can then upper bound the error probability we wish to bound as follows
        \begin{align*}
            &\mathbb{P}[R(DA)^m \text{ fails to have absolute distance }d] \\
            &\leq \mathbb{E}[\text{number of }R(DA)^m\text{ codewords of absolute weight } \leq d] \\
            &=\sum_{w_1=1}^{d}A_{w_1}(R(DA)^m) \\
            &=\sum_{w_1=1}^{n}\sum_{w_2=1}^{n}\sum_{w_3=1}^{d}A_{w_1}(R(DA)^{m-1}) p_D(w_1,w_2) p_A(w_2,w_3) \\
            &=\underbrace{\sum_{\substack{1\leq w_1 \leq h \lor \\ n-h\leq w_1 \leq n}}\sum_{w_2=1}^{n}\sum_{w_3=1}^{d}A_{w_1}(R(DA)^{m-1})p_D(w_1,w_2) p_A(w_2,w_3)}_{(*)} + \underbrace{\sum_{w_1=h}^{n-h}\sum_{w_2=1}^{n}\sum_{w_3=1}^{d} A_{w_1}(R(DA)^{m-1}) p_D(w_1,w_2) p_A(w_2,w_3)}_{(**)} \ .
        \end{align*}
        Note that we can write out an analogous expression for the dual code $R(AD)^m$. In \Cref{thm:**_negligible}, which we prove in \Cref{sec:middle-weight}, we will bound the term called $(**)$ as something negligible in $n$ for both codes, using an argument involving the spectral shape functions of the codes. In \Cref{thm:*_small}, which we prove in \Cref{sec:bounding-edge}, we will bound the term called $(*)$ as $\tilde O(n^{1/2-m/4}) = 1/\poly(n)$ for both codes. Together, we find $\mathbb{P}[R(DA)^m \text{ fails to have distance }\delta] \leq \tilde O(n^{1/2-m/4}) = 1/\poly(n)$, and the same for the dual code. Of course, since we want both codes to attain this distance, we take a final union bound, which technically doubles this error probability. 
    \end{proof}

\section{Bounding $(**)$: vectors with weight in middle of weight range before last round of encoding} \label{sec:middle-weight}
    Our goal in this section is to bound $(**)$: the expected number of vectors that enter the final round of the encoding with absolute weight that is at least $h=\log^2n$ from the boundaries $1$ and $n$ of the weight range. We need to do this for both $R(AD)^m$ and $R(AD)^m$ codes. Recall from the discussion in \Cref{sec:RMA} that one encounters a similar case in the analysis of $RA^m$ codes. There, we analysed the critical points of the spectral shape function, which led to the conclusion that the spectral shape function is flat at 0 for $\gamma \leq \delta^{(m)}$. Recall that the spectral shape function is built up round by round by adding a term $f_A(\alpha,\beta)$ to the previous function $\hat{r}^{(m-1)}_A(\alpha)$ and maximizing over $\alpha$. What we found for this small initial range of $\gamma \leq \delta^{(m)}$ is that the function over $\alpha$ is decreasing. Hence, to maximize it, we should set $\alpha=0$, making the spectral shape function equal to 0 for such $\gamma$. In this context, however, we have a lower bound $h/n \leq \alpha$ so that we should instead fix $\alpha=h/n$, making the spectral shape function strictly negative (instead of 0), and allowing us to argue that $(**) \leq \exp(-h)$: we only expect a negligible number of the vectors with relative weight at least $h/n$ after $m-1$ rounds of accumulation, to then drop below the target distance of $\delta^{(m)}-\epsilon$ after the $m$-th round. 

    We want to adapt this general strategy to bound $(**)$ for $R(AD)^m$ and $R(DA)^m$ codes. Unfortunately, due to having twice as many operations and weights, directly working out the critical points is infeasible.\footnote{At least, the analytic expressions for the critical points of the spectral shape function of an $RA^m$ code (as found by Kliewer et al~\cite{kliewer_new}) are obtained after making some very clever substitutions. It is unclear to us how this technique could be adapted for $R(AD)^m$ or $R(DA)^m$ codes.} A numerical approach, using a grid search to estimate maximizers, as Pfister applied to $RA^m$ codes, could work to estimate the minimum distance of these codes. But this would not \textit{prove} the codes attain this minimum distance, nor would it quantify the error probability.
    
    We therefore came up with another approach: directly upper bound the spectral shape function of our codes by that of $RA^m$ codes. If this is true, and if this upper bound remains true when the range over $\alpha$ is restricted, then the upper bound on $(**)$ will follow almost immediately (up to slightly higher polynomial factors from the union bound, which will be irrelevant anyway). This motivates us to consider spectral shape functions where the input weight $\alpha$ is bounded away from $0$ and $1$. Formally, we consider the following. 
    
    \begin{definition} [Restricted Spectral Shape Function]
        Fix $\tau \in [0,1/2)$. The \emph{$\tau$-restricted spectral shape function} for $RA^m$, $R(AD)^m$ and $R(DA)^m$ codes respectively, are
        \begin{align*}
            \hat{r}^{(m)}_A(\tau,\gamma) &:= \max_{\substack{\tau \leq \alpha\leq \min\{1-\tau,2\gamma,2(1-\gamma)\}}} \{\hat{r}^{(m-1)}_A(\alpha) + f_A(\alpha,\gamma)\} \ ,\\
            \hat{r}^{(m)}_{AD}(\tau,\gamma) &:= \max_{\substack{\tau \leq \alpha\leq 1-\tau\\ \beta}} \{\hat{r}^{(m-1)}_{AD}(\alpha) + f_{AD}(\alpha,\beta,\gamma)\} \ , \\
            \hat{r}^{(m)}_{DA}(\tau,\gamma) &:= \max_{\substack{\tau \leq \alpha\leq 1-\tau\\ \beta}} \{\hat{r}^{(m-1)}_{DA}(\alpha) + f_{DA}(\alpha,\beta,\gamma)\} \ .
        \end{align*}
    \end{definition}
    
    Note that the upper bound on the range of $\alpha$ for $RA^m$ codes follows because $f_A(\alpha,\gamma)$ is only defined for $\alpha \leq \min\{2\gamma,2(1-\gamma)\}$. This doesn't apply to the other two codes, as there any pair $(\alpha,\gamma)$ is a valid input to $f_{AD}$ or $f_{DA}$ (though the choice can limit the range over $\beta$). This is also the reason why we constrain $h/n \leq \alpha \leq 1-h/n$ on both sides of the weight range, while in the analysis of $RA^m$ codes it sufficed to just lower bound $\alpha$. Because of these mismatched ranges we aim to prove an inequality between the restricted spectral shape functions $\hat{r}^{(m)}_{AD}(\tau,\gamma) \leq \hat{r}^{(m)}_{A}(\tau,\gamma)$ (and likewise for the dual code $R(DA)^m$), instead of directly proving $f_{AD}(\alpha,\beta,\gamma) \leq f_A(\alpha,\gamma)$ (as it is not meaningful to prove such a pointwise inequalities as the domains are different). 

    The rest of this section is structured as follows. First, just below, we prove a number of useful inequalities involving $f_A$ and $f_D$ which will be used throughout our proofs. Second, we prove that $f_{AD}(\alpha,\beta,\gamma)$ and $f_{DA}(\alpha,\beta,\gamma)$ can both be upper bounded by the same function $g(\alpha,\gamma)$, which is quite similar to $f_A(\alpha,\gamma)$. Third, we use this upper bound of $g$ to prove that the restricted spectral shape function of both $R(AD)^m$ and $R(DA)^m$ codes are upper bounded by the restricted spectral shape function of $RA^m$ codes. Fourth and final, we show that due to this inequality between the restricted spectral shape functions, we can in fact upper bound $(**)$ as something negligible in $n$ for both codes. 

    \begin{restatable}{lemma}{fafdclaims}\label{lem:fa_fd_claims}
        \begin{enumerate}
            \item $f_A(\alpha,\beta)$ is defined for $\alpha \leq \min\{2\beta,2(1-\beta)\}$ or equivalently, $\alpha/2 \leq \beta \leq 1-\alpha/2$.
            \item $f_D(\alpha,\beta)$ is defined for $\beta \leq \min\{2\alpha,2(1-\alpha)\}$ or equivalently, $\beta/2 \leq \alpha \leq 1-\beta/2$.
            \item $f_A(\alpha,\beta)$ decreases monotonically with $\alpha$. 
            \item $f_A(\alpha,\beta)$ is symmetric around a maximum at $\beta=1/2$, so that it increases with $\beta$ until $\beta=1/2$ and decreases with $\beta$ afterwards. 
            \item $f_A(\alpha,\beta) \leq h(\beta) - h(\alpha)$.
            \item $f_A(1-\alpha,\beta) \leq h(\beta) - h(\alpha)$.
            \item $f_D(\alpha,\beta) \leq h(\beta) - h(\alpha)$.
        \end{enumerate}
    \end{restatable}

    \subsection{Bounding $f_{DA} \leq g$}   
        Our goal in this section and the next is to upper bound $f_{DA}(\alpha,\beta,\gamma)$ and $f_{AD}(\alpha,\beta,\gamma)$ by the function
        $$g(\alpha,\gamma)= \begin{cases} 
            f_A(\alpha,\gamma) & \text{if } \alpha \leq 2\gamma(1-\gamma) \ , \\
            h(\gamma)-h(\alpha) & \text{if } 2\gamma(1-\gamma) \leq \alpha \leq 1-2\gamma(1-\gamma) \ , \\
            f_A(1-\alpha,\gamma)  & \text{if } 1-2\gamma(1-\gamma) \leq \alpha \ . \\
        \end{cases}$$
        We will do so first for $f_{DA}(\alpha,\beta,\gamma)$. Essentially, the idea of this upper bound is to upper bound $f_{DA}(\alpha,\beta,\gamma) \leq f_A(\alpha,\gamma)$ across as much of the domain as possible. If this inequality would hold across the entire domain, the desired inequality between the spectral shape functions would follow immediately. Of course, it is not possible to establish this upper bound across the entire domain, as $f_A(\alpha,\gamma)$ is only defined for $\alpha \leq \min\{2\gamma,2(1-\gamma)\}$ while the expression we wish to upper bound $f_{DA}(\alpha,\beta,\gamma)$ is defined for all $(\alpha,\gamma)$ (though the choice may limit the selection over $\beta$). In particular, then, when one has large $\alpha$ and either very small or very large $\gamma$ ($2\gamma \leq \alpha$ and $2(1-\gamma) \leq \alpha$), then this upper bound will not be possible to establish, for example when $\alpha=0.3$ and $\gamma=0.1$ or $\gamma=0.9$.
    
        We come close, and manage to establish this desired upper bound for $\alpha \leq 2\gamma(1-\gamma)$. We then do the same on the other end, proving an upper bound of $f_A(1-\alpha,\gamma)$ for $1-\alpha \leq 2\gamma(1-\gamma)$. We prove a different upper bound on the middle of the range (so including most of the range which wouldn't admit either $f_A(\alpha,\gamma)$ or $f_A(1-\alpha,\gamma)$ to begin with), and afterwards show that this different bound still suffices to prove the inequality between the restricted spectral shape functions.
    
        \begin{theorem}\label{thm:fda_bound}
            If $\alpha,\beta,\gamma \in [0,1]$ such that $\alpha,\beta,\gamma$ is a valid input to $f_{DA}$, that is, if $\beta/2 \leq \min\{\alpha,1-\alpha,\gamma,1-\gamma\}$, then
            \[ f_{DA}(\alpha, \beta, \gamma) \leq g(\alpha,\gamma)=
            \begin{cases} 
            f_A(\alpha,\gamma) & \text{if } \alpha \leq 2\gamma(1-\gamma) \ , \\
            h(\gamma)-h(\alpha) & \text{if } 2\gamma(1-\gamma) \leq \alpha \leq 1-2\gamma(1-\gamma) \ , \\
            f_A(1-\alpha,\gamma)  & \text{if } 1-2\gamma(1-\gamma) \leq \alpha \ . \\
            \end{cases}
            \]
        \end{theorem}
        \begin{proof}
            We recall \Cref{lem:fa_fd_claims}.7 which stated $f_D(\alpha,\beta) \leq h(\beta) - h(\alpha)$. We also note that always $f_D(\alpha,\beta) \leq 0$ (as it is the logarithm of a probability) and that $h(\beta) - h(\alpha)$ can sometimes become positive, but note that this can't happen in case $\beta \leq \min\{\alpha,1-\alpha\}$. Put together, we have the following inequality 
            $$f_{DA}(\alpha,\beta,\gamma) \leq 
            \begin{cases} 
            h(\beta)-h(\alpha)+f_A(\beta, \gamma) & \text{in (1) : }  \beta \leq \min\{\alpha,1-\alpha\} \ ,\\
            f_A(\beta,\gamma)  & \text{in (2), (3), (4) : otherwise} \ .
            \end{cases}$$
            This inequality will be our starting point, and we will argue that we can upper bound it as $g(\alpha,\gamma)$. To do, we introduce the following figure, which portrays the space of valid inputs $(\alpha,\beta)$ to $f_D$, where we recall that this function introduces the constraints $\beta \leq 2\alpha$ and $\beta \leq 2(1-\alpha)$, see \Cref{lem:fa_fd_claims}.2.
            \begin{figure}[ht]
                \centering
                \resizebox{0.8\textwidth}{!}{%
                \begin{tikzpicture}
                % Axes
                \draw [->] (-0.5,0) -- (10.75,0)node[pos=1,right, fill=white]{$\alpha$};
                \draw [->] (0,-0.5) -- (0,5.75)node[pos=1,above, fill=white]{$\beta$};
                % Axes tiks.
                \draw (10,-0.25) -- (10,0)node[pos=0.5,below, fill=white]{$1$};
                \draw (5,-0.25) -- (5,0)node[pos=0.5,below, fill=white]{$1/2$};
                \draw (-0.25,5)--(0,5)node[pos=0.5,left, fill=white]{$1/2$};
                \draw (1.6,-0.25)--(1.6,0)node[pos=0.5,below, fill=white]{$2\gamma(1-\gamma)$};
                \draw (8.4,-0.25)--(8.4,0)node[pos=0.5,below, fill=white]{$1-2\gamma(1-\gamma)$};
                % Borders
                \draw  (0,0) -- (2.875,5.75)node[pos=0.6,above,left=0]{$\beta = 2\alpha$};
                \draw  (10,0) -- (7.125,5.75)node[pos=0.6,above,right=0]{$\beta = 2(1-\alpha)$};
                \draw  (0,0) -- (5.75,5.75)node[pos=0.4,auto,right=0,']{$\beta = \alpha$};
                \draw  (10,0) -- (4.25,5.75)node[pos=0.4,auto,left=0]{$\beta = 1-\alpha$};
                \draw  (0.8,1.6) -- (9.2,1.6)node at(5,1.3){$\beta = 2\gamma(1-\gamma)$};
                \draw  (1.6,0) -- (1.6,1.6);
                \draw  (8.4,0) -- (8.4,1.6);
                % Labels of regions
                \draw (5,2.5)node{(1)};
                % \draw (1,0.5)node{(1a)};
                % \draw (9,0.5)node{(1b)};
                \draw (5,7)node{(2)};
                \draw (3,4.5)node{(3)};
                \draw (7,4.5)node{(4)};
                \end{tikzpicture}}
            \end{figure}
            Note here that the region labelled as (1) corresponds to $\beta \leq \min\{\alpha,1-\alpha\}$ and is therefore corresponds to the `if' case above, while the leftover regions (2), (3) and (4) correspond to the `otherwise' case above. We will now go over these four regions, and for each region identify which of the three ranges over $\alpha$ (as introduced in $g(\alpha,\gamma)$ apply, and for each range argue that the corresponding upper bound on $f_{DA}$ must be true. 

            Let's start with the `otherwise' case, regions (2), (3) and (4), as we can deal with all three of them in a similar way. Note that in regions (2) and (3) we have $\alpha \leq \beta$, while we have the expression $f_A(\beta,\gamma)$ that we wish to upper bound. Recall from \Cref{lem:fa_fd_claims}.3 that $f_A(\beta,\gamma)$ is decreasing with $\beta$. This means we can simply decrease $\beta$ to $\alpha$ which only increases our expression $f_A(\beta,\gamma)$. In other words, on regions (2) and (3) we have $f_A(\beta,\gamma)\leq f_A(\alpha,\gamma)$. This establishes that on regions (2) and (3) we have the required bound for $\alpha \leq 2\gamma(1-\gamma)$. Similarly, on regions (2) and (4) we have $1-\alpha \leq \beta$. We can therefore lower $\beta$ to $1-\alpha$ on these regions to obtain the bound $f_A(\beta,\gamma)\leq f_A(1-\alpha,\gamma)$. This gives us the required bound for $1-2\gamma(1-\gamma)\leq\alpha$ on regions (3) and (4). 

            However, when we are in regions (2), (3) or (4), we can also be in the middle range over $\alpha$: $2\gamma(1-\gamma)\leq\alpha\leq 1-2\gamma(1-\gamma)$. In that case, we have to argue that we can upper bound the given expression as $h(\gamma)-h(\alpha)$. For (2) and (3), it suffices to show $f_A(\alpha,\gamma) \leq h(\gamma)-h(\alpha)$, which follows from \Cref{lem:fa_fd_claims}.5. Likewise, for (4) it suffices to show $f_A(1-\alpha,\gamma) \leq h(\gamma)-h(\alpha)$, which follows from \Cref{lem:fa_fd_claims}.6.

            Having established all the relevant inequalities on regions (2), (3) and (4), we just have to deal with region (1). We cannot repeat the same trick here, as by definition $\beta$ is smaller than both $\alpha$ and $1-\alpha$. Instead, we make use of the fact that $h(\beta) - h(\alpha) + f_A(\beta,\gamma)$ has a local maximum for $\beta$ at $\beta=2\gamma(1-\gamma)$, proven in the appendix as \Cref{lem:max-beta=2gamma(1-gamma)}. We draw this value of $\beta$ as a horizontal line in the diagram. Note that this local maximum is always $\leq 1/2$, and thus indeed goes through region (1). 

            To proceed, we consider the three ranges over $\alpha$. First, consider just $\alpha\leq 2\gamma(1-\gamma)$. This corresponds to the triangular region at the left of our picture. Here, we want to show that $h(\beta)-h(\alpha)+f_A(\beta, \gamma) \leq f_A(\alpha,\gamma)$. To see this is true, recall that we have $\beta \leq \alpha \leq 2\gamma(1-\gamma)$. Since $2\gamma(1-\gamma)$ is the local maximum for $\beta$, we want to make $\beta$ as large as possible, i.e. we want to set $\beta=\alpha$. But that yields $h(\beta) - h(\alpha) + f_A(\beta,\gamma) \leq h(\alpha) - h(\alpha) + f_A(\alpha,\gamma) = f_A(\alpha,\gamma)$. 

            Next, consider $1-\alpha\leq 2\gamma(1-\gamma)$, corresponding to a symmetric triangular region at the right of our picture. We can deal with this analogously. We want to show $h(\beta)-h(\alpha)+f_A(\beta, \gamma) \leq f_A(1-\alpha,\gamma)$. We again have to make $\beta$ as large as possible, which means $\beta=1-\alpha$, giving us $h(\beta) - h(\alpha) + f_A(\beta,\gamma) \leq h(1-\alpha) - h(\alpha) + f_A(1-\alpha,\gamma) =f_A(1-\alpha,\gamma)$.

            What remains is the middle part of (1), where $2\gamma(1-\gamma) \leq \alpha \leq 1-2\gamma(1-\gamma)$. Here, we want to show $h(\beta)-h(\alpha)+f_A(\beta, \gamma) \leq h(\gamma)-h(\alpha)$. To do so, we will set $\beta$ to its local maximum of $2\gamma(1-\gamma)$, so that we get 
            $$h(\beta)-h(\alpha)+f_A(\beta, \gamma) \leq h(2\gamma(1-\gamma)) - h(\alpha) + f_A(2\gamma(1-\gamma),\gamma) = h(\gamma) - h(\alpha) \ ,$$
            where the step step follows from the lemma below which we prove in the appendix.

            \begin{restatable}{lemma}{hfa}\label{lem:h_fa}
                $h(2\alpha(1-\alpha))+f_A(2\alpha(1-\alpha),\alpha) = h(\alpha)$.
            \end{restatable}
        \end{proof}
    
    \subsection{Bounding $f_{AD} \leq g$}
        \begin{theorem}\label{thm:fad_bound}
            If $\alpha,\beta,\gamma \in [0,1]$ such that $\alpha,\beta,\gamma$ is a valid input to $f_{AD}$, that is, if $\alpha/2 \leq \min\{\beta,1-\beta\}, \gamma/2 \leq \min\{\beta,1-\beta\}$, then
            $$f_{AD}(\alpha,\beta,\gamma) \leq g(\alpha,\gamma) =
            \begin{cases} 
            f_A(\alpha,\gamma) & \text{if } \alpha \leq 2\gamma(1-\gamma) \ , \\
            h(\gamma)-h(\alpha) & \text{if } 2\gamma(1-\gamma) \leq \alpha \leq 1-2\gamma(1-\gamma) \ , \\
            f_A(1-\alpha,\gamma)  & \text{if } 1-2\gamma(1-\gamma) \leq \alpha \ . \\
            \end{cases}$$
        \end{theorem}
        \begin{proof}
            Just as in the proof of the analogous \Cref{thm:fda_bound} for the dual code, we recall from \Cref{lem:fa_fd_claims}.7 that $f_D(\beta,\gamma) \leq h(\gamma)-h(\beta)$, and also that always $f_D(\beta,\gamma) \leq 0$. Since the former upper bound can exceed 0 when $\beta \leq \min\{\gamma,1-\gamma\}$ or when $\beta \geq \max\{\gamma,1-\gamma\}$, in these cases we apply the naive upper bound of $0$. This gives us:
            $$f_{AD}(\alpha,\beta,\gamma) \leq 
            \begin{cases}  
            f_A(\alpha, \beta) + h(\gamma)-h(\beta) & \text{in (1), (2) : }  \gamma \leq \min\{\beta,1-\beta\} \text{ or }\gamma \geq \max\{\beta,1-\beta\} \ ,\\
            f_A(\alpha,\beta)  & \text{in (3), (4) :  otherwise} \ .
            \end{cases}$$
            We will again draw a picture depicting the space of valid inputs $(\beta,\gamma)$ to $f_D$. Recall that for the dual code, this picture related instead $\alpha$ and $\beta$. We will again go region by region to prove the relevant bound per region. 

            Note that both $f_A(\alpha,\beta)$ and $ f_A(\alpha, \beta) + h(\gamma)-h(\beta)$ are symmetric around $\beta=1/2$. This is true of $f_A(\alpha,\beta)$ due to \Cref{lem:fa_fd_claims}.4, and this is obviously true for $-h(\beta)$. As such, we only need to consider the left half of our figure, i.e., where $\beta\leq 1/2$. 
            
            
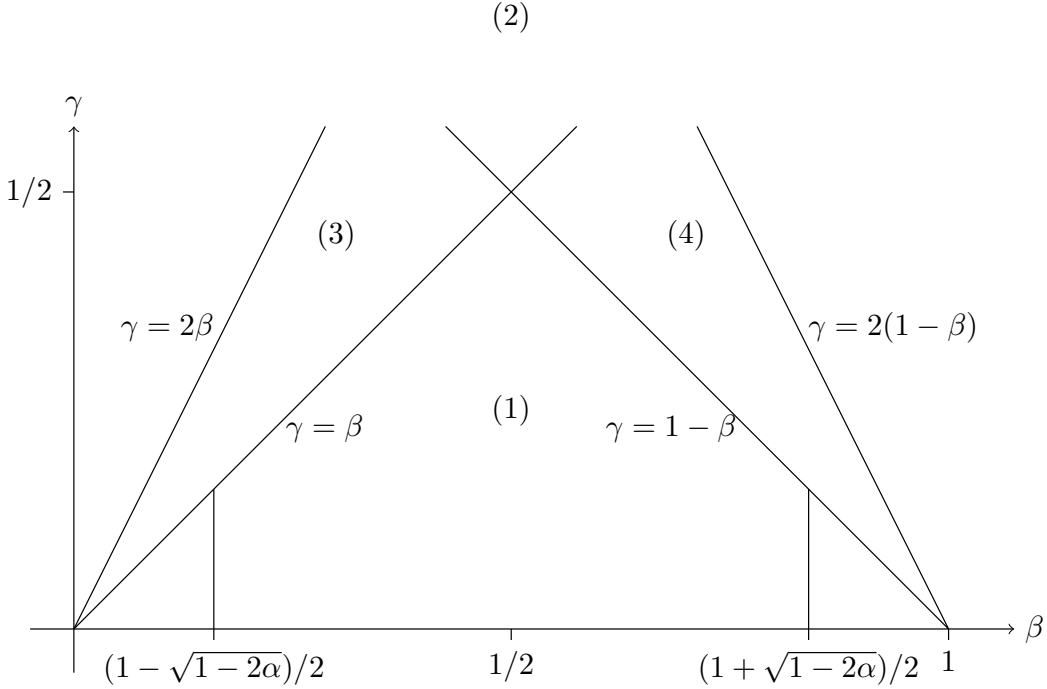
\begin{figure}[ht]
                \centering
                \resizebox{0.8\textwidth}{!}{%
                \begin{tikzpicture}
                % Axes
                \draw [->] (-0.5,0) -- (10.75,0)node[pos=1,right, fill=white]{$\beta$};
                \draw [->] (0,-0.5) -- (0,5.75)node[pos=1,above, fill=white]{$\gamma$};
                % Axes tiks.
                \draw (10,-0.25) -- (10,0)node[pos=0.5,below, fill=white]{$1$};
                \draw (5,-0.25) -- (5,0)node[pos=0.5,below, fill=white]{$1/2$};
                \draw (-0.25,5)--(0,5)node[pos=0.5,left, fill=white]{$1/2$};
                \draw (1.6,-0.25)--(1.6,0)node[pos=0.5,below, fill=white]{$(1-\sqrt{1-2\alpha})/2$};
                \draw (8.4,-0.25)--(8.4,0)node[pos=0.5,below, fill=white]{$(1+\sqrt{1-2\alpha})/2$};

                % Borders
                \draw  (0,0) -- (2.875,5.75)node[pos=0.6,above,left=0]{$\gamma = 2\beta$};
                \draw  (10,0) -- (7.125,5.75)node[pos=0.6,above,right=0]{$\gamma = 2(1-\beta)$};
                \draw  (0,0) -- (5.75,5.75)node[pos=0.4,auto,right=0,']{$\gamma = \beta$};
                \draw  (10,0) -- (4.25,5.75)node[pos=0.4,auto,left=0]{$\gamma = 1-\beta$};
                \draw  (1.6,0) -- (1.6,1.6);
                \draw  (8.4,0) -- (8.4,1.6);

                % Labels of regions
                \draw (5,2.5)node{(1)};
                \draw (5,7)node{(2)};
                \draw (3,4.5)node{(3)};
                \draw (7,4.5)node{(4)};
                \end{tikzpicture}}
            \end{figure}            
            
            Let's start with region (3), where we have the expression $f_A(\alpha,\beta)$. We first claim that on this entire region we can bound $f_A(\alpha,\beta) \leq f_A(\alpha,\gamma)$. Recall from \Cref{lem:fa_fd_claims}.4 that $f_A(\beta,\gamma)$ is symmetric around a maximum at $\gamma=1/2$, so that it increases with $\gamma$ until $\gamma=1/2$ and decreases with $\gamma$ afterwards. In region (3) we are guaranteed that $\beta\leq\gamma$. The bottom part of (3) also has $\gamma\leq 1/2$, so for that part of the region we are safe to increase $\beta$ to $\gamma$, and we obtain the desired bound. That just leaves the part of (3) where $1/2 \leq \gamma$. Here we have a situation of the form $\beta \leq 1/2 \leq \gamma$. We claim that the distance between $\beta$ and $1/2$ cannot be larger than the  distance between $1/2$ and $\gamma$. This is because on (3) we have the constraint $\gamma \leq 1-\beta$. I.e., for $\gamma$ to exceed $1/2$, we must have that $\beta$ drops below $1/2$ by the same amount. Since the maximum of $f_A(\alpha,\beta)$ lies at $\beta$ as close as possible to $1/2$, this would give us the desired upper bound: we can always swap out $\beta$ for $\gamma$, as $\gamma$ is guaranteed to be at least as close to $1/2$ as $\beta$ is. 
    
            This establishes $f_{AD}(\alpha,\beta,\gamma) \leq f_A(\alpha,\gamma)$ across all of region (3). Note that region (3) can include any $\alpha$ (depending on the value of $\beta$ and $\gamma$). This therefore definitely establishes the required inequality for $\alpha \leq 2\gamma(1-\gamma)$. Now, it is not hard to see that, for any $(\alpha,\gamma)$ such that $f_A(\alpha,\gamma)$ and $f_A(1-\alpha,\gamma)$ are both well-defined, we must have that $f_A(1-\alpha,\gamma)$ is the larger one for $\alpha \geq 1/2$ and vice versa. Hence, whenever $1/2 \leq 1-2\gamma(1-\gamma)  \leq \alpha$ we can assume $f_{AD}(\alpha,\beta,\gamma) \leq f_A(1-\alpha,\gamma)$ also. What remains to be shown for region (3) is that in the middle of the range of $\alpha$, specifically when $2\gamma(1-\gamma) \leq \alpha \leq 1-2\gamma(1-\gamma)$, we can upper bound our expression as $h(\gamma)-h(\alpha)$. It suffices to show that on this range we have $f_A(\alpha,\gamma) \leq h(\gamma)-h(\alpha)$, which we know to be true from \Cref{lem:fa_fd_claims}.5.
    
            We have thus established the upper bounds on our expression for the three different ranges over $\alpha$ over region (3). By symmetry around $\beta=1/2$, these results carry over to region (4). Next we consider region (1), where we recall again that we only need to consider the left half, i.e. where $\beta\leq 1/2$. In this case, we cannot increase $\beta$ to $\gamma$, increasing $f_A(\alpha,\beta)$ to $f_A(\alpha,\gamma)$, as by definition of region (1) we have $\gamma \leq \beta$. We therefore need to make use of $f_D(\beta,\gamma)$, which will hopefully let us decrease $\beta$ to $\gamma$, while increasing the value of our expression. Note that by including the above upper bound on $f_D$, we end up with the expression $f_A(\alpha, \beta) + h(\gamma)-h(\beta)$ on this region. Looking at the partial derivative to $\beta$, in \Cref{lem:max-beta=1/2} we show that this expression has a local maximum at $\beta=1/2$ when $\alpha \geq 1/2$. This lets us upper bound $f_A(\alpha, \beta) + h(\gamma)-h(\beta) \leq f_A(\alpha, 1/2) + h(\gamma)- h(1/2) = h(\gamma)-1$. Since $-1 \leq -h(\alpha)$ for all $\alpha$, we can bound our expression as $h(\gamma)-h(\alpha)$. We use this upper bound for $1/2 \leq \alpha \leq 1-2\gamma(1-\gamma)$. 
    
            Once $1-2\gamma(1-\gamma)\leq \alpha$, we then need to prove that $h(\gamma)-1 \leq f_A(1-\alpha,\gamma)$. Recall from \Cref{lem:fa_fd_claims}.3 that $f_A(\alpha,\gamma)$ decreases with $\alpha$, so that also $f_A(1-\alpha,\gamma)$ increases with $\alpha$. Therefore, if we establish the bound for the smallest value of $\alpha=1-2\gamma(1-\gamma)$, then it is valid for all larger $\alpha$ as well. We thus need to establish $h(\gamma)-1 \leq f_A(2\gamma(1-\gamma),\gamma)$. Recall from \Cref{lem:h_fa} that we can write $f_A(2\gamma(1-\gamma),\gamma)= h(\gamma)-h(2\gamma(1-\gamma))$, which we use to rewrite the condition we need to prove to something obviously true: $h(\gamma)-1 \leq h(\gamma)-h(2\gamma(1-\gamma)) \iff h(2\gamma(1-\gamma)) \leq 1$.
    
            We've established the required bounds for region (1) under the assumption that $1/2 \leq \alpha$. Let us now consider the case where $\alpha \leq 1/2$. By \Cref{lem:max-beta=sqrt-stuff}, we find a global maximum for the range $\beta \leq 1/2$ at the point $\beta^*=(1 - \sqrt{1-2\alpha})/2$.
            Since we are in region (1), we know that $\gamma \leq \beta$. We can consider two cases here: either the critical point $\beta^*$ is in range (meaning $\gamma \leq \beta^*$ or equivalently $2\gamma(1-\gamma)\leq \alpha$) or it is out of range (meaning $\beta^* \leq \gamma$ or equivalently $\alpha\leq2\gamma(1-\gamma)$). 
    
            Suppose first that $\alpha \leq 2\gamma(1-\gamma)$ so that the critical point $\beta^*$ is out of range of $\beta$. To maximize our expression, we then have to lower $\beta$ as far as we can, which yields $\beta=\gamma$. This lets us bound $f_A(\alpha, \beta) + h(\gamma)-h(\beta) \leq f_A(\alpha, \gamma) + h(\gamma)-h(\gamma)=f_A(\alpha, \gamma)$, as desired. Now suppose instead that $2\gamma(1-\gamma) \leq \alpha$ so that we can fix $\beta=\beta^*$. We then bound 
            $$f_A(\alpha, \beta) + h(\gamma)-h(\beta) \leq f_A(\alpha,(1 - \sqrt{1-2\alpha})/2) + h(\gamma) - h((1 - \sqrt{1-2\alpha})/2) = h(\gamma) - h(\alpha) \ ,$$
            where the last step follows from the following lemma.
            \begin{lemma}\label{lem:h_fa_2}
                $f_A(\alpha,(1 - \sqrt{1\pm2\alpha})/2) - h((1 - \sqrt{1\pm2\alpha})/2) = - h(\alpha)$ \ .
            \end{lemma}
            \begin{proof}
                Recall \Cref{lem:h_fa}, which stated that $h(2\alpha(1-\alpha))+f_A(2\alpha(1-\alpha),\alpha) = h(\alpha)$. Now write $\beta=2\alpha(1-\alpha)$. Rewriting gives $\alpha=(1\pm\sqrt{1-2\beta})/2$. This means we can rewrite the statement coming from the Lemma as $h(\beta)+f_A(\beta,(1\pm\sqrt{1-2\beta})/2) = h((1\pm\sqrt{1-2\beta})/2)$, which is what we wanted to prove.
            \end{proof}
    
            That just leaves region (2), where we again note that we only look at the half on the left, where $\beta \leq 1/2$. This is a similar situation to region (1), as we are bounding the same expression $f_A(\alpha, \beta) + h(\gamma)-h(\beta)$ with the same critical points: $\beta=1/2$ when $1/2 \leq \alpha$ and $\beta^*$ when $\alpha \leq 1/2$. The difference is that we have $1-\beta \leq \gamma$ now (instead of $\gamma \leq \beta$ as in region (1)), which can be rewritten to $1-\gamma \leq \beta$. Suppose first that $1/2 \leq \alpha$. Note that in our analysis of this case on region (1) we didn't actually use the fact that $\gamma \leq \beta$ (which is different now). It follows that the resulting bounds also apply here on region (2). Moving to the case $\alpha \leq 1/2$, this means that the critical point is in range when $1-\gamma \leq \beta^*$ which we can be rewritten to $2\gamma(1-\gamma) \leq \alpha$. The critical point is then out of range whenever $\alpha \leq 2\gamma(1-\gamma)$. Note that these are the same conditions are we had for region (1). It follows that we also obtain the same upper bounds: $f_A(\alpha,\gamma)$ whenever $\alpha \leq 2\gamma(1-\gamma)$ and $h(\gamma)-h(\alpha)$ when $2\gamma(1-\gamma) \leq \alpha \leq 1/2$.
        \end{proof}
    
    \subsection{Bounding the restricted spectral shape functions}
        \begin{theorem}\label{thm:rda_spectral_shape}
            If $m\geq 1$, $\tau \in [0,\min\{2\gamma,2(1-\gamma)\}]$ and $\gamma\in[0,1]$, then $\hat{r}^{(m)}_{DA}(\tau,\gamma) \leq \hat{r}^{(m)}_A(\tau,\gamma)$ and $\hat{r}^{(m)}_{AD}(\tau,\gamma) \leq \hat{r}^{(m)}_A(\tau,\gamma)$.
        \end{theorem}
        \begin{proof}
            We prove the theorem by induction on $m$. We write the proof out for $R(DA)^m$ codes, and at the end will note that the same proof will apply to the dual code. The base case $m=0$ is trivial, as both spectral shape functions are equal, consisting just of the term $h(\alpha)/2$, which still makes them equal when we limit their range by the same amount $\tau$. For the induction step, assume the statement is true for some $m-1 \geq 1$: $\hat{r}^{(m-1)}_{DA}(\tau,\gamma) \leq \hat{r}^{(m-1)}_A(\tau,\gamma)$. We then need to prove that also $\hat{r}^{(m)}_{DA}(\tau,\gamma) \leq \hat{r}^{(m)}_A(\tau,\gamma)$. This is equivalent to saying 
            $$\max_{\substack{\tau \leq \alpha \leq 1-\tau\\\beta}} \{\hat{r}^{(m-1)}_{DA}(\alpha) + f_{DA}(\alpha,\beta,\gamma)\}  \leq \max_{\substack{\tau \leq \alpha \leq 1-\tau}} \{\hat{r}^{(m-1)}_A(\alpha) + f_A(\alpha,\gamma)\} \ .$$
            We immediately note that we can apply the induction hypothesis to bound $\hat{r}^{(m-1)}_{DA}(\tau,\gamma) \leq \hat{r}^{(m-1)}_A(\tau,\gamma)$ on the left hand side, to obtain the following easier statement that we want to prove:
            $$\max_{\substack{\tau \leq \alpha \leq 1-\tau\\\beta}} \{\hat{r}^{(m-1)}_A(\alpha) + f_{DA}(\alpha,\beta,\gamma)\}  \leq \max_{\substack{\tau \leq \alpha \leq 1-\tau}} \{\hat{r}^{(m-1)}_A(\alpha) + f_A(\alpha,\gamma)\} \ .$$
            To prove this statement, we will split the range over $\alpha$ up into three intervals: $I_1=[\tau,2\gamma(1-\gamma)], I_2=[2\gamma(1-\gamma),1-2\gamma(1-\gamma)]$  and $I_3 = [1-2\gamma(1-\gamma),1-\tau]$, so that we can rewrite the left hand side as:
            $$\max\left\{ \max_{\alpha\in I_1,\beta} \{\hat{r}^{(m-1)}_{A}(\alpha) + f_{DA}(\alpha,\beta,\gamma)\}, \max_{\alpha\in I_2,\beta} \{\hat{r}^{(m-1)}_{A}(\alpha) + f_{DA}(\alpha,\beta,\gamma)\}, \max_{\alpha\in I_3,\beta} \{\hat{r}^{(m-1)}_{A}(\alpha) + f_{DA}(\alpha,\beta,\gamma)\}\right\} \ ,$$
            where we still want to upper bound this expression as $\max_{\substack{\tau \leq \alpha \leq 1-\tau}} \{\hat{r}^{(m-1)}_A(\alpha) + f_A(\alpha,\gamma)\}$. To do so, we consider the three ranges of $\alpha$ in turn. 

            Now, we cannot hope to directly upper bound $f_{DA}(\alpha,\beta,\gamma) \leq f_A(\alpha,\beta,\gamma)$ across these three ranges. This is because, while all the points across these three intervals contain legal inputs $\alpha$ to $f_{DA}(\alpha,\beta,\gamma)$, this is not the case for $f_A(\alpha,\gamma)$. In fact none of the $\alpha \in I_3$ are valid inputs to $f_A(\alpha,\gamma)$ (we require $\alpha \leq \min\{2\gamma,2(1-\gamma)\}$ for $f_A(\alpha,\gamma)$ and this is smaller than $1-2\gamma(1-\gamma)$ which is the lower end of $I_3$), and only some of the points $\alpha\in I_2$ are valid inputs to $f_A(\alpha,\gamma)$. In other words: while the full range of valid inputs to the left hand side is exactly $I_1\cup I_2 \cup I_3$, the full range of valid inputs to the right hand side is contained in just $I_1 \cup I_2$. In the below, we may sometimes maximize the function $\alpha \mapsto f_A(\alpha,\gamma)$ over domains that include values of $\alpha$ for which $f_A(\alpha,\gamma)$. In such cases, the maximum is defined as the maximum over the subset of values $\alpha$ for which $f_A(\alpha,\gamma)$ is defined. 

            We will therefore try to establish, first, that 
            $$\max_{\alpha\in I_1,\beta} \{\hat{r}^{(m-1)}_{A}(\alpha) + f_{DA}(\alpha,\beta,\gamma)\} \leq \max_{\alpha\in I_1} \{\hat{r}^{(m-1)}_A(\alpha) + f_{A}(\alpha,\gamma)\} \ ,$$
            then the same for $I_2$, but then for $I_3$ we show instead that 
            $$\max_{\alpha\in I_3,\beta} \{\hat{r}^{(m-1)}_{A}(\alpha) + f_{DA}(\alpha,\beta,\gamma)\} \leq \max_{\alpha\in I_1} \{\hat{r}^{(m-1)}_A(\alpha) + f_{A}(\alpha,\gamma)\} \ .$$
            Put together, we can upper bound our original expression as:
            \begin{align*}
                &\max\left\{ \max_{\alpha\in I_1,\beta} \{\hat{r}^{(m-1)}_{A}(\alpha) + f_{DA}(\alpha,\beta,\gamma)\}, \max_{\alpha\in I_2,\beta} \{\hat{r}^{(m-1)}_{A}(\alpha) + f_{DA}(\alpha,\beta,\gamma)\}, \max_{\alpha\in I_3,\beta} \{\hat{r}^{(m-1)}_{A}(\alpha) + f_{DA}(\alpha,\beta,\gamma)\}\right\} \\ 
                &\leq \max\left\{ \max_{\alpha\in I_1} \{\hat{r}^{(m-1)}_A(\alpha) + f_{A}(\alpha,\gamma)\}, \max_{\alpha\in I_2} \{\hat{r}^{(m-1)}_{A}(\alpha) + f_A(\alpha,\gamma)\}, \max_{\alpha\in I_1} \{\hat{r}^{(m-1)}_{A}(\alpha) + f_A(\alpha,\gamma)\}\right\} \\
                &= \max_{\alpha\in I_1 \cup I_2} \{\hat{r}^{(m-1)}_A(\alpha) + f_{A}(\alpha,\gamma)\} \\
                &= \max_{\tau \leq \alpha \leq 1-2\gamma(1-\gamma)} \{\hat{r}^{(m-1)}_A(\alpha) + f_{A}(\alpha,\gamma)\} \\
                &= \max_{\tau \leq \alpha \leq \min\{2\gamma,2(1-\gamma)\}} \{\hat{r}^{(m-1)}_A(\alpha) + f_{A}(\alpha,\gamma)\} \\
                &=  \hat{r}^{(m)}_{A}(\tau,\gamma) \ .
            \end{align*}
            What remains is to prove the three claims about maximizing over $I_1, I_2$ and $I_3$. We will need the following upper bound which we proved as \Cref{thm:fda_bound}. 
            $$f_{DA}(\alpha,\beta,\gamma) \leq g(\alpha,\gamma)=
            \begin{cases} 
            f_A(\alpha,\gamma) & \text{if } \alpha \leq 2\gamma(1-\gamma) \ , \\
            h(\gamma)-h(\alpha) & \text{if } 2\gamma(1-\gamma) \leq \alpha \leq 1-2\gamma(1-\gamma) \ , \\
            f_A(1-\alpha,\gamma)  & \text{if } 1-2\gamma(1-\gamma) \leq \alpha \ . \\
            \end{cases}$$
            First, consider $I_1=[\tau,2\gamma(1-\gamma)]$. For $\alpha\in I_1$ we have the bound $f_{DA}(\alpha,\beta,\gamma)\leq f_A(\alpha,\gamma)$, which immediately gives us our result:
            $$\max_{\alpha\in I_1,\beta} \{\hat{r}^{(m-1)}_A(\alpha) + f_{DA}(\alpha,\beta,\gamma)\} 
            \leq \max_{\alpha\in I_1} \{\hat{r}^{(m-1)}_A(\alpha) + f_A(\alpha,\gamma)\}\ .$$
            Second, consider $I_3=[1-2\gamma(1-\gamma), \leq 1-\tau]$. For $\alpha\in I_3$ we have the bound $f_{DA}(\alpha,\beta,\gamma)\leq f_A(\alpha,\gamma)$ for $\alpha \in I_1$, which gives us:
            $$\max_{\alpha\in I_3,\beta} \{\hat{r}^{(m-1)}_A(\alpha) + f_{DA}(\alpha,\beta,\gamma)\} 
            % \leq \max_{\alpha\in I_3} \{\hat{r}^{(m-1)}_A(\alpha) + f_A(1-\alpha,\gamma)\} 
            = \max_{\alpha\in I_3} \{\hat{r}^{(m-1)}_A(1-\alpha) + f_A(1-\alpha,\gamma)\} 
            = \max_{\alpha\in I_1} \{\hat{r}^{(m-1)}_A(\alpha) + f_A(\alpha,\gamma)\} \ ,$$
            where the first equality follows from $\hat{r}^{(m-1)}_A(\alpha)=\hat{r}^{(m-1)}_A(1-\alpha)$ (see Proposition 3.1 in \cite{ravazzi_spectra_2009}), and the second equality follows because $I_1$ and $I_3$ are symmetric around $1/2$.

            Finally, consider $I_2=[2\gamma(1-\gamma), 1-2\gamma(1-\gamma)]$. For $\alpha\in I_2$ we have the bound $f_{DA}(\alpha,\beta,\gamma) \leq h(\gamma)-h(\alpha)$, so that we write 
            $$\max_{\alpha\in I_2,\beta} \{\hat{r}^{(m-1)}_A(\alpha) + f_{DA}(\alpha,\beta,\gamma)\} 
            \leq \max_{\alpha\in I_2} \{\hat{r}^{(m-1)}_A(\alpha) + h(\gamma) - h(\alpha)\}\ .$$
            We now claim that this is maximized with respect to $\alpha$ by one of the boundary points of the range. To see this, note first that $h(x) \geq h(x)/2 \geq \hat{r}^{(m-1)}_A(\alpha)$ (this is because the spectral shape function decreases with $m$, \cite[Proposition~4]{ravazzi_spectra_2009}), so that $\hat{r}^{(m-1)}_A(\alpha) - h(x)$ is non-negative. Next, we recall that $\hat{r}^{(m-1)}_A(\alpha)$ is increasing until $x=1/2$ and is symmetric around $1/2$~\cite[Propositions~3.1,~3.4]{ravazzi_spectra_2009}. This implies that the expression over $x$ is decreasing until $1/2$ and then increasing again, so that the maximizing value is either of the two boundary points. This means that we can upper bound the previous expression as
            $$ \leq \hat{r}^{(m-1)}_A(2\gamma(1-\gamma)) - h(2\gamma(1-\gamma)) + h(y) 
            = \hat{r}^{(m-1)}_A(2\gamma(1-\gamma)) +f_A(2\gamma(1-\gamma),y) 
            \leq \max_{\alpha\in I_2} \{\hat{r}^{(m-1)}_A(\alpha) + f_{A}(\alpha,\gamma)\} \ ,$$
            where the equality follows from \Cref{lem:h_fa}: this stated that $h(2\gamma(1-\gamma)) + f_A(2\gamma(1-\gamma),y) = h(y)$, which is of course the same as saying $f_A(2\gamma(1-\gamma),y) = h(y) - h(2\gamma(1-\gamma))$. The final inequality follows because our second-to-last expression is simply $\hat{r}^{(m-1)}_A(\alpha) +f_A(\alpha,\beta)$ evaluated at $\alpha=2\gamma(1-\gamma)$, so that maximizing over a larger range can only increase the value. 

            We note that the case for the dual code $R(AD)^m$ is completely analogous, as we have exactly the same upper bound on $f_{AD}$.
        \end{proof}

    \subsection{$(**)$ is negligible}
        \begin{theorem}\label{thm:**_negligible}
            Let $\epsilon>0$ and $m \geq 1$. The expected number of vectors in an $R(AD)^m$ code or $R(DA)^m$ code that have weight between $h$ and $n-h$ after $m-1$ rounds and weight $\leq (\delta^{(m)}-\epsilon)n$ after $m$ rounds is $\negl(n)$.
        \end{theorem}
        \begin{proof}
            Recall that the expression we need to bound to prove is what we called $(**)$ in \Cref{thm:main}, where we wrote $d=(\delta^{(m)}-\epsilon)n$. We express it below for $R(DA)^m$ codes.
            \begin{align*} 
                (**) &= \sum_{w_1=h}^{n-h}\sum_{w_2=1}^{n}\sum_{w_3=1}^{d} A_{w}(R(DA)^{m-1}) p_D(w_1,w_2) p_A(w_2,w_3) \\
                &\leq n^2 \sum_{w_3=1}^d \max_{\substack{h\leq w_1 \leq n-h \\ w_2}} \left\{ A_{w}(R(DA)^{m-1}) p_D(w_1,w_2) p_A(w_2,w_3) \right\}\ .
            \end{align*}
            To deal with this, we want to rewrite the above in terms of the restricted spectral shape function of the $R(DA)^m$ code. This means we have to bound $p_A$ and $p_D$ in terms of $f_A$ and $f_D$, and then we have to bound $A_w$ in terms of the spectral shape function $\hat{r}$. We recall how to bound the former from \Cref{sec:RMA}: we can first use \Cref{lem:A_D_claims} to upper bound $p_A(a,b) \leq O(1) \cdot p_A'(a,b)$ where the latter is simply $p_A(a,b)$ but with ceilings, floors and $-1$'s removed. It then follows easily that $p_A'(a,b) \leq 2^{n\cdot f_A(a,b)}$ using Stirling's approximation. We can do the same for $p_D(a,b) \leq O(1) \cdot p_D'(a,b)$. To bound $A_w$ in terms of $\hat{r}$, we recall that this relation was given as Lemma~5 in \cite{ravazzi_spectra_2009}, which prove that $A_w(RA^m) \leq (n+1)^{2m} \cdot 2^{n \cdot\hat{r}^{(m)}_A(w/n)}$. The case for $R(AD)^m$ and $R(DA)^m$ is analogous, yielding the bound $A_w(R(DA)^m) \leq \poly(n) \cdot 2^{n \cdot \hat{r}^{(m)}_{DA}(w/n)}$ and likewise for the dual code. We then have
            \begin{align*} 
                (**) &\leq \poly(n) \sum_{w_3=1}^d \max_{\substack{h\leq w_1 \leq n-h \\ w_2}} \left\{ 2^{n \left(\hat{r}^{(m-1)}_{DA}(w_1/n) + f_{DA}(w_1/n,w_2/n,w_3/n) \right)} \right\} \\
                &= \poly(n) \sum_{w_3=1}^d 2^{n \cdot \hat{r}^{(m)}_{DA}(h/n, w_3/n)} 
            \leq \poly(n) \sum_{w_3=1}^d 2^{n \cdot \hat{r}^{(m)}_{A}(h/n, w_3/n)} \ ,
            \end{align*}
            where we note that in the last step we use our main result: the fact that the restricted spectral shape function for $R(DA)^m$ codes is upper bounded by the one for $RA^m$ codes. Note that since we have the same result for $R(AD)^m$ codes, the argument follows analogously for that code. We now claim that the above is increasing with $w_3$ until $w_3=n/2$. This is known to be true for the regular spectral shape function of $RA^m$ codes from Proposition~3.4 in \cite{ravazzi_spectra_2009}, and again a completely analogous argument carries this result over to the restricted spectral shape function. We can therefore write 
            $$ (**) 
            \leq \poly(n) \cdot 2^{n \cdot \hat{r}^{(m)}_{A}(h/n, \delta^{(m)}-\epsilon)} \leq \poly(n) \cdot 2^{n \cdot -\Omega(h/n)} = \poly(n) \cdot 2^{-\Omega(h)} =n^{-\Omega(\log n)} = \negl(n) \ ,
            $$
            where we use the fact that $\hat{r}^{(m)}_A(h/n, \delta^{(m)}-\epsilon) \leq -\Omega(h/n)$, which we will prove now. We start by writing the left hand side out as $\max\limits_{\substack{h/n\leq \alpha\leq 1-h/n}} \{\hat{r}^{(m-1)}_A(\alpha) + f_A(\alpha,\delta^{(m)}-\epsilon)\}$. Note that $A$ restricts the range of $\alpha$, so that the upper bound on $\alpha$ becomes $\min\{2(\delta^{(m)}-\epsilon),2(1-(\delta^{(m)}-\epsilon))\}$, which is less than $1-h/n$ (for large enough $n$). In fact, since we will never target $\delta^{(m)}$ above $1/2$, the former must be the minimum, so that we can write $\alpha \leq 2(\delta^{(m)}-\epsilon)$.
    
            We know that $\hat{r}^{(m)}(\delta^{(m)})=0$, so in particular, for any $\alpha$ we have that $\hat{r}^{(m-1)}_A(\alpha) + f_A(\alpha, \delta^{(m)})  \leq 0$. Recall from \Cref{lem:fa_fd_claims}.4 that $f_A(\alpha,\beta)$ decreases when $\beta$ decreases (for $\beta \leq 1/2$). However, this decrease is only strict when $\alpha>0$ (as $f_A(0,\beta)=0$ regardless of $\beta$). Since we have a lower bound $h/n \leq \alpha$, our expression must decrease when changing $\delta^{(m)}$ to $\delta^{(m)}-\epsilon$: $\hat{r}^{(m)}_A(h/n, \delta^{(m)}-\epsilon) < 0$.
    
            However, this doesn't tell us how large this decrease is. It turns out that this decrease depends on the value of $\alpha$ (it behaves roughly like $-\Omega(\alpha)$). Recall that the full range over $\alpha$ is  $h/n \leq \alpha \leq 2(\delta^{(m)}-\eps)$. We can make a simple observation here: if $\delta^{(m-1)}<\alpha$, then our expression $\hat{r}^{(m-1)}_A(\alpha) + f_A(\alpha, \delta^{(m)}) < 0$, which we know is strictly negative, is independent of $n$ (as the only dependence on $n$ occurs when $\alpha$ is smaller). Hence, the strict negativity must also be independent of $n$, meaning it is negative constant, which is definitely smaller than $-\Omega(h/n)$. That then just leaves the range $h/n<\alpha<\delta^{(m-1)}$. Here, $\alpha$ clearly depends on $h$. We will show that the decrease caused by subtracting $\epsilon$ from $\delta^{(m)}$ is minimal when $\alpha$ is minimal, and in that case is equal to $-\Omega(h/n)$. To see this, note that the derivative of $f_A(\alpha,\beta)$ to $\beta$ is: 
            \begin{align*}
                \frac{\partial f_A(\alpha,\beta)}{\partial \beta} &= \log \left(\frac{\alpha+2 \beta-2}{2 (\alpha-1)}\right)-\log \left(\frac{\alpha-2 \beta}{2 (\alpha-1)}\right)-\log (1-\beta)+\log (\beta) \\
                &= \log \left(\frac{1-\beta-\frac{\alpha}{2}}{1-\alpha}\right)-\log \left(\frac{\beta-\frac{\alpha}{2}}{1- \alpha}\right) + \log\left(\frac{\beta}{1-\beta}\right)\\
                &= \log \left(\frac{1-\beta-\frac{\alpha}{2}}{\beta-\frac{\alpha}{2}}\right) + \log\left(\frac{\beta}{1-\beta}\right) \ .
            \end{align*}
            We want to show that the above derivative $\frac{\partial f_A(\alpha,\beta)}{\partial \beta}$ is increasing with $\alpha$, thus we need to show that the derivative of that expression, but now to $\alpha$, is always positive. The derivative of $\frac{\partial f_A(\alpha,\gamma)}{\partial \beta}$ to $\alpha$ is $\frac{1}{\log2 }\frac{2-4 \beta}{(\alpha-2 \beta) (\alpha+2 \beta-2)}$. We thus need to have 
            $$\frac{1}{\log2 }\frac{2-4 \beta}{(\alpha-2 \beta) (\alpha+2 \beta-2)} > 0 \Longleftarrow 2-4 \beta > 0 \iff 1/2 > \beta \ .$$
            Thus, if we assume $\beta \leq 1/2$ (which we always do, since $\beta \leq \delta^{(m)}-\epsilon \leq 1/2$), then we can safely assume that the decrease caused when going from $f_A(\alpha,\delta^{(m)})$ to $f_A(\alpha,\delta^{(m)}-\epsilon)$ is smallest when $\alpha$ is smallest. Since the smallest value of $\alpha$ is $h/n$, we should fix $\alpha=h/n$. What remains to be shown then is that the decrease caused when going from $f_A(h/n,\delta^{(m)})$ to $f_A(h/n,\delta^{(m)}-\epsilon)$ is $\Omega(h/n)$. For this, we need to look at the derivative of $f_A(h/n,\beta)$ to $\beta$, and argue that it is $\Omega(h/n)$ for all $\beta$. This means that for any choice of $\delta^{(m)}$, we know that decreasing $\delta^{(m)}$ by a constant $\epsilon$ causes a decrease of $\Omega(h/n)$ as desired. We already computed this derivative above for general $\alpha$, let us now fill in $\alpha=h/n$ and simplify.
            $$ \frac{\partial f_A(h/n,\beta)}{\partial \beta} = \log \left(\frac{1-\beta-\frac{h}{2n}}{\beta-\frac{h}{2n}}\right) + \log\left(\frac{\beta}{1-\beta}\right) 
            = \log \left(\frac{1-\frac{h}{2n(1-\beta)}}{1-\frac{h}{2n\beta}}\right) 
            = \log\left(\left(1-\frac{h}{2n(1-\beta)}\right) \sum_{i=0}^{\infty}\left(\frac{h}{2n\beta}\right)^i\right)  \ .
            $$
            In the last step we used the known equality $\frac{1}{1-\epsilon} = \sum_{i=0}^{\infty}\epsilon^i$. As the terms in this sum as positive, we can lower bound by dropping all but the first two terms, which gives us:
            $$\log\left(\left(1-\frac{h}{2n(1-\beta)}\right)\left(1+\frac{h}{2n\beta}\right)\right)
            = \log\left(1 + \frac{h}{n} \left(\frac{1}{2\beta}-\frac{1}{2(1-\beta)} - \frac{h}{n}\frac{1}{4\beta(1-\beta)}\right)\right) 
            = \log\left(1 + \frac{h}{n} \left(\frac{2-4\beta -h/n}{4\beta(1-\beta)}\right)\right) \ . $$
            We lower bound this expression using the inequality $\ln(1+x) \geq \frac{x}{1+x}$:
            $$\frac{1}{\ln 2}\frac{\frac{h}{n} \left(\frac{2-4\beta -h/n}{4\beta(1-\beta)})\right)}{1+\frac{h}{n} \left(\frac{2-4\beta -h/n}{4\beta(1-\beta)})\right)} 
            = \frac{h}{n} \frac{1}{\ln 2}\frac{2-4\beta -h/n}{4\beta(1-\beta)\left(1+\frac{h}{n}(2-4\beta -h/n)\right)} \geq \frac{h}{n} \frac{1}{8 \ln 2 \beta(1-\beta)} \geq C \cdot \frac{h}{n} \ ,
            $$
            where the first inequality follows from increasing $h/n\leq 1$ to 1 in the denominator, which cancels the factors $(2-4\beta -h/n)$. Note that we indeed think of $C\leq \frac{1}{8 \ln 2 \beta(1-\beta)}$ as a constant as it depends only on $\beta=\delta^{(m)}-\epsilon$, our target minimum distance, which does not depend on $n$. This establishes the theorem, as it guarantees a decrease of $\Omega(h/n)$ when we subtract $\epsilon$ from $\delta^{(m)}$. In principle, the constant $C$ hidden by $\Omega(h/n)$ could vary with $\beta$, so that subtracting $\epsilon$ from $\delta^{(m)}$ doesn't guarantee a decrease of at least $C \cdot \epsilon \cdot h /n$. This would occur when $C$ would increase with $\beta$, so that the decrease at $\beta=\delta^{(m)}$ is higher than at some $\delta^{(m)}-\epsilon<\beta<\delta^{(m)}$. However, we note that this is not in fact the case, as $\frac{1}{8 \ln 2 \beta(1-\beta)}$ is decreasing with $\beta$ up to $\beta = 1/2$. This means we can guarantee that $\hat{r}^{(m)}_A(h/n, \delta^{(m)}-\epsilon)) \leq -\frac{\epsilon}{8 \ln 2 \delta^{(m)}(1-\delta^{(m)})} h/n = -\Omega(h/n)$.

            % This is the old proof involving L1 and L2.
            % We split this range over $x$ up in two parts. First, consider just $h/n \leq x \leq \delta^{(m-1)}$, calling its maximum value $L_1$. By definition, $\hat{r}^{(m-1)}$ is at most $0$ on these inputs. As such, all that's left of the expression is the function $f_A(x,\delta^{(m)}-\epsilon)$. Recall from \Cref{lem:fa_fd_claims}.4 that in case $y$ is a positive constant below $1/2$, we have the bound $f_A(\alpha,\beta)\leq -\Omega(x)$. Since we assume $m\geq 3$, we indeed have that $y$ is a positive constant below $1/2$, so that we can conclude that $L_1 \leq -\Omega(\frac{h}{n})$.
    
            % Second, consider the rest of the range over $x$. Here, we have to bound:
            % $$L_2:= \max_{\delta^{(m-1)} < x < 2(\delta^{(m)}-\epsilon)} \{ \hat{r}^{(m-1)}_A(\alpha) + f_A(x, \delta^{(m)} - \epsilon)\} \ .$$
            % Recall that $f_A(\alpha,\beta)$ is increasing in $y$ until $y=1/2$, and it is strictly increasing assuming $x>0$. Recall also that, by definition, $\hat{r}^{(m)}_A(\delta^{(m)})=0$. We can then write that for all $\delta^{(m-1)} < x < 2(\delta^{(m)}-\epsilon)$:
            % $$0 \geq \hat{r}^{(m-1)}_A(\alpha) + f_A(x, \delta^{(m)}) > \hat{r}^{(m-1)}_A(\alpha) + f_A(x, \delta^{(m)}-\epsilon)\ .$$
            % Thus, $L_2=-\Omega(1)$ is a negative constant depending only on $m$ and $\epsilon$. Asymptotically, the maximizer is then $L_1=-\Omega(h/n)$, as this is a negative value that goes to 0 with $n$, while $L_2$ is always a negative constant.
        \end{proof}

\section{Bounding $(*)$: vectors with weight near boundary of weight range before last round of encoding} \label{sec:bounding-edge}
    Our goal in this section is to bound $(*)$: the expected number of low weight codewords which entered the final round of the encoding with a weight very close to the endpoints of the weight range. We need to do this for both $R(AD)^m$ and $R(AD)^m$ codes. The section is structured as follows. First, just below, we give an overview of the proof where we bound $(*)$ for either code, making use of a number of lemmas that we prove later. Second, we prove a number of useful inequalities involving $p_A$ and $p_D$ which will be needed. Third, we prove \Cref{lem:low_weight_da} and \Cref{lem:middle_weight_da} which we use in the prove bounding $(*)$ for $R(DA)^m$ codes. Fourth, we prove \Cref{lem:low_weight_ad}, \Cref{lem:middle_weight_1_ad} and \Cref{lem:middle_weight_2_ad}, which are likewise used in the proof bounding $(*)$ for $R(AD)^m$ codes. 

    Recall that in this section, unlike in previous sections, weight by default refers to \emph{absolute} weight, i.e., the number of nonzero coordinates in the vector. For the fraction of nonzero entries, we will use \emph{relative} weight. 

    \begin{theorem}\label{thm:*_small}
        Let $\epsilon>0$. The expected number of vectors in an $R(DA)^m$ code that have weight between $\leq h$ or $\geq n-h$ after $m-1$ rounds and weight $\leq (\delta^{(m)}-\epsilon)n$ after $m \geq 2$ rounds is $\tilde O (n^{1/4-m/2})=1/\poly(n)$. Likewise, for $R(AD)^m$ codes with $m \geq 3$ we expected at most $\tilde O(n^{1/4-m/4})=1/\poly(n)$ such vectors. 
    \end{theorem}
    \begin{proof}
        Recall that the expression we need to bound is what we called $(*)$ in \Cref{thm:main}, where we wrote $d=(\delta^{(m)}-\epsilon)n$. We express it below for $R(DA)^m$ codes, and we simplify it by counting \textit{all} codewords with weight near the boundaries of the weight range after $m-1$ rounds (as opposed to only codewords which also have low weight after the entire encoding):
        $$ (*) = \sum_{\substack{1\leq w_1 \leq h \lor \\ n-h\leq w_1 \leq n}}\sum_{w_2=1}^{n}\sum_{w_3=1}^{d}A_{w_1}(R(DA)^{m-1})p_D(w_1,w_2) p_A(w_2,w_3) \leq \sum_{\substack{1\leq w_1 \leq h \lor \\ n-h\leq w_1 \leq n}} A_{w_1}(R(DA)^{m-1}) \ .$$ 

        To bound the above, we will split it into two cases based on the weight of a vector in earlier rounds. The point is that the vectors with weight $\leq h$ or $\geq n-h$ after $m-1$ rounds can have one of two sources. Either, they had weight $\leq h$ or $\geq n-h$ before all preceding $DA$ or $AD$ rounds. Or, there was some earlier $i$'th $AD$ or $DA$ round before which the vector had weight in the middle of the weight range, and after which it dropped back down to weight $\leq h$ or $\geq n-h$, where it stayed until after round $m-1$. To describe this more formally, let $A_{w,\leq h}(R(DA)^m)$ denote the expected number of weight $w$ codewords in an $R(DA)^m$ code, which before each $AD$ round had weight $\leq h$ or $\geq n-h$. Likewise, let $A_{w,\geq h}(R(DA)^m)$ denote the expected number of weight $w$ codewords in an $R(DA)^m$ code, which had weight between $h$ and $n-h$ before round $m$, but had weight below $h$ or above $n-h$ after round $m$. We can then claim we have the following upper bound:
        $$ (*) \leq \sum_{\substack{1\leq w \leq h \lor \\ n-h\leq w \leq n}} A_{w,\leq h}(R(DA)^{m-1}) + \sum_{i=1}^{m-1}\sum_{\substack{1\leq w \leq h \lor \\ n-h\leq w \leq n}} A_{w, \geq h}(R(DA)^{i}) \ .$$

        To bound the above, we will prove in \Cref{lem:low_weight_da} that we can bound the first term as follows assuming $m-1\geq 1$, so that we can instantiate this theorem with $m \geq 2$:
        $$\sum_{\substack{1\leq w \leq h \lor \\ n-h\leq w \leq n}} A_{w,\leq h}(R(DA)^{m-1})  \leq \tilde O\left(n^{1/4-(m-1)/2}\right) = 1/\poly(n) \ .$$
        In addition, we will see in \Cref{lem:middle_weight_da} that we can bound the second case as follows, and that this also works once $m-1 \geq 1$, so that the theorem indeed holds for $R(DA)^m$ codes once $m \geq 2$:
        $$ \sum_{i=1}^{m-1}\sum_{\substack{1\leq w \leq h \lor \\ n-h\leq w \leq n}} A_{w, \geq h, i}(R(DA)^{m-1}) \leq \sum_{i=1}^{m-1} \negl(n) = \negl(n) \ .$$
        We then have the final bound $(*) \leq O(n^{1/4-m/2}) = 1/\poly(n)$ for $R(DA)^m$ codes, and remark again that this is true for $m \geq 2$. Let us now turn towards $R(AD)^m$ codes, where we will have to tackle things slightly differently. The first case is mostly the same: we now prove \Cref{lem:low_weight_ad} which lets us bound 
        $$\sum_{\substack{1\leq w \leq h \lor \\ n-h\leq w \leq n}} A_{w,\leq h}(R(AD)^{m-1}) \leq \tilde O(n^{1/4-m/4}) = 1/\poly(n) \ ,$$
        where the upper bound is slightly worse, but more importantly, the result only holds for $m-1 \geq 2$, so that this limits this theorem to hold only for $R(AD)^m$ codes with $m \geq 3$. The second case will be more different. Through \Cref{lem:middle_weight_1_ad}, we are still able to bound $\sum_{\substack{1\leq w \leq h \lor \\ n-h\leq w \leq n}} A_{w, \geq h}(R(DA)^{i}) = \negl(n)$ for $i \geq 2$. But, crucially, this ends up being false for $i=1$. We therefore prove something slightly different in \Cref{lem:middle_weight_2_ad}. We don't argue that the expected number of middle weight vectors before round 1 that go to the boundaries of the weight range after round 1 to be $\negl(n)$, as this fails to be true. Instead, we track these vectors also to round 2, requiring these middle weight vectors before round 1 to be near the boundaries of the weight range after round 1 and then also after round 2. This quantity we can prove as being $\negl(n)$. This way, we still guarantee that the middle weight vectors before round 1 aren't problematic, as after round 2 only a negligible number of them are near the boundaries. The only downside is that this again requires us to take $m \geq 3$, but we recall that the previous case already forced us to do so for $R(AD)^m$ anyway. With this, we conclude that for $R(AD)^m$ codes we have the bound $(*) \leq O(n^{1/4-m/4}) = 1/\poly(n)$ which is true for $m \geq 3$.
    \end{proof}

    As for the useful inequalities, to prove them we first need another definition. Essentially, we will redefine $p_A(a,b)$ and $p_D(a,b)$ but with the floors, ceilings and $-1$'s removed. This means that some of the values in the binomial coefficients can be non-integers, specifically they can lie halfway in-between two integers. We take the standard extension of the binomial coefficient to non-integer values in terms of the $\Gamma$-function, which itself is a generalization of the factorial function to non-integers. As said, we only care about values halfway in between two integers, and therefore recall only the value of $\Gamma$ for such cases: $(n-1/2)! := \Gamma(n+1/2) = \frac{(2n)! \sqrt{\pi}}{4^nn!}$.
    \begin{definition}
        $$p_A'(a,b) = \frac{\binom{n-b}{a/2}\binom{b}{a/2}}{\binom{n}{a}} \ ,\ \  p_D'(a,b) = \frac{\binom{n-a}{b/2}\binom{a}{b/2}}{\binom{n}{a}} \ .$$
    \end{definition}            
    The point of this definition is to obtain the symmetries $p_A'(a,b)=p_A'(a,n-b)$ and $p_D'(a,b)=p_D'(n-a,b)$. The analogues of these symmetries for the asymptotic counterparts of these probabilities, $f_A$ and $f_D$, proved very useful in the analysis of $(**)$. The same will be true of these. However, it is not obvious how the size of $p_A'$ and $p_A$ are related. It turns out they differ by only a small constant factor, allowing us to always swap $p_A$ with $p_A'$ (and likewise for $D$). This is proven in the following lemma, along with further useful claims about $p_A'$ and $p_D'$. Note also that $p_A'$ has a slightly larger domain that $p_A$, e.g. admitting $p_A'(a,a/2)$ while otherwise we would need to have the slightly larger second argument $p_A(a,\lceil a/2 \rceil)$. The proof of this lemma is provided in \Cref{sec:deferred-proofs}.

    \begin{restatable}{lemma}{adclaims}\label{lem:A_D_claims}
        \begin{enumerate}
            \item $p_A(a,b)$ is defined for $\lfloor a/2 \rfloor \leq n-b$ and $\lceil a/2 \rceil \leq b$ or equivalently, $\lceil a/2 \rceil \leq b \leq n-\lfloor a/2 \rfloor$. $p_A'$ has the same domain, but with the floors and ceilings removed.
            \item $p_D(a,b)$ is defined for $\lfloor b/2 \rfloor \leq n-a$ and $\lceil b/2 \rceil \leq a$ or equivalently, $\lceil b/2 \rceil \leq a \leq n-\lfloor b/2 \rfloor$. $p_D'$ has the same domain, but with the floors and ceilings removed.
            \item $\left(\frac{n}{k}\right)^k\leq\binom{n}{k}\leq \left(\frac{en}{k}\right)^k$.
            \item $p_D'(a,b) = p_D'(n-a,b)$.
            \item $p_A'(a,b) = p_A'(a,n-b)$.
            \item $p_D(a,b) \leq 2.17 p_D'(a,b)$.
            \item $p_A(a,b) \leq 2.17 p_A'(a,b)$.
            \item $p_A'(a,b)$ is increasing for $1 \leq b \leq n/2$ and decreasing afterwards.
            \item $p_A'(a,b)$ is maximized over $a$ by some $a \leq n/2$.
            \item $p_A'(a,b) \leq \left(\frac{4b}n\right)^{a/2}$.
            \item $p_D(a,b) \leq \left(\frac{2e\sqrt{an}}{b}\right)^{b}/\binom{n}{a}$.
        \end{enumerate}
    \end{restatable}

    \subsection{$R(DA)^m$ codes}
        We first bound what we called case one in the introduction to this section: the expected number of $R(DA)^m$ codewords with weight close to the boundaries of the weight range between all the $DA$ rounds. 
        
        \begin{lemma}\label{lem:low_weight_da}
            Let $m\geq 1$. In an $R(DA)^m$ code of rate $1/2$, the expected number of vectors whose weight is in the range $[0, h] \cup [n-h, n]$ before and after each $DA$ operation is at most $\tilde{O}(n^{1/4-m/2}) = 1/\poly(n)$.
        \end{lemma}
        \begin{proof}                
            We will prove this by induction on $m$. For the base case, $m=1$. Our goal is then to bound the expected number of vectors that have weight $2w_1 \leq h$ or $n-h \leq 2w_1$ going into the first $DA$ round, which stay in that weight range $w_3 \leq h$ or $n-h \leq w_3$ after applying $DA$:
            $$ \sum_{\substack{w_1 \leq h/2,\\(n-h)/2 \leq w_1}} \sum_{w_2=1}^{2 h}\sum_{\substack{w_3 \leq h,\\n-h \leq w_3}}\binom{n/2}{w_1}p_D(2w_1,w_2)p_A(w_2,w_3) \ ,$$
            where we note that the upper bound on the range of $w_2$ is due to \Cref{lem:A_D_claims}.1. Let us first apply the inequality $p_A(a,b) \leq O(1) p_A'(a,b)$ proven as \Cref{lem:A_D_claims}.7 such that $p_A'(a,b)=p_A'(a,n-b)$, and likewise for $p_D$. The result of this will be that any bound on the expected number of vectors that end up with weight $w_3 \leq h$ is also a valid bound on the expected number of vectors that end up with weight $n-h \leq w_3$. Likewise, any bound on the expected number of vectors that start with $2w_1 \leq h$ will also apply to the vectors that have weight $n-h \leq w_1$. Combining this, it will suffice to bound the expected number of vectors that have weight $2w_1 \leq h$ going into the $DA$ operation, that end up with weight $w_3 \leq h$ after the $DA$ operation: this will suffice as an upper bound for the entire base case, up to a constant factor, so we get the upper bound
            $$ \leq O(1)\sum_{w_1=1}^{h/2}\sum_{w_2=1}^{2 h}\sum_{w_3=1}^{h}\binom{n/2}{w_1}p_D'(2w_1,w_2)p_A'(w_2,w_3) 
            \leq O(1)h \sum_{w_1=1}^{h/2}\sum_{w_2=1}^{2 h}\binom{n/2}{w_1}p_D'(2w_1,w_2) \left(\frac{4h}{n}\right)^{w_2/2} \ .$$
            The inequality here follows from applying \Cref{lem:A_D_claims}.10 to $p_A'$ and noting that \Cref{lem:A_D_claims}.8 tells us that $p_A'(a,b)$ grows with $b$ until $b = n/2$, so we can maximize by setting $w_3=h$. Next, let us also apply \Cref{lem:A_D_claims}.10 to bound $p_D'$ and use \Cref{lem:A_D_claims}.3 to bound the leftover binomial coefficients:
            $$ %\leq O(1)h \sum_{w_1=1}^{h/2}\sum_{w_2=1}^{2 h} \frac{\binom{n/2}{w_1}}{\binom{n}{2w_1}} \left(\frac{2e\sqrt{2w_1n}}{w_2}\right)^{w_2} \left(\frac{4h}{n}\right)^{w_2/2} 
            \leq O(h) \sum_{w_1=1}^{h/2}\sum_{w_2=1}^{2 h}\left(\frac{en}{2w_1}\right)^{w_1}\left(\frac{2w_1}{n}\right)^{2w_1} \left(\frac{2e\sqrt{2w_1n}}{w_2}\right)^{w_2} \left(\frac{4h}{n}\right)^{w_2/2}
            = O(h) \sum_{w_1=1}^{h/2}\sum_{w_2=1}^{2 h}\left(\frac{2ew_1}{n}\right)^{w_1} \left(\frac{4e\sqrt{2w_1h}}{w_2}\right)^{w_2} \ .$$
            We bound $(2w_1)^{w_2} \leq h^{w_2}$. Then the only occurrence of $w_1$ is the first factor, which is clearly decreasing with $w_1$, as $2ew_1/n<1$ (for large enough $n$). We can then set $w_1$ to its minimal value of $w_1=w_2/4$ to maximize the expression:
            $$ \leq O(1)h^2 \sum_{w_2=1}^{2 h}\left(\frac{ew_2}{n}\right)^{w_2/4} \left(\frac{4eh}{w_2}\right)^{w_2} 
            = O(1)h^2 \sum_{w_2=1}^{2 h} \left(\frac{4e^{5/4}h}{w_2^{3/4}n^{1/4}}\right)^{w_2} 
            \leq O(1)\frac{h^4}{n^{1/4}} = \tilde O(n^{-1/4}) \ ,$$
            where we use the fact that the final expression decreases with $w_2$ (as the $1/n^{1/4}$ guarantees that the base of the exponent is $<1$ for large enough $n$), so that we can maximize by fixing $w_2=1$.

            For the induction step, let us write $A_{w, \leq h}(R(DA)^m)$ to denote the expected number of vectors of weight $w$ in an $R(DA)^m$ code which had weight in range $[0, h] \cup [n-h, n]$ before and after each $DA$ operation. Recalling the symmetry argument from before, we can then quantify the expression we need to bound as follows, where again let ourselves say that $w_2 \leq 2h$:
            \begin{align*}
                \sum_{\substack{w_1 \leq h\\n-h \leq w_1}} A_{w_1, \leq h}(R(DA)^m) 
                &= \sum_{\substack{w_1 \leq h\\n-h \leq w_1}} A_{w_1, \leq h}(R(DA)^{m-1}) \sum_{w_2=1}^{n}\sum_{\substack{w_3 \leq h\\n-h \leq w_3}} p_D(w_1,w_2)p_A(w_2,w_3) \\
                &\leq O(1) \sum_{\substack{w_1 \leq h\\n-h \leq w_1}} A_{w_1, \leq h}(R(DA)^{m-1}) \sum_{w_2=1}^{2 h}\sum_{w_3=1}^{h} p_D'(w_1,w_2)p_A'(w_2,w_3) \ .
            \end{align*}
            Our goal is to bound this entire expression, where we note that the induction hypothesis lets us bound the factor $\sum_{w_1 \leq h, n-h \leq w_1} A_{w_1, \leq h}(R(DA)^{m-1})$. Of course, we cannot apply the induction hypothesis to the expression as written, as the sum over $w$ also contains the additional factors written after. To deal with this, we will bound these other factors in the sum over $w$ as $\tilde{O}(1/\sqrt{n})$, which we achieve by simply throwing away $p_D'(a,b)\le 1$ and bounding $p_A'(a,b)$ in the same way as we just did for the base case:
            \begin{align*}
                O(1) \sum_{w_2=1}^{2h}\sum_{w_3=1}^{h} &p_D'(w,w_2)p_A'(w_2,w_3) \leq O(1) \sum_{w_2=1}^{2h} \sum_{w_3=1}^{h}\left(\frac{4w_3}{n}\right)^{w_2/2}\\ 
                &\leq O(1)h \sum_{w_2=1}^{2h} \left(\frac{4h}{n}\right)^{w_2/2}  \leq O(1) \frac{h^{2.5}}{n^{1/2}} = \tilde O(n^{-1/2})\ . 
            \end{align*}
            Thus, the base case gives us $\tilde{O}(n^{1/4})$ while the induction step tells us that each extra round cuts off a factor $\tilde{O}\left(n^{-1/2}\right)$, so the final bound becomes $\tilde O(n^{-1/8-(m-1)/2}) = \tilde O(n^{1/4-m/2})$, as desired. 
        \end{proof}
        
        We next bound what we called case two in the introduction to this section: the expected number of $R(DA)^m$ codewords with weight in the middle of the weight range after $m-1$ rounds, to then drop the weight to close to the boundary of the weight range after the final round.
        
        \begin{lemma}\label{lem:middle_weight_da}
            Let $m\geq 1$. In a $R(DA)^m$ code of rate $1/2$, the expected number of vectors that have weight between $h$ and $n-h$ after $m-1$ rounds and weight $\leq h$ or $\geq n-h$ after $m$ rounds is $n^{-\Omega(h)}$.
        \end{lemma}
        \begin{proof}
            For $m\geq 2$ this follows from \Cref{thm:**_negligible}. Recall that this states that we expect only $\exp(-\Omega(h))$ vectors to go from having weight between $h$ and $n-h$ after $m-1$ rounds to weight below $\delta^{(m)}-\epsilon$ after the $m$'th round. Since $\delta^{(m)}$ is a positive constant for $m\geq 2$, clearly $h/n \leq \delta^{(m)}-\epsilon$ for $m\geq 2$ as well (for sensible choices of $\epsilon)$. We note that Proposition~3.1 in \cite{ravazzi_spectra_2009} said that the spectral shape function for an $RA^m$ code is symmetric around $1/2$, and that \Cref{lem:fa_fd_claims}.4 stated that $f_A(\alpha,\beta)$ is symmetric around $\beta=1/2$. It follows that the restricted spectral shape function $\hat{r}_A^{(m)}(\tau,\gamma)$ is also symmetric around $\gamma=1/2$, and an analogous argument to the one given in the proof of \Cref{thm:**_negligible} tells us that the expected number of weight $\geq n-h$ vectors must also be negligible in $n$. 
            
            This just leaves $m=1$, which asks us to bound the following as something negligible in $n$, 
            $$\sum_{w_1=h/2}^{(n-h)/2}\sum_{w_2=1}^{4w_1}\sum_{\substack{w_3 \leq h\\n-h \leq w_3}} \binom{n/2}{w_1} p_D(2w_1,w_2) p_A(w_2,w_3) \ ,$$
            where we note that the upper bound on the range of $w_2$ is due to \Cref{lem:A_D_claims}.2. Notice that the expression within the sums is the exact same as in the previous \Cref{lem:low_weight_da}; the only difference is the range over $w_1$. We will therefore copy the first few inequalities we applied in that previous lemma, crucially, include the substitution of $p_A\leq p_A'$ which let us limit ourselves to just $w_3 \leq h$, and not also $n-h \leq w_3$. This lets us bound the above as 
            $$\leq O(1)h\sum_{w_1=h/2}^{(n-h)/2}\sum_{w_2=1}^{4w_1} \left(\frac{2ew_1}{n}\right)^{w_1} \left(\frac{4e\sqrt{2w_1h}}{w_2}\right)^{w_2} 
            \leq O(1) h^2 \sum_{w_1=h/2}^{(n-h)/2} \left(\frac{2ew_1}{n}\right)^{w_1} \left(\frac{4e\sqrt{2w_1h}}{4w_1}\right)^{4w_1} \ ,$$
            where the second inequality follows because the expression is growing in $w_2$ (as $4e\sqrt{2w_1h} \geq 4e\sqrt{2h/2h} = 4eh > w_2$). Further simplifying gives us
            $$= O(1) h^2 \sum_{w_1=h/2}^{(n-h)/2} \left(\frac{2ew_1}{n}\left(\frac{e\sqrt{2h}}{\sqrt{w_1}}\right)^4\right)^{w_1} 
            = O(1) h^2 \sum_{w_1=h/2}^{(n-h)/2} \left(\frac{2^3e^5h^2}{nw_1}\right)^{w_1} 
            \leq  O(1) nh^2\left(\frac{2^3e^5h^2}{nh/2}\right)^{h/2} = n^{-\Omega(h)}\ ,$$
            where the inequality follows because the expression is decreasing with $w_1$ (as the $1/n$ eventually causes the base of the exponent to drop below $1$) and $h/2 \leq w_1$.
        \end{proof}
    
    \subsection{$R(AD)^m$ codes}
        We need to deal with the same two cases as we did for the dual code. We will need the following two lemmas, which will tell us what the maximizing value over $w_2$ is in the expression $p_A(w_1,w_2)p_D(w_2,w_3)$ for all the different ranges over $w_1$ and $w_3$ that we will encounter. We prove these two lemmas in \Cref{sec:deferred-proofs}.
    
        \begin{restatable}{lemma}{wwh}\label{lem:w2_w4_h}
            Fix $a,b,c \in \{1,2,\dots,n\}$ such that $(a,b)$ is a valid input to $p_A'$ and $(b,c)$ is valid input to $p_D'$ (that is, $\max\{a/2,c/2\} \leq b \leq n-\max\{a/2,c/2\}$). If $\max\{a, c\}=o(n)$ then, for sufficiently large $n$, a value of $b$ as above maximizing the expression $p_A'(a,b)p_D'(b,c)$ is $b=\max\{a/2,c/2\}$. 
        \end{restatable}
        \begin{restatable}{lemma}{wwnh}\label{lem:w2_w4_n-h} 
            Fix $a,b,c \in \{1,2,\dots,n\}$ such that $(a,b)$ is a valid input to $p_A'$ and $(b,c)$ is valid input to $p_D'$ (that is, $\max\{a/2,c/2\} \leq b \leq n-\max\{a/2,c/2\}$). If either $a=o(n), c=n-o(n)$ or $a=n-o(n),c=o(n)$ then, for sufficiently large $n$, a value of $b$ as above maximizing the expression $p_A'(a,b)p_D'(b,c)$ is $b = n/2$.
        \end{restatable}

        We first bound what we called case one in the introduction to this section: the expected number of $R(AD)^m$ codewords with weight close to the boundaries of the weight range between all the $AD$ rounds. Unlike with the dual code, we are not able to bound this quantity as an inverse polynomial in the block-length $n$ when $m=1$. When we aim to prove an $R(AD)^m$ code has good distance, we will need to employ the below lemma with $m-1$. Hence, the requirement for $m \geq 2$ in the below lemma limits us to only get good $R(AD)^m$ codes for $m \geq 3$.
        
        \begin{lemma}\label{lem:low_weight_ad}
            Let $m\geq 2$. In an $R(AD)^m$ code of rate $1/2$, the expected number of vectors whose weight is in the range $[0, h] \cup [n-h, n]$ before and after each $AD$ operation is at most $\tilde O(n^{1/4-m/4})=1/\poly(n)$.
        \end{lemma}
        \begin{proof} 
            As we did in the analogous \Cref{lem:low_weight_da} for the dual code, we will prove this by induction on $m$. For the base case, $m=2$. Our goal is then to bound the expected number of vector that have weight $2w_1 \leq h$ or $n-h \leq 2w_1$ going into the first $AD$ round, which stay in that weight range $w_3 \leq h$ or $n-h \leq w_3$ after applying $AD$, and again after applying $AD$ again: $w_5 \leq h$ or $n-h \leq w_5$. We get the following expression.
            $$ \sum_{\substack{w_1 \leq h/2\\(n-h)/2 \leq w_1}}\sum_{w_2=\max\{w_1,\lceil w_3/2 \rceil\}}^{n-\max\{w_1,\lfloor w_3/2 \rfloor\}}\sum_{\substack{h \leq w_3\\n-h \leq w_3}}\sum_{w_4=\max\{\lceil w_3/2 \rceil, \lceil w_5/2 \rceil\}}^{n-\max\{\lfloor w_3/2 \rfloor, \lfloor w_5/2 \rfloor\}}\sum_{\substack{h \leq w_5\\n-h \leq w_5}} \binom{n/2}{w_1}p_A(2w_1,w_2)p_D(w_2,w_3)p_A(w_3,w_4)p_D(w_4,w_5) \ .$$
            To deal with this expression, we start by applying the inequalities $p_A(a,b)\leq O(1) \cdot p_A'(a,b)$, and likewise for $p_D$, as this allows us to write use the fact that $p_A'(a,b)=p_A'(a,n-b)$ and $p_D'(a,b)=p_D'(n-a,b)$. This implies that $w_2$ and $w_4$ are symmetric around $n/2$, meaning we can assume that $w_2,w_4 \leq n/2$, as long as we multiply by a constant 4 at the end. We are now essentially left with 8 cases: the variables $2w_1, w_3, w_5$ are either $\leq h$ or $\geq n-h$. Next, we claim that the maximum over $w_1$ must occur when $w_1 \leq n/4$, so that in particular, this must occur when $w_1 \leq h/2$. This is because $\binom{n/2}{w_1}$ is symmetric around $2w_1=n/2$, so that it is equal across $w_1 \leq h/2$ and the other end of the range. Moreover, \Cref{lem:A_D_claims}.9 tells us that the same is true for $p_A'$: we have $p_A'(a,b) \geq p_A'(n-a,b)$ for $a \leq n/2$. In other words, for both factors depending on $w_1$, a maximizing value must occur when $2w_1 \leq n/2$, and thus when $w_1 \leq h/2$. Our expression then (slightly) simplifies to the following, where in particular we slightly extend the ranges of the sums:
            $$\leq O(1)\sum_{w_1}^{h/2}\sum_{w_2=\max\{w_1,w_3/2\}}^{n/2}\sum_{\substack{h \leq w_3\\n-h \leq w_3}}\sum_{w_4=\max\{w_3/2, w_5/2\}}^{n/2}\sum_{\substack{h \leq w_5\\n-h \leq w_5}} \binom{n/2}{w_1}p_A'(2w_1,w_2)p_D'(w_2,w_3)p_A'(w_3,w_4)p_D'(w_4,w_5) \ .$$
            To bound the above, we claim first that if $n-h \leq w_3$ the expression is $\negl(n)$, even with only one round (i.e., without $w_4$ and $w_5$). To see this, note that \Cref{lem:w2_w4_n-h} tells us that in case $n-h\leq w_3$, the maximizing value for $w_2$ is $n/2$. We end up with a factor $p_D'(n/2,w_3)$. From the known inequality $\binom{a}{b}\binom{c}{d}\leq\binom{a+c}{b+d}$, we have the inequality $p_D'(a,b)\leq \binom{n}{b}/\binom{n}{a}$. We therefore end up with a factor $\leq \binom{n}{h}/\binom{n}{n/2} \leq n^h/2^n$ in our expression. Suppose now that we bound all other $p_D'$ and $p_A'$ as $1$ and we bound $\binom{n/2}{w_1} \leq n^{w_1} \leq n^{h/2}$. We then have the final expression $n^hn^{h/2}/2^n=2^{-\Omega(n)}$, which easily beats the polynomial factor arising from the union bound.

            Thus, we can assume that if $n-h \leq w_3$ the expression is small enough, even in case there is only one round. We therefore imagine doing a separate induction argument with base case $m=1$ and $n-h \leq w_3$ to deal with this case, and note that the induction step below work for both this and the original $m=2$ base case. This lets us assume, for the main argument with base case $m=2$ that we only have the following two cases left: either $w_3,w_5 \leq h$ or $w_3\leq h$ while $n-h \leq w_5$. Now, in the latter case, we note that by the same reasoning we should fix $w_4=m/2$ to maximize, yielding an analogous $\negl(n)$ bound. Really, then, the only case left to deal with is $w_3,w_5 \leq h$.
            
            In this case \Cref{lem:w2_w4_h} tells us that the maximizing values of $w_2$ and $w_4$ will be the smallest possible values. Now, this case is a bit harder to deal with, as we cannot hope to obtain something negligible in $n$ here. Moreover, the sums over $w_2$ and $w_4$ each have a linear number of terms, so that any bound on the expression in the sums would have to be $1/\omega(n^2)$ to be decreasing in $n$. Our best bound was only inversely linear in $n$. Thus, we propose to split the ranges over $w_2$ and $w_4$ on $n^{1/4}$. For the initial part of the range, we will only suffer a factor $\sqrt{n}$ from the union bound, securing a final bound $1/\sqrt{n}$. For the latter part of the range, we can repeat the argument from just above to get something negligible in $n$
            
            Assume first that $n^{1/4} \leq w_2,w_4$. As noted, we should set these variables to their smallest value to maximize the expression, so we can in fact fix $w_2=w_4=n^{1/4}$. This means we end up with a factor $p_D'(w_2,w_3)\leq \binom{n}{w_3}/\binom{n}{w_2} \leq n^h \left(\frac{n^{1/4}}{n}\right)^{n^{1/4}}=n^{h-3/4n^{1/4}}$. Like before, let us bound all other $p_D'$ and $p_A'$ as 1, and bound $\binom{n/2}{w_1} \leq n^{w_1} \leq n^{h/2}$. What we are then left with is $n^{3/2h-3/4n^{1/4}}=n^{-\Omega(n^{1/4})}$, which is indeed negligible in $n$. If instead $w_2,w_4 \leq n^{1/4}$, we need to bound the following:
            $$O(1)\sum_{w_1=1}^{h/2}\sum_{w_2=\max\{w_1,w_3/2\}}^{n^{1/4}}\sum_{w_3=1}^{h}\sum_{w_4=\max\{w_3/2,w_5/2\}}^{n^{1/4}}\sum_{w_5=1}^{h} \binom{n/2}{w_1}p_A'(2w_1,w_2)p_D'(w_2,w_3)p_A'(w_3,w_4)p_D'(w_4,w_5) \ .$$
            We start by recalling that $w_4$ should be set to its minimum value, which is $w_4 \geq \max\{w_3/2,w_5/2\}$. After doing so, we bound $p_D(\max\{w_3/2,w_5/2\},w_5) \leq 1$. The only dependence on what was $w_4$ is now $p_A'(w_3,\max\{w_3/2,w_5/2\})$. We recall from \Cref{lem:A_D_claims}.8 that $p_A'(a,b)$ grows with $b$ until $b=n/2$, and noticing that $\max\{w_3/2,w_5/2\}) \leq h$, we bound $p_A'(w_3,\max\{w_3/2,w_5/2\}) \leq p_A'(w_3,h/2)$. With this, we've lost all dependence on $w_4$ and $w_5$, so we remove the sums over these variables, suffering a factor $n^{1/4}h$ in the bound:
            $$ \leq O(1)n^{1/4}h \sum_{w_1=1}^{h/2}\sum_{w_2=\max\{w_1,w_3/2\}}^{n^{1/4}}\sum_{w_3=1}^{h} \binom{n/2}{w_1}p_A'(2w_1,w_2)p_D'(w_2,w_3)p_A'(w_3,h/2) \ .$$
            We now apply \Cref{lem:A_D_claims}.10 to bound $p_A'$ and \Cref{lem:A_D_claims}.11 to bound $p_D'$:
            $$\leq O(1)n^{1/4}h\sum_{w_1=1}^{h/2}\sum_{w_2=\max\{w_1,w_3/2\}}^{n^{1/4}}\sum_{w_3=1}^{h}\binom{n/2}{w_1}\left(\frac{4w_2}{n}\right)^{w_1} \frac{1}{\binom{n}{w_2}} \left(\frac{2e\sqrt{w_2n}}{w_3}\right)^{w_3} \left(\frac{4h}{n}\right)^{w_3/2} \ . $$
            We next use \Cref{lem:A_D_claims}.3 to bound the leftover binomial coefficients, and simplify a bit to get:
            $$ \leq O(1)n^{1/4}h\sum_{w_1=1}^{h/2}\sum_{w_2=\max\{w_1,w_3/2\}}^{n^{1/4}}\sum_{w_3=1}^{h} \left(\frac{2ew_2}{w_1}\right)^{w_1} \left(\frac{w_2}{n}\right)^{w_2}\left(\frac{4e\sqrt{w_2h}}{w_3}\right)^{w_3}  \ .$$
            Since $w_1 \leq w_2$, we always have $2ew_2>w_1$ so that the above is growing with $w_1$. Quite similarly, we have $w_3/2 \leq w_2$ so that $4e\sqrt{w_2h} \geq 4e\sqrt{hw_3/2} \geq w_3$ as also $w_3 \leq h$. Setting both $w_1$ and $w_3$ to their largest value, we obtain:
            $$\leq O(1)n^{1/4}h^3 \sum_{w_2=1}^{n^{1/4}} \left(\frac{2ew_2}{w_2}\right)^{w_2} \left(\frac{w_2}{n}\right)^{w_2}\left(\frac{4e\sqrt{w_2h}}{2w_2}\right)^{2w_2} = O(1)n^{1/4}h^3 \sum_{w_2=1}^{n^{1/4}} \left(\frac{(2e)^3 w_2h}{n}\right)^{w_2} \ .$$
            This expression, in turn, is decreasing with $w_2$, as the $1/n$ factor easily causes the base of the exponent to drop below $1$, so that we should fix $w_2=1$, which gives us the final upper bound $O(1)\frac{h^{4}}{n^{1/2}} = \tilde O(n^{-1/2})$. 

            This finishes the base case. For the induction step, let us write $A_{w, \leq h}(R(AD)^m)$ to denote the expected number of vectors of weight $w$ in an $R(AD)^m$ code which had weight in range $[0, h] \cup [n-h, n]$ before and after each $AD$ operation. We then need to bound the following quantity:
            \begin{align*}
                \sum_{\substack{w_1 \leq h\\n-h \leq w_1}} A_{w_1, \leq h}(R(AD)^m) 
                &= \sum_{\substack{w_1 \leq h\\n-h \leq w_1}} A_{w_1, \leq h}(R(AD)^{m-1}) \sum_{w_2=\max\{w_1, w_3/2\}}^{n}\sum_{w_3\leq h, n-h \leq w_3} p_A(w_1,w_2)p_D(w_2,w_3) \ .
            \end{align*}
            As with the dual code, our goal is to bound $\sum_{w_2=\max\{w_1, w_3/2\}}^{n}\sum_{w_3\leq h, n-h \leq w_3} p_A(w_1,w_2)p_D(w_2,w_3)$ as $1/\poly(n)$. The remaining part of the expression $ \sum_{\substack{w_1 \leq h\\n-h \leq w_1}} A_{w_1, \leq h}(R(AD)^{m-1})$ is the induction hypothesis. Thus, if we can bound the rest of the expression as $1/\poly(n)$, we can apply the induction hypothesis, and are done. In principle, we have four cases again: $w_1 \leq h$ or $n-h \leq w_1$, and likewise for $w_3$. Now, we claim we don't need to worry about $n-h \leq w_1$. Note that we only consider the right part of the expression, i.e. starting with the sums over $w_2$ and $w_3$. The only dependence on $w_1$ in that expression is $p_A'(w_1,w_2)$, and we know from \Cref{lem:A_D_claims}.9 that the maximizing value of $w_1$ for that expression must occur when $w_1 \leq n/2$. Hence, in our case, the maximizing value must occur when $w_1 \leq h$. This just leaves the two cases $w_3 \leq h$ and $n-h \leq w_3$. Consider first the case $n-h \leq w_3$, so that again the maximizing value of $w_2$ becomes $w_2=n/2$ (after substituting $p_A\leq p_A'$ and likewise for $p_D$). As we argued above, this means $p_D(n/2,w_3)$ will be negligible in $n$ (we simply set $p_A(w_1,w_2)\leq 1$).
            
            All that remains is the case $w_3 \leq h$. We will again split this up based on whether $w_2 \leq n^{1/4}$. In case $w_2$ exceeds this value, we can re-use our argument from above to show that we only expect a negligible number of such vectors. This means we can assume that $w_2 \leq n^{1/4}$ to lessen the impact of the union bound. We are then left with the expression:
            $$O(1)\sum_{w_2=\max\{w_1, w_3/2\}}^{n^{1/4}}\sum_{w_3=1}^{h} p_A'(w_1,w_2)p_D'(w_2,w_3) \leq O(1) n^{1/4}  \sum_{w_3\leq h} p_A'(w_1,\max\{w_1/2,w_3/2\})p_D'(\max\{w_1/2,w_3/2\},w_3) \ ,$$
            where we use the now familiar claim that $w_2$ should be set to its smallest value. We again apply $p_D' \leq 1$ and use \Cref{lem:A_D_claims}.10 to bound $p_A'$ to get:
            $$\leq O(1)n^{1/4} \sum_{w_3=1}^{h}\left(\frac{4\max\{w_1/2,w_3/2\}}{n}\right)^{w_1/2} \leq O(1)n^{1/4} \sum_{w_3\leq h}\left(\frac{4\max\{w_1/2,w_3/2\}}{n}\right)^{w_1/2} \ .$$
            Now note that $2w_1, w_3 \leq h$ so $\max\{w_1/2,w_3/2\} =h/2$. The above is then $\leq O(1)n^{1/4}h \left(\frac{2h}n\right)^{w_1/2} \leq O(1)\frac{\sqrt{h}h}{n^{1/4}} = \tilde O(n^{-1/4})$, where we use the fact that it is decreasing with $w_1$ so that we set $w_1=1$.

            In conclusion, the base case gives us $\tilde{O}(n^{-1/2})$ while the induction step tells us that each extra round cuts off a factor $\tilde O({n^{-1/4}})$, so the final bound becomes $\tilde O(n^{-1/4 -(m-2)/4}) = \tilde O(n^{1/4-m/4})$.
        \end{proof}

        We next bound what we called case two in the introduction to this section: the expected number of $R(AD)^m$ codewords with weight in the middle of the weight range after $m-1$ rounds, to then drop the weight to close to the boundary of the weight range after the final round. 
        \begin{lemma}\label{lem:middle_weight_1_ad}
            Let $m\geq 2$. In a $R(AD)^m$ code of rate $1/2$, the expected number of vectors that have weight between $h$ and $n-h$ after $m-1$ rounds and weight $\leq h$ or $\geq n-h$ after $m$ rounds is $\negl(n)$. 
        \end{lemma}
        As with the dual code, this statement for $m \geq 2$ follows immediately from our bound on $(**)$. Unlike with the dual code, we are unable to prove this statement for $m=1$. This is not a problem, as we do not need this exact statement. Rather, we just want to say that we do not expect many of the vectors that are in the middle of the weight range before the first round (i.e. just after repeating) to drop to the boundaries of the weight range by round $m-1$. Since we will only use codes with $m \geq 3$, it therefore also suffices to argue that we don't expect many of the vectors with weight in the middle of the weight range before the first round to end up with weight near the boundaries after the second round (instead of after the first round). This is something that turns out to be true, and we prove it in the following lemma.        
    
        \begin{lemma}\label{lem:middle_weight_2_ad}
            In a $R(AD)^2$ code of rate $1/2$, the expected number of vectors that have weight between $h$ and $n-h$ before the first round (just after repeating) and weight $\leq h$ or $\geq n-h$ after both the first and second round is $\negl(n)$. 
        \end{lemma}
        \begin{proof}
            The quantity we need to bound is the following (where we've already applied the inequality $p_A \leq O(1) p_A'$ and similarly for $p_D$):
            $$O(1)\sum_{w_1=h/2}^{(n-h)/2}\sum_{w_2=\max\{w_1, w_3/2\}}^{n-\max\{w_1,w_3/2\}\}}\sum_{w_3=1}^{h}\sum_{w_4=\max\{w_3/2, w_5/2\}}^{n-\max\{w_3/2,w_5/2\}\}}\sum_{w_5=1}^{h} \binom{n/2}{w_1} p_A'(2w_1,w_2) p_D'(w_2,w_3) p_A'(w_3,w_4) p_D'(w_4,w_5) ,$$
            and similarly for $n-h \leq w_3$ and $n-h \leq w_5$, giving us a total of 4 different expressions we need to consider. We claim first that that the maximizing value for $w_1$ must occur when $w_1\leq n/4$. This follows from the exact same reasoning as in the proof of \Cref{lem:low_weight_ad}: $\binom{n/2}{w_1}$ is symmetric around $2w_1=n/2$ and from \Cref{lem:A_D_claims}.9 we know the maximizing value of $p_A'(a,b)$ must occur when $a \leq n/2$. We can thus assume that the maximum over $w_1$ must occur when $2w_1 \leq n/2$ so that $w_1 \leq n/4$.
        
            Now, our main tools to deal with these expressions will again be \Cref{lem:w2_w4_h} and \Cref{lem:w2_w4_n-h}: they tell us that when we have $p_A'(a,b) p_D'(b,c)$ such that either $a$ or $c$ is $\geq n-h$, then the maximizing value of $b$ is $b=n/2$; if instead $a,c \leq h$ then the maximizing value of $b$ is its minimal value: $\max\{a/2, c/2\}$. We can apply this lemma to $w_4$ directly. Unfortunately, $w_1$ doesn't quite fit the lemma, as it has range $h/2 \leq w_1 \leq n/4$. Therefore, let us split the analysis into two cases: first $w_1 \leq \sqrt{n}$ and second $\sqrt{n}<w_1$. 

            In the latter case, our problem now has a much larger input weight. It turns out that this simplifies things to the point that we \textit{can} prove the expression is small even after the first round. That is, we first bound the above as 
            $$\leq O(1) hn \sum_{w_1=\sqrt{n}}^{n/4}\sum_{w_2=\max\{w_1, w_3/2\}}^{n-\max\{w_1,w_3/2\}\}}\sum_{w_3=1}^{h}\binom{n/2}{w_1} p_A'(2w_1,w_2) p_D'(w_2,w_3) \ ,$$
            and similarly for $n-h\leq w_3$. We recall the simple inequality $p_D'(w_2,w_3)\leq\binom{w_3}{n}/\binom{w_2}{n}$ which we derived in the proof of \Cref{lem:low_weight_ad}. Applying this, we can rewrite the expression in the sums as 
            $$\binom{n/2}{w_1}p_A'(2w_1,w_2) p_D'(w_2,w_3) \leq \binom{n/2}{w_1}\frac{\binom{n-w_2}{w_1}\binom{w_2}{w_1}}{\binom{n}{2w_1}} \frac{\binom{w_3}{n}}{\binom{w_2}{n}}\ ,$$
            and we claim this expression decreases with $w_2$. To see this, note that $\binom{n-w_2}{w_1}$ is always increasing with $w_1$, as $w_1 \leq n/4$ and we can assume $w_2 \leq n/2$, so that always $w_1 \leq (n-w_2)/2$ (since both $p_A'$ and $p_D'$ are symmetric around $w_2=n/2$). Since $w_1$ is always on the left side of the top value of this binomial $\binom{n-w_2}{w_1}$, it follows that this binomial is decreasing with $w_2$. This means the only factor increasing with $w_2$ is $\binom{w_2}{w_1}$, but this is clearly dominated by the factor $1/\binom{w_2}{n}$. We can therefore conclude that we must set $w_2$ to its smallest value in order to maximize the entire expression, which means $w_2=\max\{w_1,w_3/2\}=w_1$, as $w_3 \leq h <\sqrt{n} \leq w_1$. This lets us bound the expression as follows
            $$\binom{n/2}{w_1}p_A'(2w_1,w_2) p_D'(w_2,w_3) \leq \binom{n/2}{w_1}\frac{\binom{n-w_1}{w_1}}{\binom{n}{2w_1}} \frac{\binom{n}{w_3}}{\binom{n}{w_1}} = \binom{n/2}{w_1}\frac{\binom{n-w_1}{w_1}\binom{n}{w_3}}{\binom{n}{w_1}\binom{n}{2w_1}} \ .$$
            This expression in turn is decreasing with $w_1$, as $\binom{n/2}{w_1} \leq \binom{n}{2w_1}$ and $\binom{n-w_1}{w_1} \leq \binom{n}{w_1}$. We should therefore set $w_1=\sqrt{n}$. It is also clear that the expression grows with $w_3$ until $h$, so that we should fix $w_3=h$. This gives us the final bound   
            $$\binom{n/2}{w_1}p_A'(2w_1,w_2) p_D'(w_2,w_3)
            \leq \binom{n/2}{\sqrt{n}}\frac{\binom{n-\sqrt{n}}{\sqrt{n}}\binom{n}{h}}{\binom{n}{2\sqrt{n}}\binom{n}{\sqrt{n}}} 
            \leq \left(\frac{en/2}{\sqrt{n}}\right)^{\sqrt{n}} \left(\frac{2\sqrt{n}}{n}\right)^{2\sqrt{n}} \left(\frac{n}{h}\right)^{h} 
            = \left(\frac{2e}{\sqrt{n}}\right)^{\sqrt{n}} \left(\frac{n}{h}\right)^{h} =n^{-\Omega(\sqrt{n})} \ .$$
            Note that we get the exact same bound when $n-h \leq w_3$, as in the final expression the only factor depending on $w_3$ was $\binom{n}{w_3}$, which is symmetric around $w_3=n/2$. This final bound is more than enough to overcome the polynomial factors arising from the union bound.

            Thus, when $\sqrt{n}\leq w_1$ we have shown that the expression we wish to bound is negligible in $n$. It remains to deal with $w_1 \leq \sqrt{n}$. Again, we have the four cases $w_3,w_5\leq h$ and opposite. Now, we claim first that when we have $n-h \leq w_3$, we can bound the expression as $\negl(n)$, regardless of the value of $w_5$. This thus deals with three of the cases above, after which we will only be left with $w_3, w_5 \leq h$. So, suppose that $n-h \leq w_3$. As just above, we simply throw away the last two rounds, and after that notice that \Cref{lem:w2_w4_n-h} then tells us that we should set $w_2=n/2$ to maximize the expression:
            $$\leq O(1) hn^2 \sum_{w_1=\sqrt{n}}^{n/4}\sum_{w_3=1}^{h}\binom{n/2}{w_1} p_A'(2w_1,n/2) p_D'(n/2,w_3) = O(1)\sum_{w_1=\sqrt{n}}^{n/4}\sum_{w_3=1}^{h} \frac{\binom{n/2}{w_1}^3}{\binom{n}{2w_1}} \frac{\binom{n}{w_3}}{\binom{n}{n/2}}\ .$$
            We can bound the expression in the sums as 
            $$\leq \binom{n}{h}\frac{\binom{n/2}{n/4}^3}{\binom{n}{n/2}^2} 
            = \binom{n}{h} \frac{(n/2)!^3 (n/2)^2(n-n/2)^2}{(n/4)!^3(n/2-n/4)!^3 n!^2}
            = \binom{n}{h} \frac{(n/2)!^7}{(n/4)!^6 n!^2}\ . $$
            Applying Stirling's approximation $e(n/e)^n\leq n! \leq en(n/e)^n$ gives us 
            $$ \leq n^h \frac{\left(\frac{en}{2}\left(\frac{n}{2e}\right)^{n/2}\right)^7}{\left(e\left(\frac{n}{4e}\right)^{n/4}\right)^6 \left(e\left(\frac{n}{e}\right)^{n}\right)^2} 
            = n^h \left(\frac{en}{2}\right)^7\frac{1}{e^6e^2} n^{n(7/2-6/4-2)} e^{n(7/2-6/4-2)} 2^{7n/2} 4^{6n/4} 
            =n^h \left(\frac{en}{2}\right)^7\frac{1}{e^6e^2} 2^{-n/2} \ .$$
            Clearly, this is dominated by $2^{-n}$, which beats even the super-polynomial factor $n^h$, so that this is $2^{-\Omega(n)}=\negl(n)$. This just leaves the final case: $w_3,w_5 \leq h$. We use the now familiar fact that $w_2$ and $w_4$ should be set to their smallest values. This lets us write:
            \begin{align*}
                &O(1)\sum_{w_1=h/2}^{\sqrt{n}}\sum_{w_2=\max\{w_1, w_3/2\}}^{n-\max\{w_1,w_3/2\}\}}\sum_{w_3=1}^{h}\sum_{w_4=\max\{w_3/2, w_5/2\}}^{n-\max\{w_3/2,w_5/2\}\}}\sum_{w_5=1}^{h} \binom{n/2}{w_1} p_A'(2w_1,w_2) p_D'(w_2,w_3) p_A'(w_3,w_4) p_D'(w_4,w_5) \\ 
                &\leq O(1)n^2 \sum_{w_1=h/2}^{\sqrt{n}}\sum_{w_3=1}^{h}\sum_{w_5=1}^{h} \binom{n/2}{w_1} p_A'(2w_1,\max\{w_1, w_3/2\}) p_D'(\max\{w_1, w_3/2\},w_3) \\ & \ \ \ \ \ \ \ \ \ \ \ \ \ \ \ \ \ \ \ \ \ \ \ \ \ \ \ \ \ \ \ \ \ \ \ \ \ \ \cdot p_A'(w_3,\max\{w_3/2, w_5/2\}) p_D'(\max\{w_3/2, w_5/2\},w_5) \\
                &= O(1)n^2 \sum_{w_1=h/2}^{\sqrt{n}}\sum_{w_3=1}^{h}\sum_{w_5=1}^{h} \binom{n/2}{w_1} p_A'(2w_1,w_1) p_D'(w_1,w_3) p_A'(w_3,\max\{w_3/2, w_5/2\}) p_D'(\max\{w_3/2, w_5/2\},w_5) \ .
            \end{align*}
            In the last step we used $w_2=\max\{w_1,w_3/2\}=w_1$, which is because $w_3/2 \leq h/2 \leq w_1$. To continue, we want to find the maximizing value of $w_1$, for which we look at only the factors involving $w_1$, where we bound $p_A'$ using \Cref{lem:A_D_claims}.10 bound $p_D'$ using \Cref{lem:A_D_claims}.11:
            \begin{align*}
                \binom{n/2}{w_1} p_A'(2w_1,w_1) p_D'(w_1,w_3) 
                \leq \left(\frac{2ew_1}{w_1}\right)^{w_1} \left(\frac{2e\sqrt{w_1n}}{w_3}\right)^{w_3} \left(\frac{w_1}{n}\right)^{w_1}
                = \left(\frac{2ew_1}{n}\right)^{w_1} \left(\frac{2e\sqrt{w_1n}}{w_3}\right)^{w_3} \ .
            \end{align*}
            We now claim that the above is decreasing with $w_1$. The factor $\sqrt{w_1}^{w_3}$ clearly grows with $w_1$, and its growth is biggest when $w_3=h$. But even then, that factor is $\sqrt{w_1}^{w_3} \leq n^{h/2}$, which will be overwhelmed by $(w_1/n)^{w_1} \geq n^{-\sqrt{n}/2}$. Note that this would be true for any $w_1 < n/(2e)$ (as this guarantees something of the form $2^{-\Omega(w_1)}$ which will still kill off $w_1^{h/2}$). It follows that we should make $w_1$ as small as possible, so that above becomes 
            $$ \leq\left(\frac{2ew_1}{n}\right)^{w_1} \left(\frac{2e\sqrt{w_1n}}{w_3}\right)^{w_3} \leq  \left(\frac{eh}{n}\right)^{h/2} \left(\frac{e\sqrt{2hn}}{w_3}\right)^{w_3} \ .$$
            Substituting the above back into our expression, we are left with:
            $$\leq O(1)n^{5/2} \sum_{w_1=h/2}^{\sqrt{n}}\sum_{w_3=1}^{h}\sum_{w_5=1}^{h} \left(\frac{eh}{n}\right)^{h/2} \left(\frac{e\sqrt{2hn}}{w_3}\right)^{w_3} p_A'(w_3,\max\{w_3/2, w_5/2\}) p_D'(\max\{w_3/2, w_5/2\},w_5) \ .$$
            We now throw away $p_D'$ by bounding it as $1$, and then bound $p_A'$ using \Cref{lem:A_D_claims}.10:
            $$ \leq O(1)n^{5/2} \sum_{w_1=h/2}^{\sqrt{n}}\sum_{w_3=1}^{h}\sum_{w_5=1}^{h} \left(\frac{eh}{n}\right)^{h/2}\left(\frac{e\sqrt{2hn}}{w_3}\right)^{w_3} \ . $$
            Simplifying a bit, the factors depending on $w_3$ end up looking like $\left(\frac{2e\sqrt{2h\max\{w_3/2, w_5/2\}}}{w_3} \right)^{w_3}$, and it is not hard to see that this fraction is always $>1$ so that this expression is increasing with $w_3$. We should thus set $w_3=h$ which gives us
            $$ \leq O(1)n^{5/2} \sum_{w_1=h/2}^{\sqrt{n}}\sum_{w_3=1}^{h}\sum_{w_5=1}^{h} \left(\frac{eh}{n}\right)^{h/2} \left(\frac{2e\sqrt{2hh/2}}{h} \right)^{h} \leq O(1)n^3h^2\left(\frac{2e^2\sqrt{eh}}{\sqrt{n}}\right)^{h} = n^{-\Omega(h)} = \negl(n) \ .$$
        \end{proof}

\section{Cryptographic application} \label{sec:crypto-application}

As mentioned in the introduction, such fast codes with fast good duals can be used to construct efficient secure multiparty computation (MPC) protocols for the problem of \emph{encryption matrix-vector product} (EMVP). In this section, we sketch this application, but refrain from providing precise definitions, preferring to refer to~\cite{benhamouda2025encrypted} where appropriate. 

Briefly, consider a client that wishes to delegate the task of computing matrix-vector products for a matrix $M$ to a server in the cloud. In an initial setup phase, the client takes the matrix $M$ and encrypts it into a matrix $\hat M$ (using a secret key), which is then sent to the server. Later, in the online phase, the client takes query vectors $q$ and sends encryptions $\hat q$ thereof to the server, keeping an associated secret key $q'$. The server then returns a value $M'$. The correctness requirement is that from $M'$, $q'$ and the initial secret key, the client learns the value of the matrix vector-product $Mq$. The security requirement is that the server learns nothing about the matrix $M$ or any of the query vectors. 

Benhamouda~et~al~\cite{benhamouda2025encrypted} suggest resolving this problem in the following way.\footnote{This, admittedly, simplifies certain details.} The first step is for the server to sample two secret, dual codes: one could use an $R(DA)^m$ code $\cC$ and its dual $R(AD)^m$ code $\cC^\perp$. The code $\cC^*$ will be used to encrypt the matrix $M$: one sets $\hat M := MG^*+R$, where $G^*$ is a generator matrix for $\cC^*$ and $R$ is a pseudorandom matrix (which serves to hide $M$). (One technical caveat is that the matrix $G^*$ must be systematic for this application; while we would like to be able to guarantee that this step can be done in linear time since $\cC^*$ admits a linear time encoding algorithm, we unfortunately cannot.)

Next, in the online phase, to mask the query $q$ one must use a uniformly random codeword $c \in \cC$. To sample such a codeword, one can sample a uniformly random $r \in \F_2^k$ and then output $rG$, where $G$ generates $\cC$. In particular, assuming $\cC$ admits linear time encoding this step can be done in linear time! 

Thus, we see that by finding self-dual codes with fast encoding we can speed up this algorithm, making our codes quite promising for this application. We remark that Benhamouda et al~\cite{benhamouda2025encrypted} already suggested using \emph{quasi-cyclic} codes to get the desired dual codes, which admit quasilinear encoding (but not truly linear, as we obtain). In fact, the authors write that ``one could potentially use other families of (dual) linear codes that admit fast encoding''~\cite[Section~3.1]{benhamouda2025encrypted}.

Of course, this entire discussion is moot if the resulting scheme is insecure. Naturally, the code we sample needs to satisfy a certain hardness assumption: namely, a variant of the \emph{learning subspace with noise (LSN)} assumption~\cite[Definition~3.1]{benhamouda2025encrypted} must hold, which informally states that if one can obtain random codewords from $\cC$ which, with probability $\mu \in (0,1)$, are rerandomized (i.e., one obtains a uniformly random vector a $\mu$-fraction of the time), it is hard to recover the secret code $\cC$. One can ask whether a secret $R(AD)^m$ (or $R(DA)^m$) code can be reasonably expected to satisfy this LSN assumption. For the close cousin \emph{learning parity with noise (LPN)} assumption -- where essentially the task is to decode -- the general consensus is that unless a code displays notable ``algebraic structure,'' the best general strategy is to look for low-weight dual codewords to hopefully recover information about the codeword untainted by the noise (see, e.g., \cite[Section~3]{couteau2021silver}). But, note that we have explicitly argued that $\cC$ \emph{does not} have any light dual codewords: $\cC^\perp$ (quite likely) approaches the GV bound!

Of course, this is merely a sketch of this potential application. Nonetheless, we believe it is emblematic of the principle that such codes with fast encoding and non-trivial dual properties can be useful for cryptography, especially when targetting extremely efficient implementations. We anticipate there to be many further applications in the future. 

\bibliography{refs}

\appendix

\section{Deferred proofs}\label{sec:deferred-proofs}

    In this appendix, we provide proofs that were deferred from the main body of the text.  
    \fafdclaims*
    \begin{proof}
        \begin{enumerate}
            \item Recall that $f_A(\alpha,\beta) = \alpha - h(\beta) + (1-\alpha)h\left(\frac{\beta-\alpha/2}{1-\alpha}\right)$, and that $h(p)$ is only defined for $p\in[0,1]$. The last term thus tells us that $\frac{\beta-\alpha/2}{1-\alpha} \in [0,1]$. In particular, we must then have  $0 \leq \frac{\beta-\alpha/2}{1-\alpha} \iff 0 \leq \beta - \alpha/2 \iff \alpha \leq 2\beta$. Similarly, we must have $\frac{\beta-\alpha/2}{1-\alpha} \leq 1 \iff \beta-\alpha/2 \leq 1-\alpha \iff \alpha \leq 2(1-\beta)$, as required. 
            \item Recall that $f_{D}(\alpha,\beta) = (1-\alpha )h\left(\frac{ \beta}{2(1-\alpha)}\right)+ \alpha h\left(\frac{ \beta}{2\alpha }\right) - h(\alpha)$. The first term then tells us that $\frac{ \beta}{2(1-\alpha)} \leq 1 \iff \beta \leq 2(1-\alpha)$. Similarly, the second term tells us that $\frac{ \beta}{2\alpha } \leq 1 \iff \beta \leq 2\alpha$.
            \item Note that $\frac{\partial f_A}{\partial \alpha}(\alpha,\beta) = \log \left(\frac{\sqrt{2\beta-\alpha}\sqrt{2(1-\beta)-\alpha}}{1-\alpha}\right)$. This partial derivative is non-positive if and only if $(2\beta-\alpha)(2(1-\beta)-\alpha) \leq (1-\alpha)^2$. We claim this inequality holds for all $\beta \in [0,1]$, with equality if and only if $\beta=1/2$. Indeed, by expanding the expressions and cancelling like terms we find that the above is equivalent to $4\beta(1-\beta) \leq 1$, which can be seen to hold for all $\beta \in [0,1]$ with equality if and only if $\beta=1/2$ by expanding the (obvious) inequality $(1-2\beta)^2 \geq 0$.
            \item See Claim~9.9.1 in \cite{blaze}.
            \item Note that we can rewrite $f_A(\alpha,\beta) = (1-\beta )h\left(\frac{ \alpha}{2(1-\beta)}\right)+ \beta h\left(\frac{ \alpha}{2\beta }\right) - h(\alpha)$; we take the first two terms of $f_D(\alpha,\beta)$ and swap the two variables, which we can do as $A$ and $D$ are inverse operations. The claim we need to prove is then equivalent to showing that $(1-\beta )h\left(\frac{ \alpha}{2(1-\beta)}\right)+ \beta h\left(\frac{ \alpha}{2\beta }\right) \leq h(\beta)$. It is not hard to see that the left hand side is maximized when $\alpha = 2\beta(1-\beta)$, so that we can upper bound the left hand side as $(1-\beta )h\left(\frac{2\beta(1-\beta)}{2(1-\beta)}\right)+ \beta h\left(\frac{2\beta(1-\beta)}{2\beta }\right) = (1-\beta)h(\beta)+\beta h(1-\beta)=h(\beta)$.

            To establish that $\alpha=2\beta(1-\beta)$ is the maximizer, we differentiate with respect to $\alpha$, obtaining
            \[
                \frac12 \log\left(\frac{(2(1-\beta)-\alpha)(2\beta-\alpha)}{\alpha^2}\right) \ ,
            \]
            which equals $0$ iff $\alpha^2 = 2\beta(1-\beta)$. That this is a maximum follows by taking the second derivative, which is 
            \[
                \frac{1}{\alpha\ln2}\frac{\alpha-4(1-\beta)\beta}{(2\beta-\alpha)(2(1-\beta)-\alpha)} \ ,
            \]
            which is negative on the relevant domain. Indeed, the proof of \Cref{lem:max-beta=2gamma(1-gamma)} establishes this (upon replacing $\alpha$ by $\beta$ and $\beta$ by $\gamma$). 
            
            \item This is almost equivalent to the previous claim. Here, we need to show that $(1-\beta )h\left(\frac{ 1-\alpha}{2(1-\beta)}\right)+ \beta h\left(\frac{ 1-\alpha}{2\beta }\right) \leq h(\beta)$. Here, an analogous argument shows the maximizing value of $\alpha$ is $1-2\beta(1-\beta)$, so that the $1-$ cancel out, and we are left with the same upper bound of $h(\beta)$. 
            \item This is equivalent to showing that $(1-\alpha)h\left(\frac{ \beta}{2(1-\alpha)}\right)+ \alpha h\left(\frac{ \beta}{2\alpha }\right) \leq h(\beta)$. Note that this is an equality in case $\alpha=1/2$. We now claim that moving $\alpha$ away from $1/2$ decreases the left hand side, so that the claimed inequality is always true. To see this, note that the derivative of the left hand side to $\alpha$ is $\log \left(1-\frac{\beta}{2(1-\alpha)}\right)-\log \left(1-\frac{\beta}{2 \alpha}\right)$ which is positive if and only if $1-\frac{\beta}{2(1-\alpha)} > 1-\frac{\beta}{2 \alpha}$ which is the same as $\frac{\beta}{2(1-\alpha)}<\frac{\beta}{2 \alpha}$ which in turn can be rewritten to $2\alpha<2(1-\alpha)$ which is equivalent to $\alpha < 1/2$. In other words, it is indeed true that the expression is maximized at $\alpha=1/2$.
        \end{enumerate}
    \end{proof}

    \hfa*
    \begin{proof}
        Recall that $f_A(\alpha,\beta) = \alpha - h(\beta) + (1-\alpha)h\left(\frac{\beta-\alpha/2}{1-\alpha}\right)$. Let us write $\beta:=2\alpha(1-\alpha)$ to simplify notation. We can then write out the left hand side as:
        $$ h(\beta) + f_A(2\beta,\alpha) =  h(\beta) + \beta -h(\alpha) + (1-\beta)h\left(\frac{\alpha-\beta/2}{1-\beta}\right) \ .$$
        We simplify the last term as follows:
        \begin{align*}
            &(1-\beta)h\left(\frac{\alpha-\beta/2}{1-\beta}\right) \\
            &= (1-\beta)h\left(\frac{\alpha^2}{1-\beta}\right) \\
            &= (1-\beta) \left(-\frac{\alpha^2}{1-\beta}\log\left(\frac{\alpha^2}{1-\beta}\right)  -\left(1-\frac{\alpha^2}{1-\beta}\right)\log\left(1-\frac{\alpha^2}{1-\beta}\right) \right)    \\
            &= (1-\beta) \left(-\frac{\alpha^2}{1-\beta}\log\left(\frac{\alpha^2}{1-\beta}\right)  - \frac{1-\beta-\alpha^2}{1-\beta}\log\left(\frac{1-\beta-\alpha^2}{1-\beta}\right) \right)    \\
            &= -\alpha^2 \log\left(\frac{\alpha^2}{1-\beta}\right) -(1 - \beta - \alpha^2)\log\left(\frac{1-\beta-\alpha^2}{1-\beta}\right) \\
            &= (\alpha^2 + 1-\beta-\alpha^2)\log(1-\beta) - \alpha^2 \log(\alpha^2) -(1-\beta-\alpha^2)\log(1-\beta-\alpha^2) \\
            &= (1-\beta)\log(1-\beta) - \alpha^2 \log(\alpha^2) -(1-\alpha)^2\log((1-\alpha)^2) \\
            &= (1-\beta)\log(1-\beta) - 2\alpha^2 \log(\alpha) -2(1-\alpha)^2\log(1-\alpha) \\
        \end{align*}
        Plugging this back in gives us:
        \begin{align*}
            &h(\beta) + f_A(2\beta,\alpha) \\
            &= -h(\alpha) + h(\beta) + \beta + (1-\beta)h\left(\frac{\alpha-\beta/2}{1-\beta}\right) \\
            &= -h(\alpha) -\beta\log(\beta) - (1-\beta) \log (1-\beta) + \beta + (1-\beta)\log(1-\beta) - 2\alpha^2 \log(\alpha) -2(1-\alpha)^2\log(1-\alpha) \\
            &= -h(\alpha) -\beta\log(\beta) + \beta - 2\alpha^2 \log(\alpha) -2(1-\alpha)^2\log(1-\alpha) \\
            &= -h(\alpha) -(2\alpha-2\alpha^2)\log(\beta) + 2\alpha - \alpha^2 - 2\alpha^2 \log(\alpha) -2(1-\alpha)^2\log(1-\alpha) \\
            &= -h(\alpha) + 2\alpha(1-\log(\beta)) + 2\alpha^2(\log(\beta)-1-\log(\alpha)) - 2(1-\alpha)^2\log(1-\alpha) \\
            &= -h(\alpha) + 2\alpha(\log(2)-\log(\beta)) + 2\alpha^2(\log(\beta)-\log(2)-\log(\alpha)) - 2(1-\alpha)^2\log(1-\alpha) \\
            &= -h(\alpha) - 2\alpha \log(\alpha(1-\alpha)) + 2\alpha^2\log(1-\alpha) - 2(1-\alpha)^2\log(1-\alpha) \\
            &= -h(\alpha) - 2\alpha\log(\alpha) + (2\alpha^2-2\alpha)\log(1-\alpha) - 2(1-\alpha)^2\log(1-\alpha)  \\
            &= -h(\alpha) - 2\alpha\log(\alpha) - 2\alpha(1-\alpha)\log(1-\alpha) - 2(1-\alpha)^2\log(1-\alpha)  \\
            &= -h(\alpha) - 2\alpha\log(\alpha) - (1-\alpha)\log(1-\alpha)  - 2(1-\alpha)^2\log(1-\alpha) +2(1-\alpha)\log(1-\alpha)(1-\alpha) \\
            &= -h(\alpha) + 2h(\alpha) = h(\alpha) \ , 
        \end{align*}
        where we note that that one-to-last equality follows because $-2\alpha(1-\alpha)\log(1-\alpha)$ can be rewritten $-2\alpha(1-\alpha)\log(1-\alpha) = -2(1-\alpha)\log(1-\alpha) + z$ where $z=-2\alpha(1-\alpha)\log(1-\alpha) + 2(1-\alpha)\log(1-\alpha) = 2(1-\alpha) \log(1-\alpha)(1-\alpha)$.
    \end{proof}

    \adclaims*
    \begin{proof}
        \begin{enumerate}
            \item Immediate from the definitions of $p_A$ and $p_A'$.
            \item Immediate from the definitions of $p_D$ and $p_D'$.
            \item Well-established inequality; this follows from the following strong version of Stirling's inequality: $\sqrt{2\pi k}\left(\frac{k}{e}\right)^k e^{1/(12k + 1)} \leq k! \leq \sqrt{2\pi k}\left(\frac{k}{e}\right)^k e^{1/(12k)}$.
            \item Immediate from the definition of $p_D'$.
            \item Immediate from the definition of $p_A'$.
            \item  We need to show that $\binom{n-a}{\lfloor b/2 \rfloor} \binom{a-1}{\lceil b/2 \rceil-1} \leq \binom{n-a}{b/2}\binom{a}{b/2}$. Recall that we let the binomial coefficients have non-integer values, specifically values halfway between two integers. Recall that this is defined as follows: $(n-1/2)! = \Gamma(n+1/2) = \frac{(2n)! \sqrt{\pi}}{4^nn!}$. Now, the inequality we wish to prove is clearly true for even $b$, as the floor and ceiling just disappear. For odd $b$, the statement is unfortunately not true. For instance, say that $n=12$, $a=7$. Then the statement is true for $b\leq 8$, but not for $b=9$:
            $$\binom{12-7}{\lfloor 9/2 \rfloor} \binom{7-1}{\lceil 9/2 \rceil-1} \leq \binom{12-7}{9/2}\binom{7}{9/2} \iff \binom{5}{4}\binom{6}{4} \leq \binom{5}{4.5}\binom{7}{4.5} \iff 75 \leq \frac{2097152}{2835} \leq 74.951 \ .$$
            The best we can hope for is that the upper bound is true up to some constant factor. Let's try to prove that. We will need the following fact, which we will prove below.
    
            \begin{claim}\label{lem:factorial_half}
                For all positive natural numbers $n$, we have:
                \begin{enumerate}
                    \item $\dfrac{(n+1/2)!}{n!} \leq \dfrac{4}{e} \sqrt{n} \leq 1.5\sqrt{n}$,
                    \item $\dfrac{(n-1/2)!}{n!} \leq \dfrac{1}{\sqrt{n}}$.
                \end{enumerate}
            \end{claim}
            \begin{proof}
                Recall that that $(n-1/2)! = \Gamma(n+1/2) = \frac{(2n)! \sqrt{\pi}}{4^nn!}$. This also means that $(n+1/2)! = \Gamma(n+1+1/2) = \frac{(2(n+1)! \sqrt{\pi}}{4^{n+1}(n+1)!}$. To be able to upper bound the ratio's, we will write out the factorials using Stirling's approximation as stated in \Cref{lem:A_D_claims}.3:
                 \begin{align*}
                    \frac{(n+1/2)!}{n!} 
                    &= \frac{\Gamma(n+1+1/2)}{n!} 
                    = \frac{\frac{(2(n+1))!\sqrt{\pi}}{4^{n+1}(n+1)!}}{n!} 
                    = \frac{(2(n+1))!\sqrt{\pi}}{4^{n+1}(n+1)!n!} \\
                    &\leq \frac{\sqrt{2\pi(2(n+1))}\left(\frac{2(n+1)}{e}\right)^{2(n+1)} e^{1/(12(2(n+1)))} \sqrt{\pi}}{\sqrt{2\pi n}\left(\frac{n}{e}\right)^n e^{1/(12n+1)} \sqrt{2\pi (n+1)}\left(\frac{n+1}{e}\right)^{n+1} e^{1/(12(n+1)+1)} 4^{n+1}} \\
                    &= \frac{\sqrt{2\pi(2(n+1))}\sqrt{\pi}}{\sqrt{2\pi n}\sqrt{2\pi (n+1)}} \frac{e^{1/(12(2(n+1)))}}{ e^{1/(12n+1)} e^{1/(12(n+1)+1)}} \frac{\left(\frac{2(n+1)}{e}\right)^{2(n+1)} }{\left(\frac{n}{e}\right)^n \left(\frac{n+1}{e}\right)^{n+1}  4^{n+1}} \\
                    &\leq \frac{\sqrt{2\pi(2(n+1))}\sqrt{\pi}}{\sqrt{2\pi n}\sqrt{2\pi (n+1)}}  \frac{\left(\frac{2(n+1)}{e}\right)^{2(n+1)} }{\left(\frac{n}{e}\right)^n \left(\frac{n+1}{e}\right)^{n+1}  4^{n+1}} \\
                    &= \frac1{\sqrt{n}} \left(\frac{4(n+1)^2}{e^2} \frac{e}{n+1} \frac{1}{4} \frac{e}{n}\right)^{n+1} \frac{n}{e} \\
                    &= \frac{\sqrt{n}}{e} \left(\frac{n+1}{n}\right)^{n+1}  \\
                    &\leq \frac{4}{e} \sqrt{n} \ .
                \end{align*}
                Note that the first inequality follows from applying Stirling's approximation to all the factorials. The second inequality follows because the $\exp$ factor has a negative exponent for all $n$, so that it is $\leq 1$ for all $n$. The last inequality follows since $\left(\frac{n+1}{n}\right)^{n+1}$ converges to $e$ from above with $n$, so that we can use the value at $n=1$ to get an inequality that holds for all positive natural numbers $n$. Note that this means that the above converges to $\sqrt{n}$ in the limit.
            
                We can deal with the second ratio analogously:
                 \begin{align*}
                    \frac{(n-1/2)!}{n!} 
                    &= \frac{\Gamma(n+1/2)}{n!} 
                    = \frac{(2n)!\sqrt{\pi}}{n!n!4^n} \\ 
                    &\leq \frac{\sqrt{2\pi (2n)}\left(\frac{2n}{e}\right)^{2n} e^{1/(12(2n))} \sqrt{\pi} }{\left(\sqrt{2\pi n}\left(\frac{n}{e}\right)^n e^{1/(12n+1)}\right)^2 4^n} \\
                    &= \frac{\sqrt{2\pi (2n)}\sqrt{\pi}}{\sqrt{2\pi n}^2} \frac{e^{1/(24n)}}{e^{2/(12n+1)}} \left(\frac{2n}{e} \frac{e}{n} \frac{1}{2}\right)^{2n} \\ 
                    &\leq \frac{\sqrt{2\pi (2n)}\sqrt{\pi}}{\sqrt{2\pi n}^2} \left(\frac{2n}{e} \frac{e}{n} \frac{1}{2}\right)^{2n} \\ 
                    &= \frac{1}{\sqrt{n}} \\
                \end{align*}
                Note again that the first inequality follows from applying Stirling's approximation to all the factorials. The second inequality again follows because the $\exp$ factor has a negative exponent for all $n$, so that it is $\leq 1$ for all $n$. Note that this means that the above again converges to $1/\sqrt{n}$ in the limit.
            \end{proof}
            With this lemma in hand, we can try to do what we set out to. 
            \begin{align*}
                &\binom{n-a}{\lfloor b/2 \rfloor} \binom{a-1}{\lceil b/2 \rceil-1} \leq K \binom{n-a}{b/2}\binom{a}{b/2} \\
                &\iff \binom{n-a}{\frac{b-1}2} \binom{a-1}{\frac{b-1}2} \leq K \binom{n-a}{b/2}\binom{a}{b/2} \\
                &\iff \frac{(n-a)!}{(\frac{b-1}2)!(n-a-\frac{b-1}2)!} \frac{(a-1)!}{(\frac{b-1}2)!(a-1-\frac{b-1}2)!} \frac{(b/2)!(n-a-b/2)!}{(n-a)!} \frac{(b/2)!(a-b/2)!}{a!} \leq K \\
                &\iff \left(\frac{(b/2)!}{(\frac{b-1}2)!}\right)^2  \frac{(a-1)!}{a!}\frac{(n-a-b/2)!}{(n-a-\frac{b-1}2)!} \frac{(a-b/2)!}{(a-1-\frac{b-1}2)!} \leq K \\
                &\Longleftarrow \left(\frac{4}{e} \sqrt{\frac{b-1}2}\right)^2 \frac1a \frac{1}{\sqrt{n-a-\frac{b-1}2}} \frac{1}{\sqrt{a-1-\frac{b-1}2}} \leq K \\
                &\iff \frac{16}{e^2K} \frac{b-1}2 \leq a \sqrt{n-a-\frac{b-1}2} \sqrt{a-1-\frac{b-1}2} \leq a \ .
            \end{align*}
            Looking at the binomials tells us that we must have $\frac{b-1}2 \leq a-1$. Up to the constant factors, the above tells us that $\frac{b-1}2 \leq a$. Thus, so long as the constant is $\leq 1$, the above is true. That means we must have $\frac{16}{e^2K} \leq 1$, which means that $\frac{16}{e^2} \leq K$, so that it suffices to set $K=2.17$. 

            \item This is analogous to the previous claim, as the numerator in $p_A'(a,b)$ and $p_D'(a,b)$ are the same, just with $a$ and $b$ swapped. It follows that the previous result carries over to $p_A'$.

            \item We consider the ratio $p_A'(a,b)/p_A'(a,b+1)$, which is
            \[
                \frac{\frac{\binom{n-(b+1)}{a/2}\binom{b+1}{a/2}}{\binom{n}{a}}}{\frac{\binom{n-b}{a/2}\binom{b}{a/2}}{\binom{n}{a}}} = \frac{\binom{n-b-1}{a/2}}{\binom{n-b}{a/2}}\frac{\binom{b+1}{a/2}}{\binom{b}{a/2}} \ .
            \]
            We compute
            \[
                \frac{\binom{n-b-1}{a/2}}{\binom{n-b}{a/2}} = \frac{\frac{(n-b-1)!}{(a/2)!(n-b-1-a/2)!}}{\frac{(n-b)!}{(a/2)!(n-b-a/2)!}} = \frac{n-b-a/2}{n-b}
            \]
            and
            \[
                \frac{\binom{b+1}{a/2}}{\binom{b}{a/2}} = \frac{\frac{(b+1)!}{(a/2)!(b+1-a/2)!}}{\frac{b!}{(a/2)!(b-a/2)!}} = \frac{b+1}{b+1-a/2} \ .
            \]
            Thus, 
            \[
                \frac{p_A'(a,b)}{p_A'(a,b+1)} =\frac{n-b-a/2}{n-b} \cdot \frac{b+1}{b+1-a/2} \ .
            \]
            Now, we observe that this ratio is $> 1$ iff
            \begin{align*}
                (n-b-a/2)(b+1) > (n-b)(b+1-a/2) &\iff (n-b)(b+1) - (a/2)(b+1) > (n-b)(b+1)-(n-b)(a/2) \\
                &\iff 0 > (a/2)(2b+1-n) \ .
            \end{align*}
            We see this inequality is satisfied iff $2b+1< n$. That is, the ratio is $> 1$ for all $b \in \{1,\dots,n/2-1\}$, and then is $< 1$ for all $b \in \{n/2,\dots,n-1\}$, as required. 

            \item To show that $p_A'(a,b)$ is maximized over $a$ by some $a \leq n/2$, we will show that $p_A'(a,b) \geq p_A'(n-a,b)$ whenever both are defined. Note that this occurs when $(n-a)/2\leq b \leq (n+a)/2$. Now, we note here that the domain of the latter is in general smaller than the domain of the former. In particular, $p_A'(a,b)$ will admit values $a \leq n-2b$, while $p_A(n-a,b)$ does not. But of course, if we show $p_A(a,b)$ to be bigger than $p_A(n-a,b)$ across the shared domain, then any additional domain for the former function could only increase its maximizing value further. Proving this claim thus suffices. Let us then consider the ratio $\frac{p_A'(a,b)}{p_A'(n-a,b)}$, and look to establish that it is $\ge 1$. This suffices to establish the claim.  

            $$\frac{p_A'(a,b)}{p_A'(n-a,b)} 
            = \frac{\frac{\binom{n-b}{a/2}\binom{b}{a/2}}{\binom{n}{a}}}{\frac{\binom{n-b}{\frac{n-a}{2}}\binom{b}{\frac{n-a}{2}}}{\binom{n}{a}}} 
            = \frac{\binom{n}{a}}{\binom{n}{a}} \frac{\binom{n-b}{a/2}}{\binom{n-b}{\frac{n-a}{2}}}\frac{\binom{b}{a/2}}{\binom{b}{\frac{n-a}{2}}} 
            =  \frac{\binom{n-b}{a/2}}{\binom{n-b}{\frac{n-a}{2}}}\frac{\binom{b}{a/2}}{\binom{b}{\frac{n-a}{2}}} \ .$$
            The first ratio becomes:
            $$\frac{\binom{n-b}{a/2}}{\binom{n-b}{\frac{n-a}{2}}} = \frac{(n-b-\frac{n-a}{2})!(\frac{n-a}{2})!}{(n-b-a/2)!(a/2)!} = \frac{(n/2-b +a/2)!(\frac{n-a}{2})!}{(n-b-a/2)!(a/2)!}$$
            The second ratio becomes:
            $$\frac{\binom{b}{a/2}}{\binom{b}{\frac{n-a}{2}}} = \frac{(b-\frac{n-a}{2})!(\frac{n-a}{2})!}{(b-a/2)!(a/2)!} = \frac{(b-n/2+a/2)!(\frac{n-a}{2})!}{(b-a/2)!(a/2)!}$$

            So the product of the two terms 
            \[
                \frac{((\frac{n-a}2)!)^2\cdot(n/2-b+a/2)!\cdot (b+a/2-n/2)!}{((a/2)!)^2 \cdot (n-b-a/2)!(b-a/2)!} \ .
            \]
            We first note this value is $1$ if $b=n/2$. We now consider the above as a function of $b$ -- call it $s(b)$ -- and show that the ratio increases if we move $b$ away from $n/2$. This suffices to establish the claim. We have 
            \begin{align*}
                \frac{s(b+1)}{s(b)} 
                &= \frac{\frac{((\frac{n-a}2)!)^2}{((a/2)!)^2}}{\frac{((\frac{n-a}2)!)^2}{((a/2)!)^2}} \frac{\frac{(n/2-(b+1)+a/2)!\cdot ((b+1)+a/2-n/2)!}{(n-(b+1)-a/2)!((b+1)-a/2)!}}{\frac{(n/2-b+a/2)!\cdot (b+a/2-n/2)!}{(((n-b-a/2)!(b-a/2)!}} \\
                &=\frac{(n/2-(b+1)+a/2)!}{(n/2-b+a/2)!} \frac{((b+1)+a/2-n/2)!}{(b+a/2-n/2)!} \frac{(n-b-a/2)!}{(n-(b+1)-a/2)!} \frac{(b-a/2)!}{(b+1)-a/2)!}  \\
                &= \frac{(b+1+a/2-n/2)(n-b-a/2)}{(n/2-b+a/2)(b+1-a/2)} \ .
            \end{align*}
            We will establish the above is $\leq 1$ iff $b<n/2$. This is the same as saying
            $$(b+1+a/2-n/2)(n-b-a/2) \geq (n/2-b+a/2)(b+1-a/2)$$
            which is the same as 
            \begin{align*}
                &\qquad\quad(b+1-n/2+a/2)(n-b-a/2) \geq (n/2-b+a/2)(b+1-a/2) \\
                &\iff (b+1-n/2)(n-b)+(a/2)(n-b) - (b+1-n/2)(a/2) - a^2/4 \\
                &\qquad\qquad\qquad\qquad\geq (n/2-b)(b+1) +(a/2)(b+1)-(a/2)(n/2-b)-a^2/4\\
                &\iff (b+1-n/2)(n-b)-(n/2-b)(b+1) \geq (a/2)(b+1-n/2+b+1-(n/2-b)-(n-b)) \\
                &\iff \frac{n}{2}(2b+1-n) \geq a(2b+1-n) \ .
            \end{align*}
            So, if $b\geq n/2$ the inequality holds as $a \leq \frac{n}{2}$, and the reverse holds when $b\leq n/2$. 
            
            \item We start by noting that $p_A(a,b) \leq  \left(\frac{n-b}{n}\right)^{\lfloor a/2 \rfloor} \left(\frac{b}{n}\right)^{\lceil a/2 \rceil} 2^a \frac{\lceil a/2\rceil}{b}$, which was proven as Corollary 3.4.2 in \cite{pfister2003capacity}. Note that we wish to upper bound $p_A'$, not $p_A$. There are two differences between these. First, the $-1$'s in $p_A$ give rise to the factor $\frac{\lceil a/2\rceil}{b}$, so we should remove this factor. Second, the floors and ceilings are gone in $p_A'$, so we should remove those. Note also that $\left(\frac{n-b}{n}\right)^{\lfloor a/2 \rfloor} \leq 1$, so that we can remove it. Removing these three things gives us the final bound: $p_A'(a,b) \leq \left(\frac{b}{n}\right)^{a/2} 2^a = \left(\frac{2\sqrt{b}}{\sqrt{n}}\right)^{a}$.
            
            \item Recall that $p_D' = \frac{\binom{n-a}{b/2}\binom{a}{b/2}}{\binom{n}{a}}$. To get the final bound, we rewrite the numerator using \Cref{lem:A_D_claims}.1: 
            $$\binom{n-a}{b/2}\binom{a}{\lceil b/2 \rceil} 
            \leq \frac{\left(\frac{en}{b/2}\right)^{b/2} \left(\frac{ea}{b/2}\right)^{b/2}}{2\pi \sqrt{b/2 \cdot b/2}} = \frac{1}{\pi b}\left(\frac{2e\sqrt{an}}{b}\right)^{b}\ .$$
        \end{enumerate}
    \end{proof}

    \wwh*
    \begin{proof}
        Firstly, note that the function $b \mapsto p_A'(a,b) p_D'(b,c)$ is symmetric about $n/2$. Hence, we restrict attention to $b \leq n/2$. We wish to show the function $b \mapsto p_A'(a,b)p_D'(b,c)$ is decreasing in $b$ on the domain $\{\max\{a/2,c/2\},\dots,n/2\}$. For $\max\{a/2,c/2\} \leq b\leq n/2-1$, we consider the ratio
        \[
            \frac{p_A'(a,b+1) \cdot p_D'(b+1,c)}{p_A'(a,b) \cdot p_D'(b,c)}
        \]
        and establish that it is less than $1$. By comparing the binomial coefficients appearing in these probabilities, one finds this ratio to be 
        \begin{align}
            \frac{(b+1)^3(n-b-a/2)(n-b-c/2)}{(n-b)^3(b+1-a/2)(b+1-c/2)} \ . \label{eq:ratio-of-interest}
        \end{align}
        We now set $x = b+1$ and $y=n-b$, so we have $x < y$. Continuing from the above:
        \begin{align*}
            1 \geq \frac{x^3(y-a/2)(y-c/2)}{y^3(x-a/2)(x-c/2)} &= \frac{x^3y^2\left(1-\frac{a+c}{2y}+\frac{ac}{4y^2}\right)}{y^3x^2\left(1-\frac{a+c}{2x}+\frac{ac}{4x^2}\right)} = \frac{x\left(1-\frac{a+c}{2y}+\frac{ac}{4y^2}\right)}{y\left(1-\frac{a+c}{2x}+\frac{ac}{4x^2}\right)} = \frac{x-\frac{x(a+c)}{2y}+\frac{xac}{4y}}{y-\frac{y(a+c)}{2x}+\frac{yac}{4x}} \ .
        \end{align*}
        We rearrange the above and continue:
        \begin{align*}
            & x - \frac{x(a+c)}{2y}+\frac{xac}{4y} \leq y-\frac{y(a+c)}{2x}+\frac{yac}{4x} \\
            \iff & 0 \leq (y-x) - \left(\frac yx - \frac xy\right)\frac{a+c}2 + \left(\frac{y}{x^2}-\frac{x}{y^2}\right)\frac{ac}4\\
            \iff & 0 \leq (y-x) - \frac{y^2-x^2}{yx}\cdot \frac{a+c}2 + \frac{y^3-x^3}{y^2x^2}\cdot\frac{ac}{4} \\
            \iff & 0 \leq (y-x)\left(1 - \frac{y+x}{yx}\cdot\frac{a+c}{2} + \frac{y^2+xy+x^2}{y^2x^2}\cdot\frac{ac}4\right) \ .
        \end{align*}
        Now, by dividing by $y-x$ (which is positive if and only if $b < n/2$), we find the above inequality is valid if and only if
        \begin{align}
            \iff & 0 \leq 1 - \frac{y+x}{yx}\cdot\frac{a+c}{2} + \frac{y^2+xy+x^2}{y^2x^2}\cdot\frac{ac}4 \nonumber \\
            \iff & 0 \leq \frac{4y^2x^2-2xy(y+x)(a+c)+(y^2+xy+x^2)ac}{4y^2x^2}\nonumber \\
            \iff & 0 \leq 4y^2x^2-2xy(y+x)(a+c)+(y^2+xy+x^2)ac \nonumber \\
            \iff & 2xy(y+x)(a+c) \leq 4y^2x^2 +(y^2+xy+x^2)ac \ . \label{eq:next-ineq}
        \end{align}
        We now turn to establishing this last inequality. Recall $x = b+1 > \max\{a/2,c/2\}$; noting that $a$ and $c$ appear identically in the above expression, without loss of generality we may assume $a \leq c$. Let us assume first that $c \leq 0.9x$. Thus, $a+c \leq 2c \leq 1.8x$, and so $2(a+c) \leq 3.6 x$. Thus, it follows that 
        \[
            2xy^2(a+c) \leq 3.6y^2x^2 \ .
        \]
        Now, recalling $\max\{a, c\} = o(n)$ it follows that $\frac2{0.4}(a+c) \leq y$ for large enough $n$ (recall $y = n-b > n/2$), and so 
        \[
            2x^2y(a+c) \leq 0.4y^2x^2 \ .
        \]
        This establishes this last inequality when $x$ is large enough. We now turn to the case when $c \geq 0.9x$. Recall that we still have $c \leq 2x$, and so $2(a+c) \leq 4x$. Hence, the LHS is at most
        \[
            4x^2y(x+y) = 4x^3y + 4x^2y^2 \ .
        \]
        Additionally, by assumption $ac \geq 0.9 \cdot 0.5x^2 = 0.45 x^2$. Hence, the RHS is at least 
        \[
            4y^2x^2 + (y^2+xy+x^2) \cdot 0.45 x^2 = 4x^2y^2 + 0.45 x^2y^2 + 0.45 x^3y + 0.45 x^4 \ .
        \]
        Removing the $4y^2x^2$ terms and dropping the $ 0.45 x^3y + 0.45 x^4$, it suffices to establish that
        \[
            4x^3y \leq 0.45 x^2y^2 \iff y \geq \frac4{0.45}x \ .
        \]
        Recalling $x \leq c/0.9$ and $c = o(n)$ (and hence $o(y)$), it follows that for large enough $n$ (and hence $y$) this last inequality holds. 
    \end{proof}

    \wwnh*
    \begin{proof}
        Let $x$ and $y$ be as in the previous proof, and note that we now wish to show the ratio 
        \[
            \frac{x-\frac{x(a+c)}{2y}+\frac{xac}{4y}}{y-\frac{y(a+c)}{2x}+\frac{yac}{4x}} = \begin{cases}
                < 1 & x < y \text{, i.e., } b<n/2\\
                > 1 & x > y \text{, i.e., } b\geq n/2
            \end{cases} \ .
        \]
        To establish this, following the series of inequalities from the previous proof, and recalling that dividing by $y-x$ requires us to change the direction of the inequality if $y<x$, we find that we what we want to show is equivalent to showing
        \begin{align}
            2xy(x+y)(a+c) > 4y^2x^2+(y^2+xy+x^2)ac \iff 2(x+y)(a+c) > 4xy + \left(\frac{x^2+y^2}{xy}+1\right)ac\ . \label{eq:ineq-now-reversed}
        \end{align}
        We consider the case that $a = n-o(n)$ and $c = o(n)$ (the case where $a=o(n)$ and $c=n-o(n)$ can be handled similarly). Write $a=n-h_1$ and $c=h_2$ where $\max\{h_1,h_2\} = o(n)$. Since $\max\{a/2,c/2\} \leq b \leq n-\max\{a/2,c/2\}-1$, it follows that $\frac{n-h_1}{2} \leq b \leq \frac{n+h_1}{2}-1$. Thus, $\frac{n-h_1}{2}+1 \leq x \leq \frac{n+h_1}{2}$ and $\frac{n-h_1}{2}+1\leq y \leq \frac{n-h_1}{2}$. 
    
        Now, the expression $\frac{x^2+y^2}{xy} = \frac xy + \frac yx$, subject to the constraint that $x+y=n+1$, is maximized when $|x-y|$ is as large as possible. Indeed, WLOG $x>y$ and making the substitution $t=\frac{x}{y}$, we consider the function $f(t) = t+1/t$, which has derivative $1-1/t^2$, which is positive for $t \geq 1$. That is, we should make $t = x/y$ as large as possible, which is the same as making $x-y$ as large as possible (subject to $x+y=n+1$). Thus, it is maximized when $x=\frac{n+h_1}{2}$ and $y=\frac{n-h_1}{2}+1 = \frac{n-h_1+2}{2}$ (say), in which case we have 
        \[
            \frac xy + \frac yx = \frac{n+h_1}{n-h_1+2} + \frac{n-h_1+2}{n+h_1} \leq \frac{n+h_1}{n-h_1+2} + 1 \leq 1 + \frac{2h_1}{n-h_1+2} + 1 = 2 + \frac{2h_1}{n-h_1+2} \ .
        \]
        Hence, to establish \eqref{eq:ineq-now-reversed}, it suffices to argue that 
        \begin{align}
            2(x+y)(a+c) > 4xy + \left(3 + \frac{2h_1}{n-h_1+2}\right)ac \ . \label{eq:ineq-now-reversed-implying}
        \end{align}
        To do this, first observe that the LHS is $2(n+1)(n-h_1+h_2) = 2 n^2\pm$ lower order terms. For the RHS, subject to $x+y=n+1$ to maximize $xy$ we should take $x=y = \frac{n+1}{2}$, and so the RHS is at most
        \[
            4 \cdot \left(\frac{n+1}{2}\right)^2 + \left(3 + \frac{2h_1}{n-h_1+2}\right)(n-h_1)h_2 = n^2 \pm \text{lower order terms}
        \]
        Hence, for large enough $n$ the desired inequality \eqref{eq:ineq-now-reversed-implying} holds. 
    \end{proof}

\section{Additional lemmas}

In this section, we state and prove some lemmas that were omitted from the main text. 

\begin{lemma} \label{lem:max-beta=2gamma(1-gamma)}
    For fixed $\gamma \in (0,1)$, the function $g:(0,\min\{2\gamma,2(1-\gamma)\}]\to\R$ defined by $g(\beta)=h(\beta)-h(\alpha)+f_A(\beta,\gamma)$ is maximized at $\beta=2\gamma(1-\gamma)$. 
\end{lemma}

\begin{proof}
    We find
    \[
        g'(\beta) = \frac{1}{2}\log\left(\frac{4(1-\gamma-\beta/2)(\gamma-\beta/2)}{\beta^2}\right) \ ,
    \]
    which equals $0$ iff $\beta^2 = 4(1-\gamma-\beta/2)(\gamma-\beta/2) = 4(1-\gamma)\gamma-2\beta + \beta^2$, i.e., iff $\beta=2\gamma(1-\gamma)$. To see this critical point gives a global maximum, we consider the second derivative
    \[
        g''(\beta) = \frac{1}{\beta\ln2}\cdot\frac{\beta-4(1-\gamma)\gamma}{(2\gamma-\beta)(2(1-\gamma)-\beta)} \ ,
    \]
    and argue it is $\leq 0$ (i.e., the function is concave on its domain). Note that always $2\gamma\geq\beta$, so $(2\gamma-\beta)\geq 0$ and $2(1-\gamma)\geq \beta$. So it suffices to argue $\beta \leq 4(1-\gamma)\gamma$. Without loss of generality, $\gamma\leq1/2$, so then $\beta \leq 2\gamma\leq 4\gamma(1-\gamma)$ as $2(1-\gamma)\geq 1$. 
\end{proof}

\begin{lemma} \label{lem:max-beta=1/2}
    For fixed $\alpha \in [1/2,1)$, the function $g:[\alpha/2,1-\alpha/2]\to\R$ defined by $g(\beta)=f_A(\alpha,\beta)-h(\beta)$ is maximized at $\beta=1/2$.
\end{lemma}

\begin{proof}
    We find 
    \[
        g'(\beta) = -2\log\frac{1-\beta}{\beta} + 
        \log\frac{1-\beta-\alpha/2}{\beta-\alpha/2} \ ,
    \]
    which by inspection evaluates to $0$ when $\beta=1/2$. To see this is a global maximum, we compute the second derivative
    \[
        g''(\beta)=\frac{1}{\ln 2}\left(\frac2{\beta(1-\beta)}-\frac{1-\alpha}{(\beta-\alpha/2)(1-\beta-\alpha/2)}\right) = \frac{1}{\ln2}\frac{2(\beta-\alpha/2)(1-\beta-\alpha/2)-(1-\alpha)\beta(1-\beta)}{\beta(1-\beta)(\beta-\alpha/2)(1-\beta-\alpha/2)} \ ,
    \]
    and show it is negative. As $\alpha/2\leq\beta\leq1-\alpha/2$, the denominator is positive, so it suffices to show the numerator is negative, i.e., that $(1-\alpha)\beta(1-\beta) \geq 2(\beta-\alpha/2)(1-\beta-\alpha/2)$. 

    To establish this, we again use calculus. Consider now the function $s(\beta)= 2(\beta-\alpha/2)(1-\beta-\alpha/2)-(1-\alpha)\beta(1-\beta)$. We have $s'(\beta)=1+\alpha-2(1-\alpha)\beta$, so we get a critical point of $\beta = \frac{1+\alpha}{2(1-\alpha)}$, which is manifestly the maximum for the function as $s''(\beta)=-2(1-\alpha)<0$. Thus, it suffices to show $s\left(\frac{1+\alpha}{2(1-\alpha)}\right)\leq 0$. Plugging in this value and simplifying, we get 
    \[
        \frac{1-5\alpha+5\alpha^2-11\alpha^3+2\alpha^4}{4(1-\alpha)^2} \ .
    \]
    It suffices to show the numerator is $\leq 0$. 
    
    Let $t(\alpha) = 1-5\alpha+5\alpha^2-11\alpha^3+2\alpha^4$. One can evaluate $t(1/2) = -3/2$. We now establish $t(\alpha)$ is decreasing on the range $[1/2,1]$, i.e., that $t'(\alpha) = -5+10\alpha-33\alpha^2+8\alpha^3<0$. To do this, we again first observe that $t'(1/2)=-29/4$, and then look to show $t''(\alpha) = 10-66\alpha+24\alpha^2\leq 0$. To show $t''(\alpha)\leq 0$, we again first compute $t''(1/2) = -17$, and then consider $t'''(\alpha) = -66+48\alpha$. For $\alpha \in [1/2,1]$, we certainly have $t'''(\alpha)\leq 0$, as desired. 
\end{proof}

\begin{lemma} \label{lem:max-beta=sqrt-stuff}
    For fixed $\alpha \in (0,1/2]$, the function $g:[\alpha/2,1/2]\to\R$ defined by $g(\beta)=f_A(\alpha,\beta)-h(\beta)$ is maximized at $\beta = (1-\sqrt{1-2\alpha})/2$. That is, it increases from $\alpha/2$ to $(1-\sqrt{1-2\alpha})/2$ and then decreases to $1/2$. 
\end{lemma}

\begin{proof}
    As before,
    \[
        g'(\beta) = -2\log\frac{1-\beta}{\beta} + 
        \log\frac{1-\beta-\alpha/2}{\beta-\alpha/2} = \log\frac{(1-\beta-\alpha/2)\beta^2}{(\beta-\alpha/2)(1-\beta)^2} \ ,
    \]
    which when we solve for critical points beyond the value $\beta=1/2$ found previously we also now have solutions $\beta = (1\pm \sqrt{1-2\alpha})/2$ (which we observe are now real since $0\leq\alpha \leq 1/2$). We first assume $\alpha<1/2$, and later deal with the case of $\alpha=1/2$. To establish that on the domain $[\alpha/2,1/2]$ $g$ has a global maximum at $\beta^* = (1-\sqrt{1-2\alpha})/2$, it suffices to find points $\alpha/2<\beta_-<\beta^*$ and $1/2>\beta_+>\beta^*$ such that $g'(\beta_-)>0$ and $g'(\beta_+)<0$, as then it must be that $g$ increases from $\alpha/2$ to $\beta^*$ and then decreases from $\beta^*$ to $1/2$. 

    For $\beta_-$, we take $\alpha/2+\eps'$ for a small $\eps'>0$. Note that 
    \[
        g'(\alpha/2+\eps') = \log\frac{(1-\alpha-\eps)(\alpha/2+\eps)^2}{\eps^2(1-\alpha/2-\eps)^2} \ ;
    \]
    for sufficiently small $\eps$, we clearly have the denominator smaller than the numerator, we will have $g'(\alpha/2+\eps')>0$, as required. 

    For $\beta_+$, we take $1/2-\tilde \eps$ for small $\tilde \eps>0$. We find 
    \[
        g'(1/2-\tilde \eps) = \log_2\frac{\left(\frac{1-\alpha}{2}+\tilde \eps\right)(1/2-\tilde \eps)^2}{\left(\frac{1-\alpha}{2}-\tilde \eps\right)(1/2+\tilde \eps)^2} \ .
    \]
    Let $x=\frac{1-\alpha}{2}$ and $y = \frac12$, and note that the assumption $\alpha < 1/2$ is equivalent to $y < 2x$. To show $g'(1/2-\tilde\eps)<0$, it suffices to show
    \[
        \frac{(x+\tilde\eps)(y-\tilde\eps)^2}{(x-\tilde\eps)(y+\tilde\eps)^2}<1 \ ,
    \]
    which, after rearranging and simplifying, is equivalent to 
    \[
        2y^2 + \tilde \eps^2 < 4xy \ .
    \]
    Since $y<2x$, $2y^2<4xy$, so for sufficiently small $\tilde\eps>0$ the above holds. 

    We now deal with the case of $\alpha=1/2$. We then observe that in fact all the critical points $1/2$ and $(1\pm \sqrt{1-2\alpha})/2$ are equal. To see this is a global maximum, we can reuse the earlier argument that at the point $\alpha/2+\eps'$ the derivative $g'$ is positive. 
\end{proof}
\end{document}